\PassOptionsToPackage{table}{xcolor}
\RequirePackage{tikz}%
\documentclass[twoside]{article}
\pdfoutput=1
\usepackage[utf8]{inputenc}
\usepackage[T1]{fontenc}
\usepackage{snpretty}
\usepackage[numbers,sort&compress]{natbib}%
\usepackage[american]{babel}%
\date{}
\usepackage{etoolbox}
\newtoggle{thmacm}
\newtoggle{thmnormalsize}
\togglefalse{thmacm}
\togglefalse{thmnormalsize}
\toggletrue{thmnormalsize}%
\usepackage{snjnlstuffthmarxiv}
\newcommand{\centerskip}{\medskip}
\newcommand{\beforelistskip}{\medskip}
\newcommand{\beforeproofpartskip}{\bigskip}
\newcommand{\afterproofpartskip}{\medskip}
\usepackage{caption}
\DeclareCaptionLabelFormat{caplab}{\textbf{#1 #2}}
\usepackage{comment}
\usepackage{xspace}
\usepackage{enumitem}
\usepackage{array}
\usetikzlibrary{matrix,snakes,arrows,shapes,positioning,arrows.meta,shapes.geometric}
\usepackage{pgfplots}
\pgfplotsset{compat=1.17}
\usepgfplotslibrary{fillbetween}
\usepackage{forest}
\forestset{
  nice empty nodes/.style={
    for tree={calign=fixed edge angles},
    delay={where content={}{shape=coordinate,for siblings={anchor=north}}{}}
  },
}
\usepackage{bm}
\usepackage{mleftright}
\usepackage{proof} %
\usepackage{plabels}
\newcommand{\name}[1]{\emph{#1}}
\newcommand{\defname}[1]{\emph{#1}}
\newcommand{\entails}{\models}

\newcommand{\imp}{\rightarrow}

\newcommand{\la}{\langle}
\newcommand{\ra}{\rangle}
\newcommand{\eqdef}{\; 
\raisebox{-0.1ex}[0mm]{$ \stackrel{\raisebox{-0.2ex}{\tiny 
      \textnormal{def}}}{=} $}\; }
\newcommand{\emptysubst}{\epsilon}
\newcommand{\nhphantom}[1]{\sbox0{#1}\hspace*{-\the\wd0}}

\renewcommand{\f}[1]{\mathsf{#1}}
\newcommand{\g}[1]{\mathit{#1}}
\mathchardef\hyph="2D

\newcolumntype{L}[1]{>{\raggedright\let\newline\\\arraybackslash\hspace{0pt}}p{#1}}
\newcolumntype{C}[1]{>{\centering\let\newline\\\arraybackslash\hspace{0pt}}p{#1}}
\newcolumntype{R}[1]{>{\raggedleft\let\newline\\\arraybackslash\hspace{0pt}}p{#1}}
\newcolumntype{X}[1]{>{\raggedright\let\newline\\\arraybackslash\hspace{0pt}$}p{#1}<$}
\newcolumntype{Y}[1]{>{\centering\let\newline\\\arraybackslash\hspace{0pt}$}p{#1}<$}
\newcolumntype{Z}[1]{>{\raggedleft\let\newline\\\arraybackslash\hspace{0pt}$}p{#1}<$}
\newenvironment{arrayprf}
{\begin{array}{Z{2.5em}@{\hspace{1em}}X{32em}}}
{\end{array}}
\newenvironment{arrayprfeq}
{\begin{array}{Z{2.5em}@{\hspace{1em}}Y{1.5em}X{30.5em}}}
{\end{array}}
\newcommand{\supterm}{\mathrel{\rhd}}
\newcommand{\suptermq}{\mathrel{\unrhd}}
\newcommand{\notsuptermq}{\not \mathrel{\unrhd}}
\newcommand\subsumedBy{\mathrel{\ooalign{$\geq$\cr
      \hidewidth\raise.225ex\hbox{$\cdot\mkern7.0mu$}\cr}}}
\newcommand\subsumes{\mathrel{\ooalign{$\leq$\cr
      \hidewidth\raise.225ex\hbox{$\cdot\mkern2.0mu$}\cr}}}
\newcommand\strictlySubsumes{\mathrel{\ooalign{$<$\cr
      \hidewidth\raise.0ex\hbox{$\cdot\mkern2.0mu$}\cr}}}
\newcommand\strictlySubsumedBy{\mathrel{\ooalign{$>$\cr
       \hidewidth\raise.0ex\hbox{$\cdot\mkern7.0mu$}\cr}}}
\newcommand\variant{\mathrel{\ooalign{$=$\cr
      \hidewidth\raise.7ex\hbox{$\cdot\mkern4.5mu$}\cr}}}

\newcommand{\Paragraph}{§}
\renewcommand{\P}{\f{P}}
\renewcommand{\i}{\f{i}}
\newcommand{\D}{\f{D}}

\newcommand{\Dterms}{\mbox{D-terms}\xspace}
\newcommand{\Dterm}{\mbox{D-term}\xspace}
\newcommand{\DTerms}{\mbox{D-Terms}\xspace}
\newcommand{\DTerm}{\mbox{D-Term}\xspace}
\newcommand{\dterms}{\mbox{D-terms}\xspace}
\newcommand{\dterm}{\mbox{D-term}\xspace}

\newcommand{\DSUBQ}{\m{\mathcal{S}{\hspace{-0.12em}ubeq}}}
\newcommand{\DSUBSTRICT}{\m{\mathcal{S}{\hspace{-0.12em}ub}}}

\newcommand{\DPRIMSET}{\m{\mathcal{DP}\hspace{-0.12em}rim}}
\newcommand{\DPRIM}{\m{\mathcal{DP}\hspace{-0.12em}rim}}
\newcommand{\dconstant}{primitive D-term\xspace}

\newcommand{\CompD}{\m{\mathcal{C}\hspace{-0.12em}omp\mathcal{D}}}

\newcommand{\SCsize}{SC size\xspace}

\newcommand{\compDterms}{compacted \mbox{D-terms}\xspace}
\newcommand{\compDterm}{compacted \mbox{D-term}\xspace}
\newcommand{\compd}{\delta}

\newcommand{\CD}{CD\xspace}

\newcommand{\tsize}{\f{t\hyph size}}
\newcommand{\csize}{\f{c\hyph size}}
\newcommand{\height}{\f{height}}
\newcommand{\scsize}{\f{sc\hyph size}}

\newcommand{\replace}[3]{#1{[}#2 \mapsto #3{]}}

\newcommand{\UCN}[1]{\forall #1}
\newcommand{\UCM}[1]{\forall\mspace{2mu} #1}
\newcommand{\UC}[1]{\forall\, #1}

\newcommand{\gtrc}{\mathrel{>_{\mathrm{c}}}}
\newcommand{\geqc}{\mathrel{\geq_{\mathrm{c}}}}

\newcommand{\m}[1]{\mathit{#1}}

\newcommand{\vars}{\m{\mathcal{V}\hspace{-0.11em}ar}}
\newcommand{\dom}{\m{\mathcal{D}\hspace{-0.08em}om}}
\newcommand{\rng}{\m{\mathcal{R}\hspace{-0.02em}ng}}
\newcommand{\vrng}{\m{\mathcal{VR}\hspace{-0.02em}ng}}

\newcommand{\pos}{\m{\mathcal{P}\!os}}
\newcommand{\leafpos}{\m{\mathcal{L}\hspace{-0.05em}eaf\mathcal{P}\!os}}
\newcommand{\innerpos}{\m{\mathcal{I}\hspace{-0.05em}nner\mathcal{P}\!os}}
\newcommand{\posvar}{\m{\mathcal{P}\!os}\mathcal{V}ar}

\newcommand{\positionalvars}{positional variables\xspace}

\newcommand{\mgu}{\f{mgu}}
\newcommand{\emptypos}{\epsilon}

\newcommand{\pairing}{\f{pairing}}

\newcommand{\sshift}[1]{\f{shift}_{#1}}

\newcommand{\mf}[1]{\mathit{#1}}
\newcommand{\mgt}{\mf{Mgt}}
\newcommand{\ipt}{\mf{Ipt}}
\newcommand{\MGT}{MGT\xspace}
\newcommand{\IPT}{IPT\xspace}

\newcommand{\IPTs}{IPTs\xspace}

\newcommand{\expand}{\f{expand}}

\newcommand{\Det}{\text{\textit{Det}}\xspace}
\newcommand{\Luk}{\text{\textit{{\L}ukasiewicz}}\xspace}
\newcommand{\Syll}{\text{\textit{Syll}}\xspace}
\newcommand{\Peirce}{\text{\textit{Peirce}}\xspace}
\newcommand{\Simp}{\text{\textit{Simp}}\xspace}

\newcommand{\comb}[1]{\bm{\mathsf{#1}}}

\newcommand{\PSYLL}{\text{\textit{{\L}DS}}\xspace}

\renewcommand{\IF}{\textbf{IF}\xspace}

\newcommand{\xeqdef}{\eqdef} %
\newcommand{\eeqdef}{\eqdef} %

\newcommand{\vy}{y}
\newcommand{\vx}[2]{x^{#1}_{#2}}

\newcommand{\sall}{\sigma}
\newcommand{\ssubold}{\tau}
\newcommand{\ssubnew}{\tau^{\prime}}

\newcommand{\NT}{\overline{T}}

\newcommand{\simpn}{\f{simp\hyph n}}

\newcommand{\Lukasiewicz}{{\L}u\-ka\-sie\-wicz\xspace}

\newcommand{\TPTP}{\name{TPTP}\xspace}
\newcommand{\ProverN}{\name{Prover9}\xspace}

\newcommand{\CDTools}{\name{CD~Tools}\xspace}
\newcommand{\E}{\name{E}\xspace}
\newcommand{\Vampire}{\name{Vampire}\xspace}
\newcommand{\leancop}{\name{leanCoP}\xspace}

\newcommand{\OEISNUM}[1]{\href{https://oeis.org/#1}{oeis:#1}}
\newcommand{\TP}[1]{\mbox{\name{#1}}\xspace}
\newcommand{\SETHEO}{\textit{SETHEO}\xspace}
\newcommand{\Otter}{\textit{OTTER}\xspace}
\newcommand{\PTTP}{\textit{PTTP}\xspace}
\newcommand{\SWIProlog}{\textit{SWI Prolog}\xspace}
\newcommand{\SGCD}{\textit{SGCD}\xspace}
\newcommand{\CCS}{\textit{CCS}\xspace}

\newcommand{\tabyes}{\textbullet}
\newcommand{\tabno}{--}

\newcommand{\psplevel}{\m{\mathcal{PSPL}\hspace{-0.05em}evel}}
\newcommand{\primelevel}{\m{\mathcal{P}\hspace{-0.05em}rime\mathcal{L}\hspace{-0.05em}evel}}

\newcommand{\trs}[1]{\hspace{#1em}}
\newcommand{\trsgray}[2]{\colorbox{gray!30}{\parbox{#1em}{$#2$\rule[-0.8ex]{0pt}{3.2ex}}}}
\newcommand{\trsbox}[2]{\fbox{\parbox{#1em}{$#2$\rule[-0.8ex]{0pt}{3.2ex}}}}
\newcommand{\trslash}{\hspace{1em}/\hspace{1em}}
\newcommand{\trskip}[1]{\ldots}

\newcommand{\IS}{IS\xspace}
\newcommand{\MS}{MS\xspace}
\newcommand{\XS}{S\xspace}
\newcommand{\MC}{MC\xspace}
\newcommand{\XC}{C\xspace}

\newcommand{\mereq}{\hspace{0.15em}=\hspace{0.15em}}
\newcommand{\n}{\mathrm{n}}
\newcommand{\mer}[1]{\textit{M#1}}
\newcommand{\subsumesmer}[1]{\textit{$\strictlySubsumes$M#1}}
\newcommand{\luk}[1]{\textit{{\L}#1}}
\newcommand{\subsumesluk}[1]{\textit{$\strictlySubsumes${\L}#1}}
\newcommand{\subsumedByluk}[1]{\textit{$\strictlySubsumedBy${\L}#1}}

\newcommand{\spec}[1]{\textit{{N}#1}}

\newcommand{\pr}{$\,'$}

\newcommand{\xconn}[1]{{\tiny \textbf{#1}}}

\renewcommand{\tl}{1.7cm}
\newcommand{\tll}{0.85cm}
\newcommand{\tlll}{0.425}
\newcommand{\tllll}{0.25cm}

\newcommand{\DLUK}{D_{\textsf{{\L}UK}}}
\newcommand{\DMER}{D_{\textsf{MER}}}

\newcommand{\DSHORT}{D_{\textsf{29}}}

\newcommand{\tabatt}[1]{\scalebox{0.8}{\textbf{#1}}}

\newcommand{\xatt}[1]{\textnormal{\textbf{#1}}}

\newcommand{\LUK}{{\L}UK}

\newcommand{\iv}[2]{{[\![#1,\!#2]\!]}}

\newcommand{\thes}[1]{\name{Thesis-#1}}
\newcommand{\attitem}[1]{\subsubsection{#1}}

\newcommand{\tp}[1]{\name{#1}}

\title[Investigations into Proof Structures]{Investigations into Proof
Structures}

\author*[1]{\fnm{Christoph} \sur{Wernhard}}\email{info@christophwernhard.com}
\author*[2]{\fnm{Wolfgang} \sur{Bibel}}\email{bibel@gmx.net}

\affil*[1]{\orgname{University of Potsdam}, \country{Germany}}
\affil[2]{\orgname{Technical University Darmstadt}, \country{Germany}}

\abstract{\noindent We introduce and elaborate a novel formalism for the manipulation
  and analysis of proofs as objects in a global manner. In this first approach
  the formalism is restricted to first-order problems characterized by
  condensed detachment. It is applied in an exemplary manner to a coherent and
  comprehensive formal reconstruction and analysis of historical proofs of a
  widely-studied problem due to \Lukasiewicz. The underlying approach opens
  the door towards new systematic ways of generating lemmas in the course of
  proof search to the effects of reducing the search effort and finding
  shorter proofs. Among the numerous reported experiments along this line, a
  proof of \Lukasiewicz's problem was automatically discovered that is much
  shorter than any proof found before by man or machine.}


\begin{document}

\maketitle\showabstract\newpage%
\setcounter{tocdepth}{3}%
\fancyhead[CO]{\small Contents}
\fancyhead[CE]{\small Contents}
\tableofcontents\newpage
\fancyhead[CO]{\small \nouppercase{\rightmark}}
\fancyhead[CE]{\small \nouppercase{\leftmark}}

\section{Introduction}
\label{sec-intro}

In Automated Theorem Proving (ATP) -- or Automated Deduction, or Automated
Reasoning -- the general research topic consists in the search for proofs of
formulas in order to establish their validity or theoremhood. We consider
proofs as syntactic objects defined on the basis of some formal system. There
is a variety of such formal proof systems; hence the formal objects
representing proofs in these differ widely and in consequence also the methods
for finding proofs.

In popular proof methods such as the resolution method, superposition, or the
tableau methods, proofs are sets of formulas arranged in a structured way.
This could be, for instance, in the form of a tree or graph with the formulas
-- or clauses -- labeling its nodes. Connections in the graph indicate that a
succeeding formula is derived from preceding formulas by some manipulation
such as forming the resolvent out of two clauses. Let us here refer to this
subclass of proof systems as \emph{formula-manipulative} ones.

From the point of view of proofs as a whole, formula manipulation of this kind
is a local operation. For both representing as well as finding proofs, more
global operations might be helpful. The use of lemmas may be regarded as such
a global operation. If a proof of a lemma is known, this proof may be inserted
into the overall proof wherever the formula representing the lemma occurs. For
formula-manipulative proof systems such a replacement operation is performed
implicitly by associating with an inferred formula pointers to the parent
formulas from which it was inferred. The proof structure as a whole is made
available in retrospect after a proof has been found, as a DAG (directed
acyclic graph) or as a tree formed by such pointers.

There are proof systems beyond the purely formula-manipulative ones. One such
system has been introduced and applied by Carew A.~Meredith, e.g., in a paper
from 1963 jointly authored with Arthur Prior \cite{meredith:notes:1963}. It
became known under the label of \name{condensed detachment (CD)}. A proof in
this system is represented as a list of pairs of a formula and a proof term.
The focus is on the proof-structural part represented as a term. The
formula-manipulative aspect is reduced to presenting intermediate lemmas.

The proof system underlying the connection method (CM) \cite{bibel:atp:1982}
is even more extreme in this sense. Proofs consist there exclusively of
structural information on the given formula without any manipulative part as
in formula-manipulative systems.

In ATP, CD was so far considered mainly as a special case of hyperresolution,
not taking into account its non-formula-manipulative characteristics. So far,
no adequate formal account of CD from the perspective of ATP could be found in
the literature.

The mutual advantages or disadvantages of these different kinds of formal
systems for proof search or proof representation are not at all clear at this
point in the development of ATP. Lemma-related techniques of general
importance for saturation-based provers such as the advanced use of weighting
templates, e.g., \cite{wos:resonance:95}, and hints \cite{veroff:hints:1996}
were initially devised for CD problems. There are several approaches to
integrating forms of lemma generation into variants of the CM
\cite{eder:cs:1989,astrachan:stickel:caching:1992,schumann:delta:1994,letz:cut:1994,fuchs:lemmas:ijcai:1999,leancop}. Nevertheless,
for the more structurally complicated systems such as CD or CM global
operations like the use of lemmas have never been studied systematically.

The work reported in this paper provides first results in exactly this
direction. Since a comparative analysis of different proof systems such as
those just mentioned is a truly complex enterprise, the task has to be
drastically restricted in this first approach. We thus focus on the simplest
nontrivial class of first-order formulas: a structurally simple goal statement
to be derived from an axiom and a rule with two premises and a single
conclusion. The obvious generalizations are deferred to future work: more than
one axiom, more and more complex rules, and so forth, up to arbitrary
first-order formulas.

Even under the drastic restriction just specified, our comparative task turns
out to be rather involved and proof search for this class of formulas is not
at all trivial for leading ATP systems. Global techniques for directing proof
search such as the use of lemmas or the replacement of proof parts appear to
be particularly intricate for systems that are not formula-manipulative.

The required extensive formal basis is worked out in this paper. Proofs are
represented as terms, which offers advantages not present in
formula-manipulative systems. Altogether, we open here the door towards a
better understanding of the distinctive features of known formal proof systems
with regard to their better or worse suitability for proof search, taking
first steps in this important direction.

Since CD falls into the considered class of first-order formulas, our work
includes the first comprehensive formalization of Meredith's proof system from
an ATP perspective, quasi as a side-result. At the same time this amounts to a
very detailed reconstruction of the historical proofs of a much-studied
problem first stated and proved by {\L}ukasiewicz. Our paper also gives a
rather comprehensive account of the work reported in the literature about this
well-known problem. This account includes numerous experimental results
achieved with a variety of systems. Incorporating the presented original insights,
one of our systems (\SGCD) discovered in a few seconds a new proof of
this problem, which is shorter than all previously known ones.

This work extends the results presented at CADE 2021 \cite{cwwb:lukas:2021}.
The concepts and techniques described here are backed by an implemented
system, \CDTools \cite{cw:cdtools:2022}, a library for experimenting with \CD
and related techniques, which is written in \SWIProlog \cite{swiprolog} and
available as free software. \CDTools includes two provers, \SGCD
\cite{cw:sgcd} (the name
suggesting \name{Structure-Generating proving for Condensed Detachment}) for
\CD problems, and \CCS \cite{cw:ccs} (the name suggesting \name{Compressed Combinatory
  Structures}) for \CD problems and first-order Horn problems. In the paper we
will discuss particular features of these provers and report experimental
results obtained with them. For more details on \SGCD and \CCS we refer to
\cite{cw:sgcd} and \cite{cw:ccs}, respectively.\footnote{The \CDTools system is
available from \url{http://cs.christophwernhard.com/cdtools}. Supplementary
material specific for the paper and to reproduce the experimental results is
provided at \url{http://cs.christophwernhard.com/cdtools/exp-investigations/}.}

The contributions of the paper can be summarized as follows.
\begin{enumerate}
\setlength{\parskip}{0pt}
\item A new formal characterization of CD with the proof structure as a whole
  in the focus, based on concepts and techniques known from the CM.
  
\item  New aspects concerning the interplay of tree and DAG structures in ATP.
  They relate the tree-oriented proceeding of clausal tableau methods with the
  DAG-oriented structure of CD and resolution proofs.

\item New regularity properties of proof structures and new criteria for
  shortening proofs by rewriting. Some of these are consequences of the
  interplay of tree and DAG structures.

\item Identifying and systematizing a set of ATP-relevant features of proofs
  on the basis of our formal framework.
  
\item A detailed analysis of a historic formal proof by Jan
  {\L}ukasiewicz and a variation by Meredith, from an ATP perspective, with
  respect to the identified proof features.
  
\item Generalizing specific structural features observed in the
  historic proofs to novel proof-structure-oriented techniques
  for proof search and lemma generation in ATP.

\item Providing the basis for an implemented system to experiment
  with CD problems and their proof structures \cite{cw:cdtools:2022,cw:sgcd,cw:ccs}. It
  includes two provers, each addressing a specific main aspect. One of
  them, \SGCD, realizes the newly discovered structure-oriented techniques.

\item A new short proof of \Lukasiewicz's problem, found by \SGCD with
   one of the new techniques. It is substantially shorter than the human-made
   proofs and drastically shorter than known proofs by first-order provers.
   Although the proofs by \ProverN \cite{prover9} can be substantially
   shortened with our new proof rewritings, they still remain drastically
   larger.

\item Foundation for follow-up work, including a novel approach to proof search
   over compressed combinatory structures \cite{cw:ccs} and studying the
   generation, selection and application of
   lemmas \cite{mrcwzzwb:lemmas:2023}, also with machine learning.
   As described in the latter reference,
   lemmas utilizing the new techniques already led
   to remarkable success in
   improving competitive first-order provers and solving a challenge
   problem.\footnote{The table
\url{http://cs.christophwernhard.com/cdtools/exp-lemmas/lemmas.html} indicates
the state of the art in ATP with respect to the CD problems in the TPTP --
taking into account results that already emerged from this foundational
work.}

\end{enumerate}

The paper is organized as follows. In Sect.~\ref{sec-background}, after a
very brief illustration of the CM, we
introduce {\L}ukasiewicz's problem as well as different representations of it.
We also compare different formal representations of proofs, in particular
the representation by Meredith and the ATP-oriented representation
of the CM.
Section~\ref{sec-cd-basis} presents Meredith’s proof of the
problem. There we reconstruct the historical method of CD in a novel way as a
restricted variation of the CM where proof structures are represented as
terms. The section introduces the formal basis for the comparative analysis
described above. On this basis, Sect.~\ref{sec-red} focuses on global
features to support proof search. It presents the underlying formalism and
results on reducing the size of such proof terms in order to shorten proofs
and to restrict the search space. The formalism worked out in the preceding
two sections is applied in Sect.~\ref{sec-inspecting} to provide a
comprehensive analysis of the two historical proofs by \Lukasiewicz and by
Meredith of our widely studied guiding problem. The results are summarized in
detailed feature tables for each proof. In Sect.~\ref{sec-atp} we contrast
these proofs with proofs of the same problem that were obtained as outputs of ATP
systems, general first-order provers as well as postprocessors and specialized
provers that realize observations and new techniques discussed in the paper.
Section~\ref{sec-conclusion} concludes the paper.

\section{Relating Formal Human Proofs with ATP Proofs}
\label{sec-background}

Our investigations are centered around a historic formal proof, a landmark
result by Jan \Lukasiewicz from 1936, published in 1948 \cite{luk:1948}. It is
expressed with the method of substitution and detachment. In the early 1960s
\Lukasiewicz's proof was modified and slightly shortened by Carew A. Meredith
with his method of condensed detachment (\CD) \cite{meredith:notes:1963}.
Thus, our basis are two slightly different versions of an advanced human-made
formal proof. The proven problem was, upon suggestion in 1988 by Frank
Pfenning \cite{pfenning:single:1988}, a prominent challenge problem for ATP
\cite{wos:contributes:1990,roo:parallel:1992,mccune:wos:cd:1992,fitelson:missing:2001}.
Also the background technique of \CD, translated to hyperresolution
\cite{kalman:cd:1983}, led to many successes of ATP in the 1990s
\cite{mccune:wos:cd:1992,ulrich:legacy:2001}.

Although the problem can be solved by modern ATP systems, the current state is
not satisfying. For implemented provers that operate in a goal-driven way with
the CM or with clausal tableaux the problem is still completely out of reach.
Its difficulty rating in the \TPTP \cite{tptp,tptp-rating} has not stabilized
at ``most easy'', but fluctuates and recent versions of two competition
champions fail on it.\footnote{See Sect.~\ref{subsec-challenge-zero}.}
Since the problem was proven formally by humans, this indicates that proof
search in ATP remains in need for further improvement. Also the proofs
obtained with ATP systems are much longer than the human-made proofs, which
indicates a general weakness in our methods with negative effects on their
performance, let alone the involved annoyance for ATP users.

Our aim here is to improve on these issues of general relevance for ATP.
Nevertheless, we focus on a single problem, which is solvable, yet remains a
challenge for both humans and ATP systems. Its basic structure and features
are common to many first-order problems such that results obtained for the
problem can be assumed to apply also more generally. What justifies the
particular choice of the problem is that we have two related human-made formal
proofs at hand, developed by world-leading masters in the field. Their proofs
by far improve on those by today's ATP systems with respect to invested search
effort and size. Hence, inspecting the human-made proofs in depth should lead
through some sort of ``reverse engineering'' to the discovery of techniques
that were used -- intentionally or intuitively -- by \Lukasiewicz and Meredith
and are useful to advance modern ATP.

In this section we introduce the problem proven originally by \Lukasiewicz and
indicate a new adaptation of \CD, the technique used by Meredith, for ATP,
which will be elaborated in later sections. In contrast to most previous
accounts of \CD in ATP and type theory
\cite{kalman:cd:1983,mccune:wos:cd:1992,hindley:meredith:cd:1990,hindley:book:1997}
our focus is not on regarding \CD as an inference rule, but rather on the
aspect that \CD originally comes with explicitly reified proof structures that
are as a whole accessible as trees or terms, or in compacted form as DAGs. In
this respect our modeling of \CD is oriented at the CM, and could in fact be
understood as a simple special case of the CM. For the background concerning
the CM we refer
to~\cite{bibel:atp:1982,bibel:deduction:1993,bibel:otten:2020}.
Here we just briefly illustrate the CM in the following subsection.

\subsection{A Very Short Illustration of the Connection Method (CM)}

Proof systems in ATP are designed to establish the validity of statements
represented as formulas in some logic, like the first-order logic formula
\[(\forall a\, \lnot \f{P} a \imp \lnot \exists y\, \f{Q} y) \imp \exists z\,
\f{P} z\lor \lnot\exists b\, \f{Q}\f{f} b.\]
In order to simplify the involved mechanisms, often the number of different
logical operators is minimized, e.g., restricted to $\lnot,\lor,\exists$, in
accordance with well-known logical rules. In the present example this leads to
the formula $\lnot(\exists a\, \f{P} a\lor\lnot\exists y\, \f{Q} y)\lor\exists z\, \f{P}
z\lor \lnot\exists b\, \f{Q} \f{f} b$.

Gentzen introduced his traditional calculus of natural deduction for the
same purpose. In a simplified variant of this calculus the validity of
this formula is established as its derivation shown in the following figure.

\begin{figure}[H]
\centering
\begin{minipage}{0.95\textwidth}
\input{img8/briefcm_01_gs}
\end{minipage}
\medskip
\caption{A derivation in GS.}
\label{fig:briefcm:gsderivation}
\end{figure}

This kind of proof representation is extremely redundant. The CM is designed
to eliminate this redundancy; it represents the relevant information given by
this derivation in the following structure attached to the formula.

\begin{figure}[H]
  \centering
  \begin{minipage}{28.5em}
    \vspace{0.5cm}
    $\neg(\f{P}^1\f{a}\lor\neg\exists y\, \f{Q}^0y)\lor\exists z\, \f{P}^0z\lor \neg
    \f{Q}^1\f{f} \f{b}\;$ with
    $\;\sigma = \{y_1\mapsto \f{f} \f{b}, z_1\mapsto \f{a}\}$\\
      {\setlength{\unitlength}{0.32mm}
        \begin{picture}(0.00,0.00)(0.00,0.00)
          \qbezier(20.00,25.00)(70.00,43.00)(116.00,25.00)
          \qbezier(70.00,8.00)(110.00,-10.00)(154.00,8.00)
        \end{picture}}
  \end{minipage}
  \medskip
  \caption{The connection proof for the Skolemized formula from
    Fig.~\ref{fig:briefcm:gsderivation}.}
  \label{fig:briefcm:gs2derivation}
\end{figure}

The details of this redundancy elimination are formally presented in the paper
\cite{bibel:otten:2020}. The upper indices attached to the predicates indicate
whether the literal occurs positively ($0$) or negatively ($1$) in the
formula. A \name{connection}, i.e., a pair of occurrences in the formula,
links a positive with a negative literal, both with the same predicate symbol.
The CM proof structure can as well be attached to the formula in its original
presentation (i.e., before reducing the number of used logical operators) in a
straightforward way.

CM calculi test such structures according to a certain criterion which is best
illustrated in the matrix representation of the same formula shown in the
following figure.

\begin{figure}[H]
  \centering
  \begin{minipage}{23.3em}
    \setlength\arraycolsep{2\p@}    
    \input{img8/briefcm_03_matrix}
  \end{minipage}
  \vspace{0.2cm}
  \medskip
\caption{The connection proof from Fig.~\ref{fig:briefcm:gs2derivation} in
  matrix representation.}
\label{fig:briefcm:matrix_proof}
\end{figure}

The matrix features three columns or clauses. A \name{path} through such a
matrix (or the corresponding formula) is a set of literals such that exactly
one is picked from each clause. There are exactly two different paths in the
example. A set of connections is called \name{spanning} for a formula if each
path contains (as a subset) at least one of the set's connections. If the
attached substitution unifies each pair of connected literals then the set of
connections is called \name{complementary}. The mentioned criterion for
validity of a formula is the existence of a spanning and complementary set of
connections.

Clauses may be needed multiple times for achieving a proof. This may be
realized by listing them explicitly with different variables or by indexing
just a single occurrence resulting in indexed variables as indicated in
Figs.~\ref{fig:briefcm:gsderivation} and \ref{fig:briefcm:matrix_proof}. The
latter variant complicates the illustrated concepts accordingly.

CM calculi search for spanning and complementary connection sets. A popular
method does this in a certain systematic manner starting from the goal of the
formula (top-down). More involved search strategies are mixing top-down with
bottom-up search (see, e.g., \cite{wb:conjecture:2024}). For numerous further
refinements see
\cite{bibel:atp:1982,bibel:deduction:1993,bibel:otten:2020,wb:comparison:2024}

Proofs in the CM are formal structures attached to the given formula. For ease
of understanding the present paper focuses on equivalent, less compact but
more familiar proof structures like trees and DAGs. Yet the results of the
present paper are also evidence for the advantages of such compact proof
structures like those of the CM.

\subsection{\Lukasiewicz's Shortest Single Axiom for the Implicational
  Fragment of Propositional Logic}
\label{subsec-luk-shortest-axiom}

Classical propositional logic can be formalized with different sets of logic
operators such as, for example, implication and negation, $\{\rightarrow,
\lnot\}$. Abandoning $\neg$ and leaving $\rightarrow$ as the only logic
operator yields a restricted propositional logic, the \name{implicational
  fragment \textbf{IF}}. The original investigations of this logic use
\Lukasiewicz's so-called Polish notation where the implication $p \imp q$ is
written as $\g{Cpq}$.
Following Pfenning \cite{pfenning:single:1988} we formalize \IF in the setting of
modern first-order ATP with a single unary predicate $\P$ to be interpreted as
something like ``provable'' and represent the \IF formulas by terms using the
binary function symbol $\i$ (instead of $\rightarrow$) for implication.
Implicational propositional logic is characterized by the \name{Tarski-Bernays
  Axioms}, that is, the set of the following three axioms called {\em
  simplification\/} (\Simp), {\em Peirce's law\/} (\Peirce) and {\em
  hypothetical syllogism\/} (\Syll).\footnote{As noted by Prior
\cite{prior:logicians:1956}, the Dublin logic school with Prior and Meredith
contracted from \Lukasiewicz the habit of referring to various key formulas by
proper names, in some cases by names used in the Principia Mathematica, for
example \Simp for $\g{CpCqp}$ and \Syll for $\g{CCpqCCqrCpr}$, and in other
cases by the names of logicians associated with the formulas. Thus
$\g{CCCpqpp}$ is called \Peirce and $\g{CCCpqrCCrpCsp}$ is called \Luk. Tables
of such formula nicknames are provided by Prior
\cite[p.~319]{prior:formal:logic:1962} and Dolph Ulrich \cite{ulrich:legacy:2001}.
See also Sect.~\ref{subsubsec-xatt-nicknames}.}

\centerskip
\begin{center}
  \begin{tabular}{L{6em}L{10em}L{13em}}
   Nickname & \Lukasiewicz's notation & First-order representation\\\midrule
  \Simp & $\g{CpCqp}$ & $\forall pq\, \P(\i(p,\i qp))$\\
  \Peirce & $\g{CCCpqpp}$ & $\forall pq\, \P(\i(\i(\i pq),p),p)$\\
  \Syll & $\g{CCpqCCqrCpr}$ & $\forall pqr\, \P(\i(\i pq,\i(\i qr,\i pr)))$
\end{tabular}
\end{center}
\centerskip

Alfred Tarski in 1925 raised the problem to characterize \IF by a single axiom
and solved it with a general technique for packaging axioms
together,\footnote{This technique was reconstructed by Adrian Rezuş
  \cite{rezus:1982,rezus:tc:2016,rezus:tarski:2019}.} which inherently
produced very long axioms. Jan \Lukasiewicz worked on shortening them,
initially by modifying Tarski's packaging method \cite{luk:1948}. As of 1926,
the shortest known single axioms, found by \Lukasiewicz and Mordechaj
Wajsberg, had 25 letters in \Lukasiewicz parenthesis-free
notation~\cite[p.~43]{luk:tarski:aussagenkalkuel:1930}. Two further single
axioms consisting of 17~letters were found by \Lukasiewicz in 1930 and 1932
\cite{luk:1948,sobocinski:1932,rezus:tarski:2019}. In 1936 he then found the
\emph{shortest} single axiom \cite{luk:1948}, which in the literature is
nicknamed after him.\footnote{In modern times Dolph Ulrich
  developed single axioms and axiom-pairs for further logics \cite{ulrich:single:2016}.
  His paper also surveys the known single axioms for implicational logic and provides
  references to a uniqueness result about \Lukasiewicz's shortest axiom, e.g.,
  \cite{thomas:final}.}

\centerskip
\begin{center}
  \begin{tabular}{L{6em}L{10em}L{13em}}
   Nickname & \Lukasiewicz's notation & First-order
   representation\\\midrule
   \Luk & $\g{CCCpqrCCrpCsp}$ & $\forall pqrs\, \P(\i(\i(\i pq,r),\i(\i rp,\i sp)))$
  \end{tabular}
\end{center}
\centerskip

\noindent
In order to show that \Luk is an axiom for implicational propositional logic,
\Lukasiewicz derived \Simp, \Peirce, and \Syll from \Luk with the proof method
of \name{substitution and detachment}, used by him and other logicians since
about 1930. Detachment is also familiar as \name{modus ponens}. His formal
proof published in 1948 \cite{luk:1948}\footnote{\Lukasiewicz's paper is also
reproduced in his \name{Selected Works} \cite[pp.~295--205]{luk:selected:1970}, however with a typo in
the proof: The substitution of thesis 18 reads $r/CCrCsp$ instead of the
correct $r/CCrpCsp$.} is presented in 29 steps, most of them corresponding to
a single application of substitution and detachment, but some to two
consecutive applications, such that the proof in total involves 34
applications of detachment. Among the three Tarski-Bernays axioms \Syll is by
far the most challenging to prove such that \Lukasiewicz's proof is centered
around the proof of \Syll, with \Simp and \Peirce spinning off as side
results.
In 1963 Carew A. Meredith \cite{meredith:notes:1963} presented a ``very slight
abridgement'' of \Lukasiewicz's proof, expressed in his framework of \CD
\cite{prior:logicians:1956}, where the performed substitutions are no longer
explicitly presented but implicitly assumed through unification. Meredith's
variation involves only 33 applications of detachment, one less than
\Lukasiewicz's original proof.

\subsection{The First-Order ATP View on Detachment}

In our first-order setting, detachment can be modeled with the following
axiom.
\[\Det\; \eqdef\; \forall xy\, (\P x \land \P\i xy \imp \P y).\]
In~\Det the atom $\P x$ is called the \name{minor premise}, $\P\i xy$ the
\name{major premise}, and $\P y$ the \name{conclusion}. Let us now focus on
the following particular formula.
\[\PSYLL\; \eqdef\; \Luk \land \Det \imp \Syll.\]
Showing that \Luk together with the detachment axiom implies \Syll, is then
the problem of proving the validity of the first order formula~\PSYLL. This
formula features a rather simple structure: it asserts that from the proper
axiom \Luk at the left the goal \Syll at the right can be derived via \Det,
the rule in the middle, coding the well-known modus ponens -- or detachment.
Although it looks so simple, finding its proof amounts to a real challenge,
both for humans and machines. Since formulas of a similar structure with
axiom(s), rule(s) and goal(s) are quite frequent, progress in finding their
proofs automatically is clearly desirable. We believe that a deeper
understanding of the underlying proof structure is indispensable for such
progress. The study in this paper aims exactly at such an understanding.

In view of the CM \cite{bibel:atp:1982,bibel:deduction:1993,bibel:otten:2020},
a formula is valid if and only if there is a spanning and complementary set of
connections in it. In Fig.~\ref{fig-ConnT} the formula \PSYLL is presented
again, nicknames dereferenced and quantifiers omitted as usual in ATP, with
the five unifiable connections in it. The symbols $\f{a},\f{b},\f{c}$ in the
conclusion are Skolem constants introduced for the universal variables in
\Syll. The pair consisting of axiom \Luk and the conclusion might be seen as a
further connection~\textbf{0}, but is not depicted because it is not unifiable
and thus irrelevant for any proof. Any CM proof of \PSYLL consists of a number
of instances of the five shown connections. For example, Meredith's proof of
\Syll from \Luk involves 491~instances of \Det (as shown in more detail in
Sect.~\ref{subsec-properties}), each linked with three instances of its
five incident connections. This large number already demonstrates that such a
proof cannot be found and overlooked by humans except with some structural
concept for reducing the sheer proof size.

\begin{figure}
  \centering
  \begin{minipage}{27em}
    \vspace{1cm}
    $\P\i(\i(\i pq,r),\i(\i rp,\i sp))\land (\P x\wedge \P\i xy\rightarrow
    \P y)\rightarrow \P\i(\i\f{ab},\i(\i\f{bc},\i\f{ac}))$\\
       {%
         \setlength{\unitlength}{0.28mm}
         \hspace*{1em}\begin{picture}(0.00,0.00)(-15.00,3.10)
           \qbezier(-6.00,30.00)(54.00,49.00)(114.00,30.00)
           \put(52,29.5){\scriptsize\bf 5}
           \qbezier(-9.00,32.00)(66.00,58.00)(140.00,32.00)
           \put(66,47){\scriptsize\bf 4}
           \qbezier(118.00,14.00)(152.00,-10.00)(186.00,14.00)
           \put(151,-7){\scriptsize\bf 3}
           \qbezier(144.00,15.00)(163.00,0.00)(182.00,15.00)
           \put(162,9){\scriptsize\bf 2}
           \qbezier(186.00,30.00)(204.00,44.00)(222.00,30.00)
           \put(202,39){\scriptsize\bf 1}
         \end{picture}
       }
  \end{minipage}
  \vspace{12pt}
  \caption{The first-order formula \PSYLL along with its five unifiable
    connections.}
  \label{fig-ConnT}
\end{figure}

The concept for such a reduction consists in the well-known feature of
involving lemmas. In terms of the shown connection structure this means that a
certain number of rule instances along with their connection instances are
noted as such a lemma in some abbreviated form that can be referenced several
times in the presentation of the final proof. This way the size of the proof
may be reduced substantially without dispensing the basic characterization of
proofs in the CM. By the use of lemmas that permit reusing subproofs with the
same structure but different instantiations, Meredith's proof of $\Syll$
reduces from 491~to 31~detachment steps. With two more steps, the proof also
yields $\Peirce$ and $\Simp$, resulting in the total number of 33~detachment
steps mentioned above.

Under this extended view, our aim for a deeper understanding of such proofs
raises further questions. Can lemmas that are useful for such reductions be
characterized by syntactic features of the re-used formulas? Or by features of
the proof structure, the re-used subproofs of lemmas in the context of the
overall proof? If we find such features, how could they be utilized to support
the automated search for proofs?

\subsection{Comparing Proof Representations}
\label{subsec-comparing}

\begin{figure}[t]
  \centering
  \begin{minipage}{0.81\textwidth}
    \normalsize
    \renewcommand{\xconn}[1]{{\scriptsize \textbf{#1}}}

\noindent(a)\\[-14pt]
\noindent
\hspace*{-0.28cm}
\scalebox{0.722}{{%
    \begin{tikzpicture}
      \input{img7/cmproof_expanded_02_xy}
    \end{tikzpicture}
}}

\vspace{-0.76cm}
  
\noindent(b)
\raisebox{0.2cm}{%
\scalebox{0.8}{%
\begin{tikzpicture}
  [baseline=(current bounding box.north),
    xdot/.style = {minimum size=11pt,
      inner sep=0pt, outer sep=2pt, on grid},
    xsqr/.style = {minimum size=11pt,
      inner sep=0pt, outer sep=2pt, on grid}]  
\node[xdot] (e) {\small $D_1$};
\node[xdot, below left=1cm and \tl of e] (1) {\small $D_2$};
\node[xdot, below right=1cm and \tl of e] (2) {\small $D_3$};
\node[xsqr, below left=1cm and \tll of 1] (11) {\small $A_1$};
\node[xsqr, below right=1cm and \tll of 1] (12) {\small $A_2$};
\node[xdot, below left=1cm and \tll of 2] (21) {\small $D_4$};
\node[xdot, below right=1cm and \tll of 2] (22) {\small $D_6$};
\node[xsqr, below left=1cm and \tlll of 21] (211) {\small $A_3$};
\node[xdot, below right=1cm and \tlll of 21] (212) {\small $D_5$};
\node[xsqr, below left=1cm and \tlll of 22] (221) {\small $A_6$};
\node[xdot, below right=1cm and \tlll of 22] (222) {\small $D_7$};
\node[xsqr, below left=1cm and \tllll of 212] (2121) {\small $A_4$};
\node[xsqr, below right=1cm and \tllll of 212] (2122) {\small $A_5$};
\node[xsqr, below left=1cm and \tllll of 222] (2221) {\small $A_7$};
\node[xsqr, below right=1cm and \tllll of 222] (2222) {\small $A_8$};
\draw[] (e) -- node [pos=0.5,above left] {\xconn{2}} (1);
\draw[] (e) -- node [pos=0.5,above right] {\xconn{3}} (2);
\draw[] (1) -- node [pos=0.6,above left] {\xconn{4}} (11);
\draw[] (1) -- node [pos=0.6,above right] {\xconn{5}} (12);
\draw[] (2) -- node [pos=0.6,above left] {\xconn{2}} (21);
\draw[] (2) -- node [pos=0.6,above right] {\xconn{3}} (22);
\draw[] (21) -- node [pos=0.6,above left] {\xconn{4}} (211);
\draw[] (21) -- node [pos=0.6,above right] {\xconn{3}} (212);
\draw[] (22) -- node [pos=0.6,above left] {\xconn{4}} (221);
\draw[] (22) -- node [pos=0.6,above right] {\xconn{3}} (222);
\draw[] (212) -- node [pos=0.6,above left] {\xconn{4}} (2121);
\draw[] (212) -- node [pos=0.6,above right] {\xconn{5}} (2122);
\draw[] (222) -- node [pos=0.6,above left] {\xconn{4}} (2221);
\draw[] (222) -- node [pos=0.6,above right] {\xconn{5}} (2222);
\end{tikzpicture}}}%
\hspace{35pt}%
\raisebox{-28pt}{%
\noindent(c)
\scalebox{0.8}{%
\begin{tabular}[t]{R{2em}@{\hspace{0.5em}}l}
 1. & $\g{CCCpqrCqr}$\\
 2. & $\g{CpCqp} \mereq \f{D}11$\\
 3. & $\g{CpCqCrp} \mereq \f{D}12$\\
 * 4. & $\g{CpCqCrCsCtCus} \mereq \f{D}2\f{D}33$\\
  \end{tabular}}}\\[-7pt]

\noindent(d)%
\raisebox{0.2cm}{%
\scalebox{0.8}{%
\begin{tikzpicture}
  [baseline=(current bounding box.north),
    dot/.style = {circle, minimum size=11pt,
      inner sep=0pt, outer sep=0pt, draw, on grid},
    sqr/.style = {regular polygon sides=4, minimum size=11pt,
      inner sep=0pt, outer sep=0pt, draw, on grid}]  
\node[dot] (e) {\small $4$};
\node[dot, below left=1cm and \tl of e] (1) {\small $2$};
\node[dot, below right=1cm and \tl of e] (2) {\small $$};
\node[sqr, below left=1cm and \tll of 1] (11) {\small $1$};
\node[sqr, below right=1cm and \tll of 1] (12) {\small $1$};
\node[dot, below left=1cm and \tll of 2] (21) {\small $3$};
\node[dot, below right=1cm and \tll of 2] (22) {\small $3$};
\node[sqr, below left=1cm and \tlll of 21] (211) {\small $1$};
\node[dot, below right=1cm and \tlll of 21] (212) {\small $2$};
\node[sqr, below left=1cm and \tlll of 22] (221) {\small $1$};
\node[dot, below right=1cm and \tlll of 22] (222) {\small $2$};
\node[sqr, below left=1cm and \tllll of 212] (2121) {\small $1$};
\node[sqr, below right=1cm and \tllll of 212] (2122) {\small $1$};
\node[sqr, below left=1cm and \tllll of 222] (2221) {\small $1$};
\node[sqr, below right=1cm and \tllll of 222] (2222) {\small $1$};
\draw[] (e) -- node [pos=0.5,above left] {\xconn{2}} (1);
\draw[] (e) -- node [pos=0.5,above right] {\xconn{3}} (2);
\draw[] (1) -- node [pos=0.6,above left] {\xconn{4}} (11);
\draw[] (1) -- node [pos=0.6,above right] {\xconn{5}} (12);
\draw[] (2) -- node [pos=0.6,above left] {\xconn{2}} (21);
\draw[] (2) -- node [pos=0.6,above right] {\xconn{3}} (22);
\draw[] (21) -- node [pos=0.6,above left] {\xconn{4}} (211);
\draw[] (21) -- node [pos=0.6,above right] {\xconn{3}} (212);
\draw[] (22) -- node [pos=0.6,above left] {\xconn{4}} (221);
\draw[] (22) -- node [pos=0.6,above right] {\xconn{3}} (222);
\draw[] (212) -- node [pos=0.6,above left] {\xconn{4}} (2121);
\draw[] (212) -- node [pos=0.6,above right] {\xconn{5}} (2122);
\draw[] (222) -- node [pos=0.6,above left] {\xconn{4}} (2221);
\draw[] (222) -- node [pos=0.6,above right] {\xconn{5}} (2222);
\end{tikzpicture}}}
\hspace{\fill}%
(e)%
\raisebox{0.2cm}{%
\scalebox{0.8}{%
\begin{tikzpicture}
  [baseline=(current bounding box.north),
    >={Latex[length=4.5pt]},
    dot/.style = {circle, minimum size=11pt,
      inner sep=0pt, outer sep=0pt, draw, on grid},
    sqr/.style = {regular polygon sides=4, minimum size=11pt,
      inner sep=0pt, outer sep=0pt, draw, on grid}]  
\node[dot] (e) {\small $4$};
\node[dot, below left=1cm and \tl of e] (1) {\small $2$};
\node[dot, below right=1cm and \tl of e] (2) {\small $$};
\node[sqr, below left=1cm and \tll of 1] (11) {\small $1$};
\node[dot, below left=1cm and \tll of 2] (21) {\small $3$};
\draw[->] (e) -- node [pos=0.5,above left] {\xconn{2}} (1);
\draw[->] (e) -- node [pos=0.5,above right] {\xconn{3}} (2);
\draw[->] (1) -- node [pos=0.5,above left] {\xconn{4}} (11);
\draw[->] (1) to [out=-40,in=15] node [pos=0.5,below right] {\xconn{5}} (11);
\draw[->] (2) -- node [pos=0.5,above left] {\xconn{2}} (21);
\draw[->] (2) to [out=-40,in=15] node [pos=0.5,below right] {\xconn{3}} (21);
\draw[->] (21) to [out=240,in=-22] node [pos=0.5,above] {\xconn{4}} (11);
\draw[->] (21) to [out=-60,in=-10,out looseness=1.5] node [pos=0.2,below] {\xconn{3}} (1);
\end{tikzpicture}}}
  \end{minipage}
  \vspace{10pt}
  \caption{A proof in different representations.}
  \label{fig-representations}
\end{figure}

Figure~\ref{fig-representations} compares different representations of a short
formal proof with the \Det axiom. There is a single proper axiom,
\name{Syll-Simp}\footnote{As noted by Prior \cite{prior:formal:logic:1962},
\name{Syll-Simp} is the goal theorem of a proof with a single \CD step applied
to \Syll and \Simp as axioms. It appears as a lemma in the investigated proofs
by \Lukasiewicz and Meredith, see Sect.~\ref{sec-inspecting}.} defined as
  follows.

\centerskip
\begin{center}
  \begin{tabular}{L{6em}L{10em}L{13em}}  
   Nickname & \Lukasiewicz's notation & First-order
   representation\\\midrule
  \name{Syll-Simp} & $\g{CCCpqrCqr}$ & $\forall pqr\,
\P\i(\i(\i pq,r),\i qr)$\\
  \end{tabular}
\end{center}
\centerskip

\noindent
The goal theorem is $\forall abcdef\, \P \i(a,\i(b,\i(c,\i(d,\i(e,\i
fd)))))$.\footnote{This theorem has been chosen as proof goal just because it
has a proof that is suitable to illustrate the interplay of the considered
proof representations.} Figure~\ref{fig-representations}a shows the structure
of a CM proof. It involves seven instances of \Det, shown in columns
$D_1, \ldots, D_7$.\footnote{Note that the separate display of these instances
of \Det is only for a better understanding of the reader but not a feature of
the CM, which rather involves instead indexed connections for the formula
given in Fig.~\ref{fig-ConnT}.} The major premise $\P\i x_i y_i$ is displayed
there on top of the minor premise~$\P x_i$, and the (negated) conclusion
$\lnot \P y_i$, where $x_i, y_i$ are variables. Instances of the axiom appear
as literals $\lnot \P a_i$, with $a_i$ a shorthand for the term $\i(\i(\i p_i
q_i,r_i),\i q_i r_i)$. The rightmost literal~$\P g$ is a shorthand for the
Skolemized goal theorem.
The clause instances are linked through edges representing connection
instances. The edge labels identify the respective connections as in
Fig.~\ref{fig-ConnT}. An actual connection proof is obtained by supplementing
this structure with a substitution under which all pairs of literals related
through a connection instance become complementary.

Figure~\ref{fig-representations}b represents the \emph{tree} implicit in the
CM proof. Its inner nodes $D_1,\ldots,D_7$ correspond to the seven instances
of \Det, and its leaf nodes $A_1,\ldots,A_8$ to the instances of the axiom.
Edges appear ordered to the effect that those originating in a major premise
of \Det are directed to the left and those from a minor premise to the right.
The goal clause~$\P g$ is dropped. The resulting tree is a \emph{full binary
tree}, i.e., a binary tree where each node has 0 or 2 children. We observe
that the ordering of the children makes the connection labeling redundant as
it directly corresponds to the tree structure.

Figure~\ref{fig-representations}c presents the proof in Meredith's notation
for \CD. Each line shows a formula, line~1 the axiom and lines 2--4 derived
formulas, with proofs annotated in the last column. Proofs are written as
terms in Polish notation with the binary function symbol $\f{D}$ for
\name{detachment} where the subproofs of the major and minor premise are
supplied as first and second argument, respectively. Formula~4, for example,
is obtained as conclusion of \Det applied to formula~2 as major premise and
another formula not made explicit in the presentation as minor premise, namely
the conclusion of \Det applied to formula~3 as both, major and minor,
premises. An asterisk marks the goal theorem.

Figure~\ref{fig-representations}d is like Fig.~\ref{fig-representations}b, but
with a different labeling: Node labels now refer to the line in
Fig.~\ref{fig-representations}c, which corresponds to the subproof rooted at
the node. The blank node represents the mentioned subproof of the formula that
is not made explicit in Fig.~\ref{fig-representations}b. An inner node
represents a \CD step applied to the subproof of the major premise (left
child) and minor premise (right child).

Figure~\ref{fig-representations}e shows a DAG representation
of~Fig.~\ref{fig-representations}d. It is the unique minimal, or maximally
factored, DAG representation of the tree, i.e., it has no multiple occurrences
of the same subtree. Each of the four proof line labels of
Fig.~\ref{fig-representations}c appears exactly once in the DAG. The
presentation layout of the DAG reflects a tree compacting procedure, the
\name{value-number method}, which computes unique identifiers for all subtrees
in a post-order tree traversal \cite{aho:compilers:86,genitrini:2020}. A
straight edge corresponds to the first visit of the subtree rooted at its
endpoint, and a bended edge to a pointer to a previously identified subtree.
Observe that each of the four proof line labels of Meredith's representation
(Fig.~\ref{fig-representations}c) appears exactly once in the DAG. In fact,
the structural component of the textual proof representation (that is, if we
disregard the displayed formulas) can be considered as a compact notation for
such a DAG.

\section{Condensed Detachment and a Formal Basis}
\label{sec-cd-basis}

With Fig.~\ref{fig-representations}c we already have seen a small example of
a \CD proof in Meredith's notation. Figure~\ref{fig-proof-mer} shows
Meredith's \CD proof that \Luk entails \Syll, \Peirce and \Simp, taken from a
1963 paper by Meredith and Prior \cite{meredith:notes:1963}. There is a single
axiom, 1, which is \Luk. The proven theorems are \Syll (17), \Peirce (18)
and \Simp (19), marked by asterisks. In addition to line numbers also the
symbol ``$\n$'' appears in some of the proof terms. We will discuss its
meaning in Sect.~\ref{subsec-simp-n} and, for now, read it just as ``1''.
Dots are used in the Polish notation to disambiguate numeric identifiers with
more than a single digit, for example in line~11.

\begin{figure}
  \centering
  \begin{tabular}{rl@{\hspace{0.2em}}l}
      1. & $\g{CCCpqrCCrpCsp}$\\
      2. & $\g{CCCpqpCrp} \mereq \f{DDD}1\f{D}111\n$\\
      3. & $\g{CCCpqrCqr}  \mereq  \f{DDD}1\f{D}1\f{D}121\n$\\
      4. & $\g{CpCCpqCrq}  \mereq  \f{D}31$\\
      5. & $\g{CCCpqCrsCCCqtsCrs}  \mereq  \f{DDD}1\f{D}1\f{D}1\f{D}141\n$\\
      6. & $\g{CCCpqCrsCCpsCrs}  \mereq  \f{D}51$\\
      7. & $\g{CCpCqrCCCpsrCqr}  \mereq  \f{D}64$\\
      8. & $\g{CCCCCpqrtCspCCrpCsp}  \mereq  \f{D}71$\\
      9. & $\g{CCpqCpq}  \mereq  \f{D}83$\\
     10. & $\g{CCCCrpCtpCCCpqrsCuCCCpqrs}  \mereq  \f{D}18$\\
     11. & $\g{CCCCpqrCsqCCCqtsCpq}  \mereq  \f{DD}10.10.\n$\\
     12. & $\g{CCCCpqrCsqCCCqtpCsq}  \mereq  \f{D}5.11$\\
     13. & $\g{CCCCpqrsCCsqCpq}  \mereq  \f{D}12.6$\\
     14. & $\g{CCCpqrCCrpp}  \mereq  \f{D}12.9$\\
     15. & $\g{CpCCpqq}  \mereq  \f{D}3.14$\\
    16. & $\g{CCpqCCCprqq}  \mereq  \f{D}6.15$\\
    *  17. & $\g{CCpqCCqrCpr}  \mereq  \f{DD}13.\f{D}16.16.13$\\
    *  18. & $\g{CCCpqpp}  \mereq  \f{D}14.9$\\
    * 19. & $\g{CpCqp}  \mereq  \f{D}33$\\
  \end{tabular}
  \vspace{6pt}
  \caption{Meredith's CD variation of \Lukasiewicz's proof \cite{luk:1948} of
    \Syll (17), \Peirce (18) and \Simp (19) from \Luk. From Meredith and Prior
    \cite{meredith:notes:1963}, with correction of a typo in the formula of
    line~7. We will refer to this proof as $\DMER$.}
  \label{fig-proof-mer}
\end{figure}

Following Martin W.~Bunder \cite{bunder:cd:1995}, the idea of \CD can be
described as follows: Given premises \mbox{$F \imp G$} and $H$, we can
conclude $G'$, where $G'$ is the most general result that can be obtained by
using a substitution instance $H'$ of $H$ as minor premise with the
substitution instance $F' \imp G'$ of $F \imp G$ as major premise in modus
ponens. \CD was introduced by Meredith in the mid-1950s as an evolution of the
earlier \name{method of substitution and detachment}, where the involved
substitutions were explicitly given.

The original presentations of \CD are informal, by means of examples
\cite{prior:logicians:1956,lemmon:meredith:purestrict:1957,prior:formal:logic:1962,meredith:memoriam:1977}.
Only later, formal specifications have been given. John A. Kalman
\cite{kalman:cd:1983} provides two characterizations, one in terms of
resolution. \CD was then considered in the context of type theory, the
formulas-as-types view, where J.~Roger Hindley and David Me\-re\-dith
\cite{hindley:meredith:cd:1990,hindley:book:1997} notice and fix an inaccuracy
related to the notion of \name{most general unifier} in the early
formalizations of \CD and Bunder \cite{bunder:cd:1995} provides a
formalization that is independent from this notion. A particular investigated
topic concerning \CD in type theory is the relationship to substitution and
detachment.

Unfortunately it seems that not much is bequeathed about the methods by which
humans \emph{found} advanced \CD
proofs. \Lukasiewicz \cite[\Paragraph~4]{luk:1948} discusses an important
intermediate step for his proof by substitution and detachment. Legend has
it that Meredith often sent his finished \CD proofs as
postcards \cite{prior:logicians:1956,bull:interview}.

In ATP, the rendering of \CD by positive hyperresolution with the clausal form
of axiom \Det is so far the prevalent view. As overviewed by William McCune
and Larry Wos \cite{mccune:wos:cd:1992}, and Dolph
Ulrich \cite{ulrich:legacy:2001}, many of the early successes of ATP were
based on \CD. Starting from the hyperresolution view, structural aspects
of \CD have been considered by Robert Veroff
\cite{veroff:shortest:2001} with the use of term representations of proofs and
\emph{linked} resolution. Results of ATP systems on deriving the
Tarski-Bernays axioms from \Luk are reported in several papers
\cite{pfenning:single:1988,wos:contributes:1990,roo:parallel:1992,mccune:wos:cd:1992,fitelson:missing:2001}.
The problems of deriving \Syll, \Peirce and \Simp from \Luk are in the \TPTP
as \TP{LCL038-1}, \TP{LCL083-1} and \TP{LCL082-1}, respectively.
In general, many refinements of the \Otter prover \cite{otter} in the 1990s,
some of which have found their ways into modern saturating provers, were
originally conceived and explored in the setting of CD
\cite{wos:contributes:1990,wos:bledsoe:91,mccune:wos:cd:1992,wos:resonance:95,wos:combining:96,veroff:shortest:2001,fitelson:missing:2001,wos:meredith}.
Various sources compile open and challenge problems concerning CD, along with
some solutions or partial solutions
\cite{ulrich:legacy:2001,ulrich:cd:web,veroff:cd:2011,fitelson:walsh:2021}.
A sustaining and far-reaching application of CD
is \name{Metamath} \cite{megill:1995,metamath:book}, a successful
computer-processable language for verifying, archiving, and presenting
mathematical proofs. ``Simple by design'', it is entirely based on CD extended
by a second rule for \name{condensed generalization}.

From the viewpoint of general first-order ATP, \CD basically offers a
simplified, streamlined setting for investigations and developments that
nevertheless includes with first-order variables, binary function symbols and
cyclic predicate dependency core characteristics of first-order ATP. The
simplifications concern the restricted application domain, axiomatizations of
propositional logics, which is, however, not difficult to lift to Horn
problems in general,\footnote{This has, e.g., been implemented in the
\name{CCS} \CD reasoner \cite{cw:ccs}.} no explicit consideration of non-Horn
problems\footnote{First-order logic permits to encode a non-Horn problem as a
Horn problem.}, and no explicit use of equality.\footnote{There are, however,
relationships to equality. It is well-known that equality can be axiomatized
by Horn clauses expressing reflexivity, symmetry, transitivity and
substitutivity. It is also possible to encode Horn problems as purely
equational problems \cite{claessen:smallbone:2018}, where,
e.g., \TP{LCL038-10} is an equational variation of \TP{LCL038-1} (\Syll
from \Luk). Some \CD problems, e.g., \TP{LCL006-1}, are about axiomatizations
of an equivalential calculus.} The \TPTP contains easy and still very hard \CD
problems.

But \CD offers more. It integrates various features of relevance to ATP in a
natural and formally accessible way, which we outline in the following
paragraphs.

\CD distinguishes from its predecessor, the method of substitution and
detachment, by applying most general substitutions that are obtained through
unification. In \CD proof presentations, just most general formulas resulting
from unification are written, the involved substitutions are left implicit.
Remarkably, unification was applied with \CD extensively in formal deduction a
decade before it became popular in the context of resolution through
John Alan Robinson \cite{robinson:1965}.

\CD proofs are presented in the literature as a sequence of pairs of a lemma
and a proof structure term that describes how the lemma is proven from
previous lemmas. The structure terms can be combined to form a tree for each
goal theorem or to a DAG representing the set of these trees more compactly
such that subtrees with multiple occurrences appear only once. Both
representations have their merits. The explicit tree view facilitates to
associate semantic properties and formula substitutions in an inductive
fashion. It permits to understand variables in a particular simple way as
scoped over the whole structure, known as \name{rigid} variables in tableaux.
The compacted view in particular provides an adequate notion of \emph{proof
size} and, in printed form, is much easier to comprehend by humans.

A related separation of concerns regarding proof structure and associated
formulas is provided among the modern approaches to ATP by the CM. In fact, as
illustrated with Fig.~\ref{fig-representations} above, \CD can be understood
as an adaptation of the CM to inputs of a specific simple form: a single
clause with three literals, which represents the \Det axiom, and otherwise
just unit clauses, representing proper axioms and the theorem to be proven.

The separation of a deductive derivation into a formula part and a proof
structural part, as illustrated in Fig.~\ref{fig-proof-mer}, can be seen as a
precursor of the CM. Namely, the CM has carried this separation to the extreme
in that it keeps the formula part completely unchanged within such a
derivation and shifts all deductive information into the proof structural part
(see, e.g., \cite[Section~III.6]{bibel:atp:1982}).

In the traditional presentation of a \CD proof the members of the sequence of
lemma-structure pairs are labeled with numbers, where the labeling turns out
to be useful for the following two purposes. For a lemma that is referenced
multiple times in the overall proof, a label is necessary to represent
the proof structure compactly as a DAG. For a lemma that is referenced only
once, the presentation by a labeled pair is optional and serves the
convenience of a human reader or points out some special significance of the
lemma. Otherwise, lemmas that are referenced only once do not appear
explicitly in the proof presentation but could be obtained as the most general
formulas proven by the substructures of the structure components of the
labeled lemma-structure pairs.

The term view of proof structures lets the replacement of subproof occurrences
appear as a form of term rewriting, with shorter subproofs that preserve
equivalence in some sense. A suitable notion of equivalence can be based on
the most general formula that can be proven with a given proof term by
applying detachment steps according to the term structure from given axioms.
Such proof reductions can be applied to simplify given proofs, or in proof
search, to justify that a subproof recognized as reducible can be immediately
discarded, because there must exist a different preferable subproof.

The term view of proof structures is also the basis of a recent technique
where combinators are applied to express stronger compressions of the proof
structure than just to DAGs \cite{cw:ccs}. Such compressions can be applied to
shorten given proofs and in proof search. They correspond to more complex
lemma formulas than the unit lemmas considered in the DAG compression, and can
express simulations of other calculi.

Search for a \CD proof can be performed goal- or axiom-driven. Consideration
of a goal (e.g., a ground atom resulting from Skolemizing a universally
quantified atom) in the unifying substitution to determine the formulas
involved in the proof is optional. Taking the goal into account effects
restriction of the search space, as in the conventional goal-driven
realizations of the CM. Nevertheless, also axiom-driven proceeding without
supplied goal is possible with very similar search mechanisms, enumerating
proof structures interwoven with unification. The results then are
consequences derived from axioms, which optionally may be used as lemmas to
improve proof search in a second goal-driven phase
\cite{schumann:delta:1994,cw:sgcd}.

In a wider perspective the consideration of the proof structure as a whole,
for example as term, which may be compacted into a DAG, introduces an
important separation of concerns for proof search. Namely, the way in which
the concrete structure is built up in proof search is not obliged to follow
the inductive specification of the structure. The concrete structure can be
built up in various ways, including rewriting of subproofs as indicated above,
or by combining given proof fragments. This contrasts with calculi such as
typical tableau methods where proof construction rules are directly taken to
build up the proof structures.

Our goal in this section is to provide a formal framework that takes account
of these aspects and provides a basis for experiments and future developments
in ATP.

\subsection{Notation}
\label{subsec-notation}

Most of our notation follows common practice \cite{dershowitz:notations:1991}.
The set of variables occurring in a term~$s$ is denoted by~$\vars(s)$. We
extend this to other objects~$s$ such as, e.g., sets of terms.
A \defname{substitution} is a mapping from variables to terms which is almost
everywhere equal to identity. If $\sigma$ is a substitution, then
the \defname{domain} of $\sigma$ is the set of variables
$\dom(\sigma) \eqdef \{x \mid x\sigma \neq x\}$, the
\defname{range} of $\sigma$ is $\rng(\sigma) \eqdef \bigcup_{x \in
  \dom(\sigma)} \{x\sigma\}$, and the \defname{restriction} of $\sigma$ to a
set~$X$ of variables, denoted by $\sigma|_X$, is the substitution which is
equal to the identity everywhere except over $X \cap \dom(\sigma)$, where it
is equal to~$\sigma$. The identity substitution is denoted by $\emptysubst$.
We write the set~$\vars(\rng(\sigma))$ of variables in the range of
substitution~$\sigma$ also as $\vrng(\sigma)$. A substitution can be
represented by a set of assignments of the variables in its
domain, e.g., $\{x_1 \mapsto t_1, \ldots, x_n \mapsto t_n\}$. The application
of a substitution~$\sigma$ to a term $s$ is written as $s\sigma$, $s\sigma$ is
called an \defname{instance} of $s$ and $s$ is said to
\defname{subsume} $s\sigma$.  That~$s$ subsumes~$t$, or synonymously, that $t$
is an instance of $s$, is expressed symbolically by \[t \subsumedBy s.\] If
both, $s \subsumedBy t$ and $t \subsumedBy s$, hold we say that $s$ and~$t$
are \defname{variants} of each other, expressed symbolically as $s \variant
t$. Composition of substitutions is written as juxtaposition. Hence, if
$\sigma$ and $\theta$ are both substitutions, then $E\sigma\theta$ stands for
$(E\sigma)\theta$. A substitution~$\sigma$ is \defname{idempotent} if
$\sigma\sigma = \sigma$, or, equivalently, $\dom(\sigma) \cap \vrng(\sigma) =
\emptyset$.  A substitution $\sigma$ is called \defname{more general than} a
substitution $\theta$, in symbols $\theta \subsumedBy \sigma$, if there exists
a substitution~$\rho$ such that $\sigma\rho = \theta$. That both, $\sigma
\subsumedBy \theta$ and $\theta \subsumedBy \sigma$ hold is expressed by
$\sigma \variant \theta$.

A \defname{position} is a sequence of positive integers that specifies the
occurrence of a subterm in a term as a path in Dewey decimal notation starting
from the root of the term. The set of all positions of a term $s$ is denoted
by $\pos(s)$. For example, $\pos(\f{f}(x,\f{g}(y))) = \{\emptypos, 1, 2,
2.1\}$. If position~$p$ is a prefix of position~$q$, we write \[p \leq q\] and
say that $p$ is \defname{above}~$q$, and $q$ is \defname{below}~$p$. We also
use $p \not \leq q$, $p < q$ and $p \not < q$ for positions $p,q$ with the
obvious analog meanings. For $p \in \pos(s)$, the \defname{subterm of $s$ at
  position $p$} is denoted by $s|_p$. For example, if $s = \f{f}(x,\f{g}(y))$,
then $s|_{\emptypos} = s = \f{f}(x,\f{g}(y))$, $s|_{1} = x$, $s|_{2} =
\f{g}(y)$ and $s|_{2.1} = y$. That $s$ is a subterm of $t$ is expressed
symbolically as \[t \suptermq s\] and that $s$ is a strict subterm of $t$ as
$t \supterm s$.
For $p \in \pos(s)$, the expression \[s[t]_p\] denotes the term obtained from
$s$ by replacing the subterm occurrence at position $p$ with term~$t$, or, in
case $s|_p = t$, to denote~$s$ with indicating that $t$ occurs at position $p$
in $s$.

In addition to common notation, we use a few special symbols and conventions:
The set of positions $p \in \pos(s)$ such that $s|_p$ is a variable or a
constant is denoted by $\leafpos(s)$ and the set of positions $p \in \pos(s)$
such that $s|_p$ is a compound term by~$\innerpos(s)$. We use the postfix
notation for the application of a substitution~$\sigma$ also for sets~$M$ of
pairs of terms: $M\sigma$ stands for $\{\{s\sigma, t\sigma\} \mid \{s, t\} \in
M\}$. For terms $s,t,u$, the expression \[\replace{s}{t}{u}\] denotes $s$
after simultaneously replacing all occurrences of $t$ with $u$. If $F$ is a
formula, then $\UCN{F}$ denotes the universal closure of $F$.

\subsection{Proof Structures: \DTerms}
\label{subsec-struct-base}

In this subsection (as well as in Sect.~\ref{subsec-struct-red} below) we
consider only the purely structural aspects of \CD proofs. Emphasis is on a
twofold view on the proof structure, as a tree and as a DAG (directed acyclic
graph), which factorizes multiple occurrences of the same subtree. Both
representation forms are useful: the compacted DAG form captures that lemmas
can be repeatedly used in a proof, whereas the tree form facilitates to
specify properties in an inductive manner.

\subsubsection{Basic Definitions: Term View and Tree View}

We call the tree representation of proofs by terms with the binary function
symbol~$\D$ \name{\dterms}.

\begin{defn}%
  \label{def:d-term}
  \ \par \subdef{def-d-variables} We assume a distinguished set
  \defname{$\DPRIMSET$} of symbols, called the \defname{primitive \Dterms}.

  \subdef{def-d-term} A \defname{\dterm} is specified inductively as follows.

  \begin{enumerate}[leftmargin=4em]
    \item Any member of $\DPRIMSET$ is a \dterm.
    \item If $d_1$ and $d_2$ are \dterms, then $\D(d_1,d_2)$ is a \dterm.
  \end{enumerate}

  \subdef{def-compound} A \Dterm of the form $\D(d_1,d_2)$ is called
  \defname{compound}.
  
  \subdef{def-dconst} For \Dterms~$d$ define $\DPRIM(d)\; \eqdef\; \{e \mid d
  \suptermq e\} \cap \m{\DPRIMSET}$.

\end{defn}

\noindent
A \dterm~$d$ is a full binary tree (a binary tree where every inner node has
exactly two children, its left and its right child) whose leaves are labeled by primitive \Dterms. $\DPRIM(d)$
denotes the set of the primitive \Dterms that occur in~$d$, or, in other
words, the set of leaf labels of $d$.

\begin{examp}%
  \label{examp-dterm}
  Assume that $\DPRIMSET$ contains the numeral~$1$. Then
  \[d\; \eeqdef\; \D(\D(1,1),\D(\D(1,\D(1,1)),\D(1,\D(1,1))))\]
  is a \Dterm with $\DPRIM(d) = \{1\}$ that represents the structure of the
  proof shown in Fig.~\ref{fig-representations}. Its visualization is shown in
  Fig.~\ref{fig-dterm-tree} (which is identical to
  Fig.~\ref{fig-representations}d after removing all labels with exception of
  the leaf labels).
\end{examp}

\begin{figure}
  \centering
  \scalebox{0.8}{\begin{tikzpicture}
[baseline=(current bounding box.north),
  dot/.style = {circle, minimum size=11pt,
    inner sep=0pt, outer sep=0pt, draw, on grid},
  sqr/.style = {regular polygon sides=4, minimum size=11pt,
    inner sep=0pt, outer sep=0pt, draw, on grid}]
\node[dot] (e) {\small $$};
\node[dot, below left=1cm and \tl of e] (1) {\small $$};
\node[dot, below right=1cm and \tl of e] (2) {\small $$};
\node[sqr, below left=1cm and \tll of 1] (11) {\small $1$};
\node[sqr, below right=1cm and \tll of 1] (12) {\small $1$};
\node[dot, below left=1cm and \tll of 2] (21) {\small $$};
\node[dot, below right=1cm and \tll of 2] (22) {\small $$};
\node[sqr, below left=1cm and \tlll of 21] (211) {\small $1$};
\node[dot, below right=1cm and \tlll of 21] (212) {\small $$};
\node[sqr, below left=1cm and \tlll of 22] (221) {\small $1$};
\node[dot, below right=1cm and \tlll of 22] (222) {\small $$};
\node[sqr, below left=1cm and \tllll of 212] (2121) {\small $1$};
\node[sqr, below right=1cm and \tllll of 212] (2122) {\small $1$};
\node[sqr, below left=1cm and \tllll of 222] (2221) {\small $1$};
\node[sqr, below right=1cm and \tllll of 222] (2222) {\small $1$};
\draw[] (e) -- (1);
\draw[] (e) -- (2);
\draw[] (1) -- (11);
\draw[] (1) -- (12);
\draw[] (2) -- (21);
\draw[] (2) -- (22);
\draw[] (21) -- (211);
\draw[] (21) -- (212);
\draw[] (22) -- (221);
\draw[] (22) -- (222);
\draw[] (212) -- (2121);
\draw[] (212) -- (2122);
\draw[] (222) -- (2221);
\draw[] (222) -- (2222);
\end{tikzpicture}}
  \vspace{8pt}
  \caption{The \Dterm~$d$ from Example~\ref{examp-dterm}.}
  \label{fig-dterm-tree}
\end{figure}

\begin{examp}
The proof annotations in Fig.~\ref{fig-representations}c and
Fig.~\ref{fig-proof-mer} are \dterms written in Polish notation, where
$\DPRIMSET$ is a set $\{1,2,3,\ldots\}$ of numerals. The expression
$\f{D}2\f{D}33$ in line~4 of Fig.~\ref{fig-representations}, for example,
stands for the \dterm $\D(2,\D(3,3))$. Its set $\DPRIM(\D(2,\D(3,3)))$ of
primitive \Dterms is $\{2,3\}$.
\end{examp}

\subsubsection{Tree Size and Height}

The following definition specifies two basic size measures of \DTerms.

\begin{defn}\ \par
  \label{def-treesize-height}
  \subdef{def-treesize}
  The \defname{tree size} of a \dterm $d$, in symbols $\tsize(d)$, is the
  number of occurrences of the function symbol~$\D$ in $d$.

  \subdef{def-height} The \defname{height} of a \dterm $d$, in symbols
  $\height(d)$ is, viewing the term as a tree, the number of edges of the
  longest downward path from the root to a leaf.
\end{defn}

The tree size of a \dterm can equivalently be characterized as the number of
its inner nodes. Veroff \cite{veroff:shortest:2001} calls it \name{CDcount}.
As will be explicated in more detail in Sect.~\ref{subsec-subst-base},
each occurrence of the function symbol~$\D$ in a \dterm corresponds to an
instance of the axiom $\Det$ in the represented proof. Hence the tree size
measures the number of instances, or multiplicity, of $\Det$ in the proof.
Another view is that each occurrence of $\f{D}$ in a \dterm corresponds to a
detachment step, without re-using already proven lemmas and instead again
re-proving each lemma whenever it is used. The tree size of the \Dterm~$d$ of
Example~\ref{examp-dterm} is $\tsize(d) = 7$.

The height of a \Dterm is just its height according to the conventional notion
of the height of a tree. Applied to terms it is often also
called \name{depth}. For \Dterms, it is called \name{level}
by Veroff \cite{veroff:shortest:2001}. The height of the \Dterm~$d$ of
Example~\ref{examp-dterm} is $\height(d) = 4$.

\subsubsection{DAG Representation and Compacted Size}

A finite tree and, more generally, a finite set of finite trees can be
represented as a DAG, where each node in the DAG corresponds to a
subtree\footnote{We use \name{subtree} with the meaning common in computer
science and matching the notion of \name{subterm}: A subtree of a tree $T$ is
a tree consisting of a node in $T$ and all of its descendants in $T$.} of a
tree in the given set. It is well known that there is a unique (modulo
isomorphism) \name{minimal} such DAG, which is maximally factored (it has no
multiple occurrences of the same subtree) or, equivalently, is minimal with
respect to the number of nodes, and, moreover, can be computed in linear time
\cite{downey:variations:1980}. The number of nodes of the minimal DAG is the
number of distinct subtrees of the members of the set of trees. This can be
used as the basis for proof size measures defined as follows.
\begin{defn}%
  \ \par

  \subdef{def-subq} For \Dterms~$d$ define $\DSUBQ(d)\; \eqdef\; \{\D(e_1,e_2)
  \mid d \suptermq \D(e_1,e_2)\}$.
  
  \subdef{def-csize} For \Dterms~$d$ define the \defname{compacted size of
  $d$} as $\csize(d)\; \eqdef\; |\DSUBQ(d)|.$

  \subdef{def-csize-set} For finite sets~$D$ of \Dterms define
  the \defname{compacted size} of $D$ as $\csize(D)\; \eqdef\; |\bigcup_{d \in D}
  \DSUBQ(d)|$.
\end{defn}

\noindent
$\DSUBQ(d)$ denotes the set of all compound subterms of a \Dterm~$d$. The
compacted size\footnote{We took the name \name{compacted size} from
Flajolet, Sipala and Steyaert \cite{flajolet:1990}.} of a \dterm, called \name{length} by Veroff
\cite{veroff:shortest:2001}, is the number of its distinct compound subterms,
reflecting the view that the size of the proof of a lemma is only counted
once, even if the lemma is used multiple times in the proof. It can equivalently be
characterized as the number of the inner nodes of its minimal DAG.

\begin{figure}
  \centering
  \scalebox{0.8}{\begin{tikzpicture}
  [baseline=(current bounding box.north),
    >={Latex[length=4.5pt]},
    dot/.style = {circle, minimum size=11pt,
      inner sep=0pt, outer sep=0pt, draw, on grid},
    sqr/.style = {regular polygon sides=4, minimum size=11pt,
      inner sep=0pt, outer sep=0pt, draw, on grid}]  
\node[dot] (e) {\small $$};
\node[dot, below left=1cm and \tl of e] (1) {\small $$};
\node[dot, below right=1cm and \tl of e] (2) {\small $$};
\node[sqr, below left=1cm and \tll of 1] (11) {\small $1$};
\node[dot, below left=1cm and \tll of 2] (21) {\small $$};
\draw[->] (e) -- (1);
\draw[->] (e) -- (2);
\draw[->] (1) -- (11);
\draw[->] (1) to [out=-40,in=15] (11);
\draw[->] (2) -- (21);
\draw[->] (2) to [out=-40,in=15] (21);
\draw[->] (21) to [out=240,in=-22] (11);
\draw[->] (21) to [out=-60,in=-10,out looseness=1.5] (1);
\end{tikzpicture}
  }
  \vspace{-8pt}
  \caption{The minimal DAG of the \Dterm~$d$ from Fig.~\ref{fig-dterm-tree}
    and Examples~\ref{examp-dterm} and~\ref{examp-dag}.}
  \label{fig-dterm-dag-plain}
\end{figure}

\begin{examp}%
  \label{examp-dag}
  Consider the \dterm
  \[d\; \eeqdef\; \D(\D(1,1),\D(\D(1,\D(1,1)),\D(1,\D(1,1))))\]
  from Example~\ref{examp-dterm}. Its compacted size is $\csize(d) = 4$. This
  is the number of inner nodes of the minimal DAG of~$d$, which is shown in
  Fig.~\ref{fig-dterm-dag-plain} (which is identical to
  Fig.~\ref{fig-representations}e after removing all labels with exception of
  the leaf label), or, equivalently, the cardinality of the set \[\DSUBQ(d) =
  \{\D(1,1),\; \D(1,\D(1,1)),\; \D(\D(1,\D(1,1)),\D(1,\D(1,1))),\; d\}\] of
  compound subterms of~$d$.
\end{examp}

A textual representation of \Dterms that respects the compacted size, that is,
is at most linearly larger than the compacted size, is possible by introducing
labels and references for subterms with multiple occurrences, which can be
done with a variety of concrete mechanisms. Our approach is to extend the set
of primitive \Dterms with labels used for referencing subproofs.
Formally, we view a \name{compacted \Dterm} as
a special kind of substitution whose domain members are primitive
\Dterms. Written out as a set of bindings, as common for substitutions, a
compacted \Dterm provides the desired compact textual representation of a set
of \Dterms.

\vspace{-3pt} %
\begin{defn}%
  \label{def:ldt}
  \

  \subdef{def-cdterm} A \defname{\compDterm} is a mapping~$\compd$ whose domain is a
  finite set of primitive \Dterms and whose range is a set of compound \Dterms
  such that the relation~$<_{\compd}$, called \defname{label dependency ordering}
  of $\compd$, defined as the transitive closure of $\{\la l, l'\ra \mid l,l' \in
  \DPRIMSET \text{ and } l \in \DPRIM(l'\compd)\}$ is a strict partial order.

  \subdef{def-roots} The \defname{roots} of a \compDterm~$\compd$ are the elements
  of $\dom(\compd)$ that are maximal with respect to $<_{\compd}$.

  \subdef{def-expand}  The binary function $\expand$ from \compDterms~$\compd$
  and primitive \Dterms~$l \in \dom(\compd)$ to \Dterms is defined as
  $\expand_{\compd}(l)\; \eqdef\; l\compd\{l_n \mapsto l_n\compd\}\{l_{n-1} \mapsto
  l_{n-1}\compd\}\ldots\{l_1 \mapsto l_1\compd\}$, where $l_1,l_2,\ldots,l_n$ is
  some $<_{\compd}$-linearization of the set~$\{l' \in \dom(\compd) \mid l' <_{\compd}
  l\}$.
\end{defn}
\vspace{-3pt} %

We write the application of a \compDterm (a special kind of substitution) in
postfix notation. A \compDterm~$\compd$ represents the finite \emph{set} of
\dterms, or trees, that correspond to its roots, that is,
$\{\expand_{\compd}(l) \mid l \text{ is a root of } \compd\}$. If~$\compd$ has
a single root~$l$, we also say that it represents the
\Dterm~$\expand_{\compd}(l)$.

\begin{examp}%
  \label{examp-cdterm}
   The \dterm $d$ from Examples~\ref{examp-dterm} and~\ref{examp-dag} is
   represented by the \compDterm
  \[\begin{array}{l@{\hspace{0.5em}}l@{\hspace{0.5em}}c@{\hspace{0.5em}}l}
  \compd\; \eeqdef\; \{ & 2 & \mapsto & \D(1,1),\\
  & 3 & \mapsto & \D(1,2),\\
  & 4 & \mapsto & \D(2,\D(3,3))\;\}.
  \end{array}\]
  The label dependency ordering $<_\compd$ can be described as $1 <_{\compd} 2
  <_{\compd} 3 <_{\compd} 4$ and $\compd$ has a single root, namely $4$.
\end{examp}

\begin{examp}%
  \label{examp-cdterm-mer}
  Consider Meredith's proof shown in Fig.~\ref{fig-proof-mer}. Its structure
  can be represented by the \compDterm $\compd_{\f{mer}} \eeqdef \{i \mapsto d_i
  \mid 2 \leq i \leq 19\}$ where $d_i$ is the \dterm representation of the
  proof term in line~$i$. Thus, $\compd_{\f{mer}} = \{2 \mapsto
  \D(\D(\D(1,\D(1,1)),1),n),\; 3\mapsto \D(\D(\D(1,\D(1,\D(1,2))),1),n),\;
  4\mapsto \D(3,1),\; \ldots,\; 19\mapsto \D(3,3)\}$. The label dependency
  ordering is visualized in Fig.~\ref{fig-mercd-ordering}. The
  \compDterm~$\compd_{\f{mer}}$ has three roots, $17$, $18$ and $19$. Meredith's
  representation of the proof structure can be reconstructed in full as a
  linearization of the label dependency ordering from the \compDterm
  $\compd_{\f{mer}}$.
\end{examp}

\begin{figure}
  \centering
  \scalebox{0.9}{
  \begin{tikzpicture}[>=latex',line join=bevel,scale=0.65]
      \pgfsetlinewidth{1bp}
\pgfsetcolor{black}
  \draw [] (29.891bp,64.0bp) .. controls (27.385bp,64.0bp) and (24.698bp,64.0bp)  .. (22.188bp,64.0bp);
  \draw [] (59.891bp,64.0bp) .. controls (57.385bp,64.0bp) and (54.698bp,64.0bp)  .. (52.188bp,64.0bp);
  \draw [] (89.891bp,81.252bp) .. controls (87.385bp,78.836bp) and (84.698bp,76.244bp)  .. (82.188bp,73.824bp);
  \draw [] (119.89bp,91.0bp) .. controls (117.39bp,91.0bp) and (114.7bp,91.0bp)  .. (112.19bp,91.0bp);
  \draw [] (149.89bp,91.0bp) .. controls (147.39bp,91.0bp) and (144.7bp,91.0bp)  .. (142.19bp,91.0bp);
  \draw [] (179.89bp,91.0bp) .. controls (177.39bp,91.0bp) and (174.7bp,91.0bp)  .. (172.19bp,91.0bp);
  \draw [] (209.89bp,91.0bp) .. controls (207.39bp,91.0bp) and (204.7bp,91.0bp)  .. (202.19bp,91.0bp);
  \draw [] (269.9bp,68.703bp) .. controls (259.38bp,73.598bp) and (242.99bp,81.227bp)  .. (232.38bp,86.169bp);
  \draw [] (239.89bp,106.97bp) .. controls (237.39bp,104.74bp) and (234.7bp,102.34bp)  .. (232.19bp,100.1bp);
  \draw [] (269.89bp,117.28bp) .. controls (267.39bp,117.1bp) and (264.7bp,116.91bp)  .. (262.19bp,116.73bp);
  \draw [] (299.89bp,118.0bp) .. controls (297.39bp,118.0bp) and (294.7bp,118.0bp)  .. (292.19bp,118.0bp);
  \draw [] (329.89bp,118.0bp) .. controls (327.39bp,118.0bp) and (324.7bp,118.0bp)  .. (322.19bp,118.0bp);
  \draw [] (329.9bp,64.0bp) .. controls (319.38bp,64.0bp) and (302.99bp,64.0bp)  .. (292.38bp,64.0bp);
  \draw [] (330.52bp,82.282bp) .. controls (327.6bp,87.912bp) and (324.39bp,94.108bp)  .. (321.47bp,99.737bp);
  \draw [] (359.89bp,69.112bp) .. controls (357.39bp,68.396bp) and (354.7bp,67.628bp)  .. (352.19bp,66.911bp);
  \draw [] (389.89bp,72.0bp) .. controls (387.39bp,72.0bp) and (384.7bp,72.0bp)  .. (382.19bp,72.0bp);
  \draw [] (419.91bp,97.638bp) .. controls (403.02bp,102.05bp) and (369.23bp,110.88bp)  .. (352.25bp,115.32bp);
  \draw [] (419.89bp,86.696bp) .. controls (417.39bp,84.638bp) and (414.7bp,82.43bp)  .. (412.19bp,80.369bp);
  \draw [] (359.89bp,34.607bp) .. controls (357.39bp,38.725bp) and (354.7bp,43.139bp)  .. (352.19bp,47.263bp);
  \draw [] (89.891bp,46.748bp) .. controls (87.385bp,49.164bp) and (84.698bp,51.756bp)  .. (82.188bp,54.176bp);
\begin{scope}
  \definecolor{strokecol}{rgb}{0.0,0.0,0.0}
  \pgfsetstrokecolor{strokecol}
  \draw (11.0bp,64.0bp) node {1};
\end{scope}
\begin{scope}
  \definecolor{strokecol}{rgb}{0.0,0.0,0.0}
  \pgfsetstrokecolor{strokecol}
  \draw (41.0bp,64.0bp) node {2};
\end{scope}
\begin{scope}
  \definecolor{strokecol}{rgb}{0.0,0.0,0.0}
  \pgfsetstrokecolor{strokecol}
  \draw (71.0bp,64.0bp) node {3};
\end{scope}
\begin{scope}
  \definecolor{strokecol}{rgb}{0.0,0.0,0.0}
  \pgfsetstrokecolor{strokecol}
  \draw (101.0bp,91.0bp) node {4};
\end{scope}
\begin{scope}
  \definecolor{strokecol}{rgb}{0.0,0.0,0.0}
  \pgfsetstrokecolor{strokecol}
  \draw (131.0bp,91.0bp) node {5};
\end{scope}
\begin{scope}
  \definecolor{strokecol}{rgb}{0.0,0.0,0.0}
  \pgfsetstrokecolor{strokecol}
  \draw (161.0bp,91.0bp) node {6};
\end{scope}
\begin{scope}
  \definecolor{strokecol}{rgb}{0.0,0.0,0.0}
  \pgfsetstrokecolor{strokecol}
  \draw (191.0bp,91.0bp) node {7};
\end{scope}
\begin{scope}
  \definecolor{strokecol}{rgb}{0.0,0.0,0.0}
  \pgfsetstrokecolor{strokecol}
  \draw (221.0bp,91.0bp) node {8};
\end{scope}
\begin{scope}
  \definecolor{strokecol}{rgb}{0.0,0.0,0.0}
  \pgfsetstrokecolor{strokecol}
  \draw (281.0bp,64.0bp) node {9};
\end{scope}
\begin{scope}
  \definecolor{strokecol}{rgb}{0.0,0.0,0.0}
  \pgfsetstrokecolor{strokecol}
  \draw (251.0bp,116.0bp) node {10};
\end{scope}
\begin{scope}
  \definecolor{strokecol}{rgb}{0.0,0.0,0.0}
  \pgfsetstrokecolor{strokecol}
  \draw (281.0bp,118.0bp) node {11};
\end{scope}
\begin{scope}
  \definecolor{strokecol}{rgb}{0.0,0.0,0.0}
  \pgfsetstrokecolor{strokecol}
  \draw (311.0bp,118.0bp) node {12};
\end{scope}
\begin{scope}
  \definecolor{strokecol}{rgb}{0.0,0.0,0.0}
  \pgfsetstrokecolor{strokecol}
  \draw (341.0bp,118.0bp) node {13};
\end{scope}
\begin{scope}
  \definecolor{strokecol}{rgb}{0.0,0.0,0.0}
  \pgfsetstrokecolor{strokecol}
  \draw (341.0bp,64.0bp) node {14};
\end{scope}
\begin{scope}
  \definecolor{strokecol}{rgb}{0.0,0.0,0.0}
  \pgfsetstrokecolor{strokecol}
  \draw (371.0bp,72.0bp) node {15};
\end{scope}
\begin{scope}
  \definecolor{strokecol}{rgb}{0.0,0.0,0.0}
  \pgfsetstrokecolor{strokecol}
  \draw (401.0bp,72.0bp) node {16};
\end{scope}
\begin{scope}
  \definecolor{strokecol}{rgb}{0.0,0.0,0.0}
  \pgfsetstrokecolor{strokecol}
  \draw (431.0bp,95.0bp) node {17};
\end{scope}
\begin{scope}
  \definecolor{strokecol}{rgb}{0.0,0.0,0.0}
  \pgfsetstrokecolor{strokecol}
  \draw (371.0bp,18.0bp) node {18};
\end{scope}
\begin{scope}
  \definecolor{strokecol}{rgb}{0.0,0.0,0.0}
  \pgfsetstrokecolor{strokecol}
  \draw (101.0bp,37.0bp) node {19};
\end{scope}
  \end{tikzpicture}}%
  \vspace{2pt}
  \caption{The label dependency ordering $<_{\compd}$ of Meredith's proof
    (Fig.~\ref{fig-proof-mer}) when viewed as a \compDterm according to
    Example~\ref{examp-cdterm-mer}.}
  \label{fig-mercd-ordering}
\end{figure}

\vspace{-2pt} %
A \compDterm directly represents a DAG: The \defname{DAG of} a \compDterm $\compd =
\{l_1 \mapsto d_1, \ldots, l_n \mapsto d_n\}$ is the graph obtained from the
trees $d_1, \ldots, d_n$ by considering any edge to a leaf labeled with~$l_i$
as an edge to the root of $d_i$, and any edge to a leaf labeled with a symbol
not in $\dom(\compd)$ as an edge to a unique node representing that
symbol in the DAG. Figure~\ref{fig-dterm-dag-numbered} shows an example. The
DAGs of \compDterms inherit from \dterms, full binary trees, the property that
each inner node has exactly two children, a left and a right
child.\footnote{Properties of such binary DAGs for the special case of a
single root and a single leaf have been recently investigated by Genitrini et al.
\cite{genitrini:2020}, where they are called \name{compacted trees}.}

\begin{figure}
  \centering
  \scalebox{0.8}{  \begin{tikzpicture}
  [baseline=(current bounding box.north),
    >={Latex[length=4.5pt]},
    dot/.style = {circle, minimum size=11pt,
      inner sep=0pt, outer sep=0pt, draw, on grid},
    sqr/.style = {regular polygon sides=4, minimum size=11pt,
      inner sep=0pt, outer sep=0pt, draw, on grid}]  
\node[dot] (e) {\small $4$};
\node[dot, below left=1cm and \tl of e] (1) {\small $2$};
\node[dot, below right=1cm and \tl of e] (2) {\small $$};
\node[sqr, below left=1cm and \tll of 1] (11) {\small $1$};
\node[dot, below left=1cm and \tll of 2] (21) {\small $3$};
\draw[->] (e) -- (1);
\draw[->] (e) -- (2);
\draw[->] (1) -- (11);
\draw[->] (1) to [out=-40,in=15] (11);
\draw[->] (2) -- (21);
\draw[->] (2) to [out=-40,in=15] (21);
\draw[->] (21) to [out=240,in=-22] (11);
\draw[->] (21) to [out=-60,in=-10,out looseness=1.5] (1);
\end{tikzpicture}
  }
  \vspace{-8pt}
  \caption{The DAG of the \compDterm $\compd$ from Example~\ref{examp-cdterm}.
    It is identical to Fig.~\ref{fig-dterm-dag-plain}, except that labels
    of inner nodes indicate the corresponding primitive \Dterms in the domain
    of $\compd$.}
  \label{fig-dterm-dag-numbered}
\end{figure}

The number of inner nodes of the DAG of a \compDterm is $\sum_{\,l \in
  \dom(\compd)} \tsize(l\compd)$. If the \compDterm is written as a set of
bindings as in Example~\ref{examp-cdterm}, it can be read off as the total
number of occurrences of $\D$ in the bindings' right sides.

An alternative possible technical understanding of a \compDterm with a single
root is as a regular tree grammar where the domain forms the set of
nonterminals. Each nonterminal there has exactly one production and the
grammar describes a single tree
\cite{lohrey:survey:2015,lohrey:treerepair:2013}. If the regularity condition
is dropped, the grammar framework generalizes to tree representations that are
more strongly compressed than DAGs, offering further compression possibilities
also for \Dterms \cite{cw:ccs}.

\subsubsection{Comparing the Number of \Dterms of a Given Size for Different Size Measures}

The number of distinct \Dterms for increasing values of some size measure like
tree size, height or compacted size, gives an upper bound of the number of
trees to consider in proof search by enumerating \Dterms with iterative
deepening upon that size measure. This number is just an upper bound of the
actual structures to consider, because it does not take into account that
\Dterm enumeration may be interwoven with unification constrained by given
axioms and possibly a given goal where fragments of \Dterms for which
unifiability fails are immediately discarded. Heuristic restrictions may in
practice further restrict the considered number of structures. The number of
distinct \Dterms for increasing values of a size measure also indicates a
measure-specific size value up to which it is easily possible to compute for
given axioms all proofs, together with the lemmas proven by them.

If we assume a single proper axiom such that we can identify \compDterms with
full binary trees without any additional labeling, the sequences of the number
of distinct \Dterms for increasing tree size, height or compacted size are
well-known and can be found in \name{The {O}n-{L}ine {E}ncyclopedia of
  {I}nteger {S}equences} \url{https://oeis.org/} \cite{oeis}, with identifiers
\href{https://oeis.org/A000108}{A000108},
\href{https://oeis.org/A001699}{A001699}, and
\href{https://oeis.org/A254789}{A254789}, respectively, as shown in
Table~\ref{tab-oeis-numbers}.

\begin{table}[t]
  \centering
  \caption{The numbers of distinct \Dterms for a single axiom (or full binary
    trees) of given size~$n$ for different size measures.}
  \label{tab-oeis-numbers}
  \begin{tabular}{llcrrrrrrr}
    $n$ && 0 & 1 & 2 & 3 & 4 & 5 & 6\\\midrule
      Tree size & \OEISNUM{A000108} &
    1 & 1 & 2 & 5 & 14 & 42 & 132\\
    Height & \OEISNUM{A001699}
    & 1 & 1
    & 3 & 21 & 651 & 457,653 & 210,065,930,571\\
    Compacted size & \OEISNUM{A254789}
    & 1 & 1 &
    3 & 15 & 111 & 1,119 & 14,487
  \end{tabular}
\end{table}

\subsubsection{Node Labels for Proof Modularization}
\label{subsubsec-modularization}

That a \compDterm $\compd$ represents a set $D = \{d_1, \ldots, d_n\}$ of
\Dterms does not imply that the DAG of $\compd$ is the \emph{minimal} DAG
corresponding to~$D$. If the number of inner nodes of the DAG is larger than
the compacted size of~$D$, this indicates that not all subtrees of $D$ with multiple occurrences
have properly been factored out in~$\compd$. Although
obviously burdened with redundancy, such non-minimal DAGs cannot be excluded
from the outset because automated theorem provers might produce them, as in
general they do not always detect different subproof occurrences with
identical structure.

A \compDterm comprises not just the representation of a DAG, but also a
labeling of \emph{some} of its inner nodes. Nodes that receive such a label
include in particular all root nodes and all nodes that have more than a
single incoming edge. Figure~\ref{fig-dterm-dag-numbered} shows these
labelings for the \dterm $\compd$ of Example~\ref{examp-cdterm}: The unlabeled
blank node corresponds to the subtree $\D(3,3)$ which has only a single
incoming edge. In addition to labels that are necessary to describe the
structure, a \compDterm can provide labels for further nodes. In other words,
its domain may include primitive \Dterms that are neither a root nor occur
``multiple times'' in its range, where occurring ``multiple times'' in the range means
occurring in different members of the range or with multiple occurrences in
some member of the range.

\vspace{-4pt}

\begin{examp}
  Consider the \compDterm~$\compd$ of Example~\ref{examp-cdterm}, whose DAG is
  shown in Fig.~\ref{fig-representations}e and which represents the \Dterm~$d$
  from Examples~\ref{examp-dterm} and~\ref{examp-dag}. The root of the
  following \compDterm~$\compd'$ represents the same \Dterm as~$\compd$ and
  has the same number of inner nodes, but has with $3'$ one more primitive
  \Dterm in its domain, which it maps to the subterm $\D(3,3)$ and which has
  just a single occurrence in its range. This occurrence is in $\D(2,3')$,
  which is the value of $4\compd'$.
  \[\begin{array}{l@{\hspace{0.5em}}l@{\hspace{0.5em}}c@{\hspace{0.5em}}l}
  \compd'\; \eeqdef\; \{
  & 2 & \mapsto & \D(1,1),\\
  & 3 & \mapsto & \D(1,2),\\
  & 3' & \mapsto & \D(3,3),\\
  & 4 & \mapsto & \D(2,3')\;\}.
  \end{array}\]
\end{examp}

\vspace{-16pt}

\begin{examp}
  \label{examp-minimal-mer}
  The compacted \dterm~$\compd_{\f{mer}}$ from Example~\ref{examp-cdterm-mer},
  which represents the structure of Meredith's proof from
  Fig.~\ref{fig-proof-mer}, is a \compDterm where not all non-root members of the
  domain occur multiple times in the range, which is not difficult but somewhat
  tedious to verify: The primitive \Dterms $2$, $7$, $11$ and $15$ each have
  only a single occurrence in the range of $\compd_{\f{mer}}$.
\end{examp}

\vspace{-4pt}

Such node labels or domain members of a \compDterm, which are superfluous for
the purpose of describing the proof structure, can nevertheless be meaningful
for the presentation of a proof, because they indicate a modularization into
subproofs that is motivated by other reasons than the multiple occurrence of a
subproof or multiple use of a lemma. For example, to exhibit a subproof
obtained with a specific inference technique or to explicitly show the lemma
proven by a subproof as an intermediate proof stage for better comprehension
by humans.

\subsection{Proof Structures, Formula Substitutions and Semantics}
\label{subsec-subst-base}

A \CD proof combines structural aspects represented by a \Dterm, a full binary
tree, with atomic formulas associated with its nodes. Similar to a CM proof
of a clausal formula, a \CD proof involves different instances of the input
clauses, specifically the proper axioms and the detachment axiom \Det. The
atomic formulas associated with nodes are induced through unification from the
axioms and, via instances of \Det, the tree structure of the \Dterm. The
atomic formula associated with the root of the tree is the ``most general''
formula proven by the \Dterm with respect to the given proper axioms. In
particular, it proves all ground formulas that are instances of it and are
obtained from Skolemizing a universally quantified goal formula. For
goal-driven proof search, such a ground formula is taken into account from the
beginning, such that fragments of \Dterms whose root-associated formula does
not subsume the goal can be excluded early through failure of unification.

We call the most general formula proven by a \Dterm with respect to given
proper axioms the \name{most general theorem (MGT)} of the \Dterm. The MGT of
a subproof $d|_p$ of a proof $d$ represents the lemma used in $d$ at position
$p$. This MGT is determined just by the subproof $d|_p$ and the proper axioms.
Thus, occurrences of the same subproof at other positions in $d$ have the same
MGT.
There is a second useful way to associate formulas with positions in a \Dterm,
the \name{in-place theorem (IPT)} of a \Dterm $d$ at position $p$, which
represents the actual
\emph{instance} of the lemma used in $d$ at position $p$. Like the MGT, the IPT
is determined through most general unification but, in addition to the subtree
$d|_p$, also with respect to the context of $p$ in $d$. The notions of MGT and
IPT as well as their interplay will be made precise in this subsection.

\subsubsection{Most General Unifiers}

\CD involves the implicit characterization of proven lemmas as formulas that
are \emph{most general} in a certain sense, which can be specified with the
notion of most general unifier, a standard concept in modern ATP. We use it
here in a version that applies to a set of pairs of terms, as convenient in
discussions based on the CM
\cite{bibel:atp:1982,eder:subst:1985,eder:relative:1992}, and assume useful
restricting properties suggested by Elmar Eder \cite{eder:subst:1985}, gathered here
under the label \name{clean}.

\begin{defn}%
  \label{def-mgu-main}%
  Let~$M$ be a set of pairs of terms and let $\sigma$ be a substitution.
  
  \subdef{def-unifier} $\sigma$ is said to be a \defname{unifier} of $M$ if for
  all $\{s,t\} \in M$ it holds that $s\sigma = t\sigma$.

  \subdef{def-most-general-unifier} $\sigma$ is called a \defname{most general
    unifier} of $M$ if $\sigma$ is a unifier of $M$ and for all
  unifiers~$\sigma'$ of $M$ it holds that $\sigma' \subsumedBy \sigma$.

  \subdef{def-clean} $\sigma$ is called a \defname{clean most general unifier}
  of $M$ if it is a most general unifier of $M$ and, in addition, is
  idempotent and satisfies $\dom(\sigma) \cup \vrng(\sigma) \subseteq
  \vars(M)$.

  \subdef{def-mgu} If $M$ has a unifier, then $\mgu(M)$ denotes some clean
  most general unifier of $M$. $M$ is called \defname{unifiable} and $\mgu(M)$
  is called \defname{defined} in this case, otherwise it is called
  \defname{undefined}.

\end{defn}

\begin{convention}
  \label{convention-mgu-partial} Proposition,
  lemma and theorem statements implicitly assert their claims only for the
  case where occurrences of $\mgu$ in them are defined.
\end{convention}

\noindent
Although a unifier of a finite set of
pairs~$\{\{s_1,t_1\},\ldots,\{s_n,t_n\}\}$ can be expressed as unifier of the
single pair $\{\f{f}(s_1,\ldots,s_n), \f{f}(t_1,\ldots,t_n)\}$, the explicit
definition for a set of pairs fits well with the graphs in the CM and the
related \Dterms, where such sets of pairs naturally arise.

The additional properties required for a \name{clean} most general unifier do
not hold for all most general unifiers.\footnote{The inaccuracy observed by
Hindley and David Meredith
\cite{hindley:meredith:cd:1990} in early formalizations of \CD based on the
notion of \name{most general unifier} can be attributed to disregarding the
requirement $\dom(\sigma) \cup \vrng(\sigma) \subseteq \vars(M)$ of the
\name{clean} property.} However, the unification algorithms known from the
literature produce \emph{clean} most general unifiers \cite[Remark
  4.2]{eder:subst:1985}. If a set of pairs of terms has a unifier, then it has
a most general unifier and, moreover, also a \emph{clean} most general
unifier. Since we define $\mgu(M)$ as a \emph{clean} most general unifier,
whenever necessary, we can assume that it is idempotent and that all variables
occurring in its domain and range occur in $M$.
Convention~\ref{convention-mgu-partial} has the purpose to reduce clutter in
proposition, lemma and theorem statements.

\subsubsection{Positional Variables}
\label{subsubsec-posvars}

The atomic formulas associated with the nodes of a \Dterm are based on
instances of the proper axioms and \Det, which may conceptually be considered
as obtained in two steps: first, ``copies'', i.e., variants with fresh
variables, are created; second, a substitution determined by the proof
structure is applied to these copies. Let us begin with describing the first
step. We only need formulas with specific variables, which we call
\defname{\positionalvars} because each of them is firmly tied to a term
position. They are defined as follows.
\begin{defn}%
  \ \\
  \subdef{def-vpos} For all positions~$p$ and positive integers $i$ let
  $\vx{i}{p}$ and $\vy_p$ denote pairwise different variables. We call the
  variables $\vx{i}{p}$ and $\vy_p$ \defname{positional variables}.

  \subdef{def-vpos-p} For all sets~$P$ of positions define
  \[\posvar(P)\; \eqdef\; \{\vy_p \mid p \in P\}
  \cup \{\vx{i}{p} \mid p \in P \text{ and } i \geq 1\}.\]
\end{defn}

\noindent
With each leaf of a \dterm~$d$ a dedicated copy of some proper axiom is
associated. The variables $\vx{i}{p}$ are for use in these copies, where the
subscript $p$ is the position of the leaf node in $d$. The upper index~$i$
serves to distinguish different variables within the copies, as indicated with
the right side of the following exemplary equivalence, which holds for all
positions~$p$.
\begin{equation}
\label{eq-luk-canonical}
  \Luk \;\equiv\; \UC{
\P(\i(\i(\i(\vx{1}{p},\vx{2}{p}),\vx{3}{p}),
\i(\i(\vx{3}{p},\vx{1}{p}),\i(\vx{4}{p},\vx{1}{p}))))}.
\end{equation}
A variable $\vy_p$ can be associated with each position~$p$ of a \Dterm. That
each inner node of a \Dterm corresponds to a dedicated copy of the $\Det$
axiom is reflected in the following equivalence, which holds for all
positions~$p$.
\begin{equation}
  \label{eq-det-canonical}
\Det \;\equiv\; \UC{(\P(\i(\vy_{p.2},\vy_{p})) \land \P(\vy_{p.2})
    \imp \P(\vy_p))}. 
\end{equation}
Here the major premise of \Det is written to the left of the minor one,
matching the argument order of the $\D$ function symbol for proof tree
construction.
$\posvar(P)$ provides notation for referring to all positional variables
associated with members of a given set~$P$ of positions, regardless of whether
they are of the form $\vy_p$ or $\vx{i}{p}$.

The following substitution $\sshift{p}$ is a tool to systematically rename
\positionalvars while preserving the internal relationships between the
index-referenced positions.
\begin{defn}%
  For all positions~$p$ define the substitution $\sshift{p}$ as
  \[\begin{array}{r@{\hspace{0.5em}}c@{\hspace{0.5em}}l}
  \sshift{p} & \eqdef &
  \{\vy_{q} \mapsto \vy_{p.q} \mid q \text{ is a position}\}\; \cup\\[2pt]
  && \{\vx{i}{q} \mapsto \vx{i}{p.q} \mid i \geq 1 \text{ and } q \text{ is a
    position}\}.
  \end{array}
  \]
  \vspace{1pt}
\end{defn}

\noindent
The application of $\sshift{p}$ to a term~$s$ effects that $p$ is prepended to
the position indexes of all the \positionalvars occurring
in~$s$.

\begin{examp}%
  \[\begin{array}{l@{\hspace{0.5em}}c@{\hspace{0.5em}}l}
  \i(\vx{1}{\emptypos},\vx{2}{\emptypos})\sshift{1.1.2.1} & = &
  \i(\vx{1}{1.1.2.1},\vx{2}{1.1.2.1}).\\[3pt]
  \i(\vy_{2.1},\vy_{2.1.2}))\sshift{1.1} & = &
  \i(\vy_{1.1.2.1},\vy_{1.1.2.1.2}).
  \end{array}\]
  In the second equality, observe that position~$2.1.2$ refers to the right
  child of position~$2.1$. After applying $\sshift{1.1}$, it is
  position~$1.1.2.1.2$ that, again, refers to the right child of
  position~$1.1.2.1$.
\end{examp}

\noindent
Applying a $\sshift{p}$ substitution to a term always yields a variant, as
stated in the following proposition.
\begin{prop}%
  \label{prop-shift-variant}
  For all terms~$s$ whose variables are positional variables
  (Definition~\ref{def-vpos}) and for all positions~$p$ it holds that
  \[s\; \variant\; s\sshift{p}.\]
\end{prop}
\begin{proof}
  Easy to see.
\end{proof}

\subsubsection{Axiom Assignments}

The association of axioms with primitive \Dterms is represented by a mapping
which we call \name{axiom assignment}, defined as follows.
\begin{defn}%
  An \defname{axiom assignment} $\alpha$ is a mapping whose domain is a set
  of primitive \Dterms and whose range is a set of terms whose variables are
  in $\{\vx{i}{\emptypos} \mid i \geq 1\}$. We say that $\alpha$ is
  \defname{for} a \dterm~$d$ if $\dom(\alpha) \supseteq \DPRIM(d)$.
\end{defn}
We write the application of an axiom assignment in postfix notation.

\begin{examp}%
  \label{examp-axiom-assignment}
The mapping
\[\alpha\;\eqdef\;\{1 \mapsto
\i(\i(\i(\vx{1}{\emptypos},\vx{2}{\emptypos}),\vx{3}{\emptypos}),\i(\i(\vx{3}{\emptypos},\vx{1}{\emptypos}),\i(\vx{4}{\emptypos},\vx{1}{\emptypos})))\}\]
is an axiom assignment for all \dterms~$d$ with $\DPRIM(d) = \{1\}$. Its sole
range element is a variant of the argument term of \Luk in the form of the
right side of (\ref{eq-luk-canonical}), with $p$ instantiated to the empty
position $\emptypos$. The application of $\alpha$ to the primitive \Dterm $1$
is written in postfix notation as $1\alpha$.
\end{examp}

\begin{examp}
In Meredith's proof presentation the axiom assignment is represented by the
steps with no trailing \dterm, such as line~1 in
Fig.~\ref{fig-representations}c, or line~1 in Fig.~\ref{fig-proof-mer}. The
latter actually represents the same axiom assignment as
Example~\ref{examp-axiom-assignment}.
\end{examp}

\subsubsection{Pairings}

As noted in the beginning of Sect.~\ref{subsubsec-posvars}, the clause
instances involved in a \CD proof may, similarly as in the CM, conceptually be
considered as obtained in two steps. We now turn to the second step, the
application of a substitution determined by the proof structure to the
previously created clause copies. This substitution is characterized as the
most general unifier of a set of term pairs that contains exactly one pair for
each node, or position, of the \Dterm. The following definition specifies this
pair for a given position.
\begin{defn}%
  \label{def-pairing}
  For \dterms~$d$, axiom assignments $\alpha$ and positions $p \in \pos(d)$
  define the pair $\pairing_{\alpha}(d, p)$ of terms as
  \[
  \begin{array}{rcll}
    \pairing_{\alpha}(d, p) & \eqdef & \{\vy_p,\, d|_p\alpha\sshift{p}\} &
    \text{ if } p \in
    \leafpos(d)\\[3pt]
    && \{\vy_{p.1},\, \i(\vy_{p.2}, \vy_p)\} & \text{ if } p \in \innerpos(d).
  \end{array}
  \]
\end{defn}

\noindent
A unifier of the set of pairings of all positions of a \dterm~$d$ equates for
each leaf position~$p$ the variable $y_p$ with the value of the axiom
assignment~$\alpha$ for the primitive \Dterm at $p$, after ``shifting''
variables by~$p$. This ``shifting'' means that the position subscript
$\emptypos$ of the variables in~$d|_p\alpha$ is replaced by $p$, yielding a
dedicated copy of the axiom for the leaf position~$p$. For inner
positions~$p$, which represent detachment steps, the unifier equates
$\vy_{p.1}$ and $\i(\vy_{p.2}, \vy_p)$, reflecting that the major premise of
\Det is proven by the left child of~$p$.
With respect to the connections shown for the case of a single axiom in
Fig.~\ref{fig-ConnT}, the pairing $\{\vy_{p.1},\, \i(\vy_{p.2}, \vy_p)\}$ for
an inner position~$p$ is induced by connection~\textbf{2} or~\textbf{4},
respectively, depending on whether $\vy_{p.1}$ is an inner node or a leaf.
Connections~\textbf{3} and~\textbf{5} would just induce the void requirement
$\{\vy_{p.2}, \vy_{p.2}\}$ and thus have no explicit correspondent in the
specification of $\pairing$. An example of a set of pairings and its most
general unifier is included in Example~\ref{examp-ipt-mgt} below.

The following proposition shows an interplay of $\pairing$ and $\f{shift}$,
which is useful as a lemma in further derivations.
\begin{prop}%
  \label{prop-tree-renaming-add}
  For all \dterms~$d$, axiom assignments $\alpha$ for $d$ and positions $p
  \in \pos(d)$ it holds that
  \[\begin{array}[t]{c@{\hspace{0.5em}}l}
  & y_{\emptypos}\mgu(\{\pairing_{\alpha}(d|_p, q) \mid q \in \pos(d|_p)\})
  \sshift{p}\\
  = & y_{p}\mgu(\{\pairing_{\alpha}(d, q) \mid q \in \pos(d) \text{ and } p
  \leq q\}).
  \end{array}\]
\end{prop}
\begin{proof}
  Easy to see.
\end{proof}

\subsubsection{In-Place Theorem (IPT) and Most General Theorem (MGT)}

Based on the most general unifier of the set of pairings of all positions of a
\Dterm, a specific formula can be associated with each position of the \Dterm,
called the \name{in-place theorem (\IPT)}. The case where the position is the
top position~$\emptypos$ is distinguished as \name{most general theorem
  (\MGT)}.
\begin{defn}%
  \label{def-ipt-mgt}%
  For \dterms~$d$, axiom assignments $\alpha$ and positions $p \in \pos(d)$
  define the \defname{in-place theorem (\IPT) of}~$d$ \defname{at}~$p$
  \defname{for}~$\alpha$, $\ipt_{\alpha}(d,p)$, and the \defname{most general
    theorem (\MGT) of}~$d$ \defname{for}~$\alpha$, $\mgt_{\alpha}(d)$, as
  \smallskip\par
  \subdef{def-ipt} $\ipt_{\alpha}(d,p)\;
  \eqdef$ $\P(y_{p}\mgu(\{\pairing_{\alpha}(d, q) \mid q \in
  \pos(d)\})).$
  \smallskip\par
  \subdef{def-mgt} $\mgt_{\alpha}(d)\; \eqdef
  \ipt_{\alpha}(d,\emptypos)$.
\end{defn}

\noindent
Since $\ipt$ and $\mgt$ are defined on the basis of $\mgu$, they are undefined
if the set of pairs of terms underlying the respective application of $\mgu$
is not unifiable. Hence, we apply Convention~\ref{convention-mgu-partial} for
$\mgu$ also to occurrences of $\ipt$ and $\mgt$. If $\ipt$ and $\mgt$ are
defined, they both denote an atom whose variables are constrained by the
\name{clean} property of the underlying application of $\mgu$.

Let us illustrate the two formulas specified in Definition~\ref{def-ipt-mgt}
in a more informal way, beginning with the conceptually simpler MGT. We assume
that the axiom assignment $\alpha$ is $ \{1 \mapsto \Luk_{\emptypos}\}$, that
is, we have just a single proper axiom, \Luk, which is labeled by~$1$. The
argument $d$ of $\mgt$ is a \Dterm. If it is a \dconstant, that is, if $d=1$,
then $\mgt_{\alpha}(d)$ is just a variant of the axiom \Luk, corresponding to
the value of $1$ in the axiom assignment. Otherwise $d$ refers to some
instance of the detachment axiom $\P x \land \P\i xy \imp \P y$. If, for
example, $d = \D(1,1)$, then both premises of~$d$ are connected with two
different instances of the axiom \Luk resulting in a substitution $\sigma$ for
$x$ and $y$ such that $\mgt_\alpha(d)=\P y\sigma$. In other words, the
resulting MGT is the derived conclusion of the detachment axiom, applied
to two copies of the proper axiom as premises.

In the general case we have more instances of the detachment axiom and
more instances of the proper axiom involved; but the resulting MGT is still
the derived conclusion of the applications of the detachment axiom, one
application for each inner node of~$d$. In such a more general case, we could
be interested in the conclusion of some instance of the detachment axiom
\emph{within} the derivation other than the final one, say the one at
position~$p$. This situation is captured by the IPT, which renders exactly
such a conclusion formula. The substitution to obtain the IPT is induced not
only by the pairing constraints of the subtree rooted at position~$p$, but
also by the pairing constraints of its context in the overall proof.

In accounts of \CD in type theory
\cite{hindley:meredith:cd:1990,hindley:book:1997} the \MGT is considered as
\name{principal type-scheme} or \name{principal type}. A primitive \Dterm is
identified there with the associated axiom. A compound \Dterm $\D(d_1,d_2)$ is
identified with the principal type of the application of a function with
principal type $d_1$ to an argument with principal type $d_2$.

The following proposition relates \IPT and \MGT with respect to subsumption.
\begin{prop}%
 \label{prop-ipt-subsumedby-mgt}
  For all \dterms~$d$, axiom assignments $\alpha$ for $d$ and positions $p
  \in \pos(d)$ it holds that \[\ipt_{\alpha}(d,p) \subsumedBy
  \mgt_{\alpha}(d|_p).\]
\end{prop}
\begin{proof}
  \prlReset{prop-ipt-subsumedby-mgt} Can be shown in the following steps,
  explained below.
  \[
  \begin{array}{r@{\hspace{0.5em}}c@{\hspace{0.5em}}l@{\hspace{0.5em}}l}
  \prl{1} & & \ipt_{\alpha}(d,p)\\
  \prl{2} & = & \P(y_{p}\mgu(\{\pairing_{\alpha}(d, q) \mid q \in \pos(d)\}))\\
  \prl{3} & \subsumedBy & \P(y_{p}\mgu(\{\pairing_{\alpha}(d,q) \mid q \in \pos(d) \text{ and } p \leq q\}))\\
  \prl{4} & = & \P(y_{\emptypos}\mgu(\{\pairing_{\alpha}(d|_p, q) \mid q \in \pos(d|_p)\})\sshift{p})\\
  \prl{5} & \variant & \P(y_{\emptypos}\mgu(\{\pairing_{\alpha}(d|_p, q) \mid q \in \pos(d|_p)\}))\\
  \prl{6} & = & \ipt_{\alpha}(d|_p, \emptypos)\\
  \prl{7} & = & \mgt_{\alpha}(d|_p).\\
  \end{array}
  \]
   Step~\pref{3} follows easily from the definition of most general unifier.
   Step~\pref{4} is justified by Proposition~\ref{prop-tree-renaming-add},
   step~\pref{5} by Proposition~\ref{prop-shift-variant}.  The remaining steps are
   obtained by expanding and contracting definitions.
\end{proof}

\smallskip

\noindent
By Proposition~\ref{prop-ipt-subsumedby-mgt}, the \IPT at some position~$p$ of a
\dterm~$d$ is subsumed by the \MGT of the subterm~$d|_p$ of $d$ rooted at
position~$p$. An intuitive argument is that the only constraints that
determine the most general unifier underlying the \MGT are induced by
positions of $d|_p$, that is, \emph{below}~$p$ (including~$p$ itself). In
contrast, the most general unifier underlying the \IPT is determined by
\emph{all} positions of $d$, including those that are not below~$p$.

\begin{exampapp}%
  \label{examp-ipt-mgt}
   This example shows for a given \Dterm the set of associated pairings
   (Definition~\ref{def-pairing}) and its most general unifier
   (Definition~\ref{def-mgu-main}), as well as the \IPT and \MGT for a specific
   position in the \Dterm (Definition~\ref{def-ipt-mgt}). Let \[\alpha\; \eeqdef\;
   \{1 \mapsto \i(\vx{1}{\emptypos}, \i(\vx{2}{\emptypos},
   \vx{1}{\emptypos}))\}.\] That is, $\alpha$ is an axiom assignment that maps
   the \dconstant~$1$ to a variant of the argument term of axiom \Simp whose
   variables are positional variables~$x_{\emptypos}^i$. Consider the
   \dterm \[d\; \eeqdef\; \D(\D(1,1),1).\] Then $\pos(d) = \{\emptypos, 1, 1.1,
   1.2, 2\}$ and
  \[\begin{array}{l@{\hspace{0.5em}}c@{\hspace{0.5em}}l}
  \pairing_{\alpha}(d,\emptypos) & = & \{\vy_{1},\, \i(\vy_{2}, \vy_{\emptypos})\}.\\
  \pairing_{\alpha}(d,1) & = & \{\vy_{1.1},\, \i(\vy_{1.2}, \vy_1)\}.\\
  \pairing_{\alpha}(d,1.1) & = & \{\vy_{1.1},\, \i(\vx{1}{1.1}, \i(\vx{2}{1.1}, \vx{1}{1.1}))\}.\\
  \pairing_{\alpha}(d,1.2) & = & \{\vy_{1.2},\, \i(\vx{1}{1.2}, \i(\vx{2}{1.2}, \vx{1}{1.2}))\}.\\
  \pairing_{\alpha}(d,2) & = & \{\vy_{2},\, \i(\vx{1}{2}, \i(\vx{2}{2}, \vx{1}{2}))\}.\\
  \end{array}
  \]
  Let $\sigma \eeqdef \mgu(\{\pairing_{\alpha}(d, q) \mid q \in
  \pos(d)\}))$. We can then calculate that
  \[\begin{array}{l@{\hspace{0.5em}}c@{\hspace{0.5em}}ll}
  \sigma & \variant & \{
  & \vy_{\emptypos} \mapsto \i(\vx{1}{1.2}, \i(\vx{2}{1.2}, \vx{1}{1.2})),\\
  &&& \vy_{1} \mapsto \i(\i(\vx{1}{2}, \i(\vx{2}{2}, \vx{1}{2})), \i(\vx{1}{1.2}, \i(\vx{2}{1.2}, \vx{1}{1.2}))),\\
  &&& \vy_{1.1} \mapsto \i(\i(\vx{1}{1.2}, \i(\vx{2}{1.2}, \vx{1}{1.2})), \i(\i(\vx{1}{2}, \i(\vx{2}{2}, \vx{1}{2})), \i(\vx{1}{1.2}, \i(\vx{2}{1.2}, \vx{1}{1.2})))),\\
  &&& \vy_{1.2} \mapsto \i(\vx{1}{1.2}, \i(\vx{2}{1.2}, \vx{1}{1.2})),\\
  &&& \vy_{2} \mapsto \i(\vx{1}{2}, \i(\vx{2}{2}, \vx{1}{2})),\\
  &&& \vx{1}{1.1} \mapsto \i(\vx{1}{1.2}, \i(\vx{2}{1.2}, \vx{1}{1.2})),\\
  &&& \vx{2}{1.1} \mapsto \i(\vx{1}{2}, \i(\vx{2}{2}, \vx{1}{2}))\; \}.
  \end{array}
  \]
  We are going to compare the IPT and MGT of
  \[d' \eeqdef d|_1,\] that is, the subterm of $d$ at
  position~$1$. Then $d' = \D(1,1)$, $\pos(d') = \{\emptypos, 1, 2\}$, and
  \[\begin{array}{l@{\hspace{0.5em}}c@{\hspace{0.5em}}l}
  \pairing_{\alpha}(d',\emptypos) & = & \{\vy_{1},\, \i(\vy_{2}, \vy_{\emptypos})\}.\\
  \pairing_{\alpha}(d',1) & = & \{\vy_{1},\, \i(\vx{1}{1}, \i(\vx{2}{1}, \vx{1}{1}))\}.\\
  \pairing_{\alpha}(d',2) & = & \{\vy_{2},\, \i(\vx{1}{2}, \i(\vx{2}{2}, \vx{1}{2}))\}.\\
  \end{array}
  \]
  Let $\sigma' \eeqdef \mgu(\{\pairing_{\alpha}(d', q) \mid q \in
  \pos(d)\}))$. We can calculate that
  \[\begin{array}{l@{\hspace{0.5em}}c@{\hspace{0.5em}}ll}
  \sigma' & \variant & \{
  & \vy_{\emptypos} \mapsto \i(\vx{2}{1}, \i(\vx{1}{2}, \i(\vx{2}{2}, \vx{1}{2}))),\\
  &&& \vy_{1} \mapsto \i(\i(\vx{1}{2}, \i(\vx{2}{2}, \vx{1}{2})), \i(\vx{2}{1}, \i(\vx{1}{2}, \i(\vx{2}{2}, \vx{1}{2})))),\\
  &&& \vy_{2} \mapsto \i(\vx{1}{2}, \i(\vx{2}{2}, \vx{1}{2})),\\
  &&& \vx{1}{1} \mapsto \i(\vx{1}{2}, \i(\vx{2}{2}, \vx{1}{2}))\; \}.
  \end{array}
  \]
  Now $\ipt(d,1)$ and $\mgt(d|_1)$ can be determined as follows, where, to
  increase readability, we supplement additional variants with variable names
  $p,q,r,s$.
  \[
  \begin{array}{l@{\hspace{0.5em}}c@{\hspace{0.5em}}l}
    \ipt(d,1) & = & \P(y_1\sigma)\\
    & \variant & \P(\i(\i(\vx{1}{2}, \i(\vx{2}{2}, \vx{1}{2})), \i(\vx{1}{1.2}, \i(\vx{2}{1.2},
    \vx{1}{1.2}))))\\
    & \variant & \P(\i(\i(p, \i qp), \i(r, \i sr))).\\[1ex]
    \mgt(d|_1) & = & \mgt(d')\\
    & = & \ipt(d',\emptypos)\\
    & = & \P(y_{\emptypos}\sigma')\\
    & \variant & \P(\i(\vx{2}{1}, \i(\vx{1}{2}, \i(\vx{2}{2}, \vx{1}{2}))))\\
    & \variant & \P(\i(p, \i(q, \i rq))).
  \end{array}
  \]
  By Proposition~\ref{prop-ipt-subsumedby-mgt} it holds that $\ipt(d,1) \subsumedBy
  \mgt(d|_1)$, that is,
  \[\i(\i(p, \i qp), \i(r, \i sr)) \subsumedBy \i(p, \i(q, \i rq)),\]
  which is easy to verify.
  
  Side remark: In this simple example it holds that $\mgt(d) \variant \P(\i(p,
  \i(q, p)))$, that is, the \MGT of $d$ is a variant of the axiom \Syll. There
  is some apparent redundancy inherent in $d$, because it does just prove what
  a strict subterm of it, the primitive \Dterm $1$, proves. Such redundancies
  will be discussed in Sect.~\ref{sec-red}.

\end{exampapp}

Semantics now enters the stage with the entailment relationship $\entails$. By
universally closing the atoms on both sides of
Proposition~\ref{prop-ipt-subsumedby-mgt} we can relate MGT and IPT through
entailment.
\begin{prop}%
  \label{prop-mgt-entails-ipt}
  For all \dterms~$d$, axiom assignments $\alpha$ for $d$ and positions $p \in
  \pos(d)$ it holds that
  \[\UCN{\mgt_{\alpha}(d|_p)}\; \entails\; \UCN{\ipt_{\alpha}(d,p)}.\]
\end{prop}
\begin{proof} Follows from
  Proposition~\ref{prop-ipt-subsumedby-mgt}.
\end{proof}

The following lemma expresses the core relationships between the structure of
a proof (a \dterm), the unifying substitution of the pairings (underlying the
specification of \IPTs) and semantic entailment of the formulas associated
with positions in the structure (\IPTs).

\begin{lem}[Junction Core Lemma]
  \label{lem-core} For all \dterms~$d$, axiom assignments $\alpha$ for $d$ and
  positions $p \in \pos(d)$ it holds that

  \smallskip
  \subprop{lem-core-leaf} If $p \in \leafpos(d)$, then
  \[\UCM{\P(d|_p\alpha)}\; \entails\; \ipt_{\alpha}(d,p).\]

  \subprop{lem-core-inner} If $p \in \innerpos(d)$, then
  \[\Det \land \ipt_{\alpha}(d,p.1) \land \ipt_{\alpha}(d,p.2)\; \entails\; \ipt_{\alpha}(d,p).\]
\end{lem}
\begin{proof}
  Let $\sigma = \mgu(\{\pairing_{\alpha}(d, q) \mid q \in \pos(d)\})$ and
  assume it is defined.

  \smallskip
  (\ref{lem-core-leaf}) From Definition~\ref{def-ipt} and Definition~\ref{def-pairing} we
  can conclude $\ipt_{\alpha}(d,p)\, =\, \P(\vy_{p}\sigma)\, =\, 
  \P(d|_p\alpha\sshift{p}\sigma)\, \subsumedBy\, \P(d|_p\alpha)$,
  which implies the proposition to be proven.

  \smallskip
  (\ref{lem-core-inner}) From Definition~\ref{def-ipt} and Definition~\ref{def-pairing} we
  can conclude $\ipt(d,p.1)\, =\, \P(\vy_{p.1}\sigma)\, =\, \P(\i(\vy_{p.2},
  \vy_p)\sigma)$, $\ipt(d,p.2)\, =\, \P(\vy_{p.2}\sigma)$, and $\ipt(d,p)\,
  =\, \P(\vy_{p}\sigma)$. Hence, we can rephrase the proposition statement as
  \[\Det \land \P(\i(\vy_{p.2},
  \vy_p)\sigma) \land \P(\vy_{p.2}\sigma)\; \entails\; \P(\vy_{p}\sigma).\] By
  expanding the definition of~$\Det$ and rearranging formula components, this
  entailment can be brought into the following form, which obviously holds as
  its right side is obtained from instantiating universal quantifiers on the
  left side.
  \[\forall xy\, (\P x \land \P\i xy \imp \P y)\;
  \entails\; \P(\vy_{p.2}\sigma) \land \P(\i(\vy_{p.2}, \vy_p)\sigma) \imp
  \P(\vy_{p}\sigma).\vspace{-2ex}\]
\end{proof}

\noindent
Based on Lemma~\ref{lem-core}, the following theorem expresses how \Det
together with the axioms referenced in leaves entails the \MGT of a \dterm.

\begin{thm}[\MGT Entailment]%
  \label{thm-sem-mgt}
  For all \dterms~$d$ and axiom assignments~$\alpha$ for $d$ 
  it holds that
  \[\Det \land \bigwedge_{p \in \leafpos(d)}\hspace{-1.2em} \UCM{\P(d|_p\alpha)}
  \;\entails\; \UCN{\mgt_{\alpha}(d)}.\]
\end{thm}
\begin{proof}
 By induction on the structure of $d$ it follows from Lemma~\ref{lem-core}
 that \[\Det \land \bigwedge_{p \in \leafpos(d)}\hspace{-1.2em}
 \UCM{\P(d|_p\alpha)}
 \;\entails\; \ipt_{\alpha}(d,\emptypos).\] Contracting the definition of
 $\mgt$, the right side of this entailment can be written as
 $\mgt_{\alpha}(d)$. Since the left side of the entailment has no free
 variables, we can replace the right side with its universal closure and
 obtain the statement to be proven.
\end{proof}

\noindent
Theorem~\ref{thm-sem-mgt} states that \Det together with the ``axioms
referenced in the proof'', that is, the values of the axiom
assignment~$\alpha$ for the leaf nodes of the \Dterm~$d$ considered as
universally closed atoms, entail the universal closure of the \MGT of~$d$
for~$\alpha$.

In Meredith's proof notation, the displayed formulas represent the universal
closure of the \MGT. In a line without trailing \Dterm, the formula is an
axiom. In a line with a trailing \Dterm, the formula can be understood as
derived in two alternate ways, both yielding the same result. First, as the
universal closure of the MGT of the \Dterm after expanding the numeric labels
into their defining trees, exhaustively until all primitive \Dterms are axiom
labels. Second, as the universal closure of the MGT of the trailing \Dterm as
is, where its primitive \Dterms are taken as labels of displayed formulas in
the role of axioms.

\section{Reducing the Proof Size by Replacing Subproofs}
\label{sec-red}

The term view on proof trees suggests to shorten proofs by rewriting subterms,
that is, replacing occurrences of subproofs by other ones, with three main
aims:
\begin{enumerate}
\item To shorten given proofs, with respect to the tree size or the compacted
  size.
\item To investigate given proofs -- by humans or machines -- whether they can
  be shortened by certain rewritings or are closed under these.
\item To develop notions of redundancy for use in proof search. A proof
  fragment constructed during search may be rejected if it can be rewritten to
  a shorter one.
\end{enumerate}

\pagebreak %
Of course, any given proof of some theorem could be trivially shortened by
enumerating all smaller structures and checking whether one of them provides a
proof of the theorem.  Here our interest is in techniques for shortening
proofs that require less computational effort because they are based on
properties of subproofs of the given proof and involve criteria that can be
evaluated on the basis of a smaller search space than the set of all smaller
proofs.
As in Sect.~\ref{sec-cd-basis}, we consider purely structural aspects
separated from aspects involving formulas.

\subsection{Structural Criteria for Reducing the Compacted Size}
\label{subsec-struct-red}

To convert a proof to a smaller one or to detect that a proof is redundant
because of the existence of a smaller proof, it is essential to compare the
size of proofs before and after replacing occurrences of subproofs.
While for tree size the replacement of a subproof by a smaller one evidently
results in a smaller overall proof, for compacted size the effects of subproof
replacements are more intricate. In this subsection, a replacement criterion
for reducing the compacted size is developed, which is stated as
Theorem~\ref{thm-replace-csize-scsize} below. The theorem is based on ordering
relations on \dterms that are defined in terms of a strict version of
$\DSUBQ(d)$ (Definition~\ref{def-subq}), the set of all compound subterms of a
\Dterm~$d$.

\subsubsection{Compaction Orderings}

\begin{defn}
  \label{def-dsubsstrict} For \Dterms~$d$ define \[\DSUBSTRICT(d)\; \eqdef\; \{\D(e_1,e_2)
  \mid d \supterm \D(e_1,e_2)\}.\]
\end{defn}

\begin{defn}%
  \label{def-c-orderings}%
  For \dterms $d,e$ define

  \subdef{def:geqc}
  $d \geqc e\;\eqdef\; \DSUBSTRICT(d) \supseteq \DSUBSTRICT(e).$

  \subdef{def:gtrc}
  $d \gtrc e\; \eqdef\; d \geqc e \text{ and } e \not \geqc d.$
\end{defn}

\vspace{-5pt} %
\enlargethispage{10pt} %

We call the ordering relations $d \geqc e$ and $d \gtrc e$ \emph{compaction
orderings} because they relate to \emph{compacted} size rather than tree size.
They compare \dterms~$d$ and $e$ with respect to the superset relationship of
their sets of those strict subterms that are compound terms. For example,
$\D(\D(\D(1,1),1),1) \gtrc \D(1,\D(1,1))$ because $\{\D(1,1),\,
\D(\D(1,1),1)\} \supseteq \{\D(1,1)\}$.
The relation $d \gtrc e$ (Definition~\ref{def-c-orderings}) can equivalently be
characterized as $\DSUBSTRICT(d) \supset \DSUBSTRICT(e)$. Hence, the
underlying comparison is for $\geqc$ with respect to the non-strict superset
relationship and for $\gtrc$ the strict superset relationship. The $\geqc$
relation is a preorder on the set of \dterms, while $\gtrc$ is a strict
partial order. The subterm relationship includes the compaction orderings, as
noted by the following proposition.

\vspace{-5pt} %
\begin{prop}%
  \label{prop-gecq-supterm-all}
  For all \dterms $d,e,f$ it holds that
  \smallskip

  \subprop{prop-geqc-supterm} If $d \suptermq e $, then $d \geqc e$.

  \subprop{prop-gtrc-supterm} If $d \supterm e $ and $d$ is not of the form
  $\D(l_1,l_2)$ where both of $l_1, l_2$ are primitive \Dterms, then $d \gtrc e$.

  \subprop{prop-geqc-supterm-trans} If $d \suptermq e $ and $e \geqc f$, then
  $d \geqc f$.

  \subprop{prop-gtcr-supterm-trans} If $d \suptermq e $ and $e \gtrc f$, then
  $d \gtrc f$.
\end{prop}

\vspace{-5pt} %
\begin{proof}
  Easy to verify.
\end{proof}

According to Propositions~\ref{prop-geqc-supterm} and~\ref{prop-gtrc-supterm}
the subterm relationship includes the compaction orderings, with an exception,
as stated in the precondition of Proposition~\ref{prop-gtrc-supterm}. An example for
this exception is $\D(1,1) \supterm 1$ but $\D(1,1) \not \gtrc 1$. However, $d
\geqc e$ or $d \gtrc e$ also holds in cases where $d \notsuptermq e$, as
demonstrated by the following example.

\pagebreak %
\begin{examp}
  \label{examp-co}\prlReset{examp-co}
  The following table shows counterexamples for the converse statements of
  Propositions~\ref{prop-geqc-supterm} and~\ref{prop-gtrc-supterm}, that is,
  \dterms~$d$ and $e$ where $d \geqc e$ or $d \gtrc e$ holds but $d \not
  \suptermq e$. The respective values of $\DSUBSTRICT(d)$ and $\DSUBSTRICT(e)$
  underlying the definition of $\geqc$ are shown in a second table.
  {\small
  \[
  \begin{array}{r@{\hspace{0.5em}}l@{\hspace{0.5em}}c@{\hspace{0.5em}}l}
    &d&&e\\\midrule
    \prl{1} & 1 & \geqc & \D(1,1).\\
    \prl{2} & \D(1,\D(1,\D(1,1))) & \geqc & \D(\D(1,\D(1,1)),1).\\
    \prl{3} & \D(1,\D(1,\D(1,1))) & \gtrc & \D(\D(1,1),1).\\
    \prl{4} & \D(1,\D(1,\D(1,\D(1,1)))) & \gtrc & \D(\D(1,\D(1,1)),\D(1,\D(1,1))).\\
    \prl{5} & \D(1,\D(2,\D(3,3))) & \gtrc & \D(4,\D(3,3)).\\
  \end{array}
  \]}
  {\small
  \[
  \begin{array}{l@{\hspace{0.5em}}l@{\hspace{0.5em}}l}
    & \DSUBSTRICT(d) & \DSUBSTRICT(e)\\\midrule
    \pref{1} & \emptyset & \emptyset\\
    \pref{2} & \{\D(1,1),\; \D(1,\D(1,1))\} & \{\D(1,1),\; \D(1,\D(1,1))\}\\
    \pref{3} & \{\D(1,1),\; \D(1,\D(1,1))\} & \{\D(1,1)\}\\
    \pref{4} & \{\D(1,1),\; \D(1,\D(1,1)),\; \D(1,\D(1,\D(1,1)))\} & \{\D(1,1),\; \D(1,\D(1,1))\}\\
    \pref{5} & \{\D(3,3),\; \D(2,\D(3,3))\} & \{\D(3,3)\}\\
  \end{array}
  \]}
\end{examp}

The following proposition relates the compaction orderings to the compacted
size of the compared \dterms.
\begin{prop}%
  For all \dterms $d,e$ it holds that
  \smallskip

  \subprop{prop-geqc-csize} If $d$ is compound and $d \geqc e$, then
  $\csize(d) \geq \csize(e)$.

  \subprop{prop-cgt-csize} If $d \gtrc e$, then $\csize(d) > \csize(e)$.
\end{prop}
\begin{proof}
  Easy to verify.
\end{proof}

The converse statements of Propositions~\ref{prop-geqc-csize}
and~\ref{prop-cgt-csize} do not hold, as demonstrated by
the following example.
\begin{examp}%
  \label{examp-cocs}\prlReset{examp-cocs}
  The following table shows two counterexamples for the converse statements of
  Propositions~\ref{prop-geqc-csize} and~\ref{prop-cgt-csize}, that is,
  \dterms $d$ and $e$ such that $\csize(d) > \csize(e)$ and $d \not\geqc e$.
  The respective values of $\DSUBSTRICT(d)$ and $\DSUBSTRICT(e)$ underlying
  the definition of $\geqc$ are shown in a second table. {\small
  \[
  \begin{array}{r@{\hspace{0.5em}}l@{\hspace{0.5em}}c@{\hspace{0.5em}}l}
    &d&&e\\\midrule
    \prl{1} & \D(1,\D(1,\D(1,\D(1,1)))) & \not \geqc & \D(1,\D(\D(1,1),1)).\\
    \prl{2} & \D(1,\D(2,\D(3,3))) & \not \geqc & \D(4,\D(5,5)).\\
  \end{array}
  \]}
  {\small
  \[
  \begin{array}{l@{\hspace{0.5em}}l@{\hspace{0.5em}}l}
    & \DSUBSTRICT(d) & \DSUBSTRICT(e)\\\midrule
    \pref{1} & \{\D(1,1),\; \D(1,\D(1,1)),\; \D(1,\D(1,\D(1,1)))\}
             & \{\D(1,1),\; \D(\D(1,1),1)\}\\
    \pref{2} & \{\D(3,3),\; \D(2,\D(3,3))\} & \{\D(5,5)\}\\
  \end{array}
  \]}
\end{examp}

\subsubsection{The SC Size Measure of D-Terms}

Before we can state the main result of this subsection,
Theorem~\ref{thm-replace-csize-scsize}, we need to define a further size
measure of \Dterms, which we call \name{\SCsize}, suggesting \name{Sum of
  Compacted subterm sizes}. This auxiliary measure is useful in termination
arguments of repeated subterm replacement: The theorem shows a criterion under
which replacing subterm occurrences of a \Dterm reduces the compacted size,
but just \emph{non-strictly}, whereas the \SCsize is reduced \emph{strictly}.
The \SCsize is defined as follows.
\begin{defn}%
  \label{def-scsize} For \dterms $d$ define
  the \defname{\SCsize of $d$} as
  \[\scsize(d)\; \eqdef\; \sum_{d \suptermq e} \csize(e).\]
\end{defn}

The following two examples illustrate the \SCsize measure.

\begin{examp}%
  \label{examp-sc-size-1}
  Let $d$ be the \Dterm
  \[d\; \eeqdef\; \D(\D(\D(1,1),\D(1,1)),\D(\D(1,1),1)).\]
  Then the set $\{e \mid d \suptermq e \}$ of
  subterms of~$d$ is
  \[\begin{array}{ll}
  \{ & 1,\;
  \D(1,1),\;
  \D(\D(1,1),1),\;
  \D(\D(1,1),\D(1,1)),\\
  & \D(\D(\D(1,1),\D(1,1)),\D(\D(1,1),1))\; \},
  \end{array}
  \]
  and $\scsize(d) = 0+1+2+2+4 = 9$.
\end{examp}

\begin{examp}%
  \label{examp-sc-size-c}
  If $d,e$ are \dterms such that $\csize(d) > \csize(e)$, then
  it does not necessarily hold that $\scsize(d) \geq \scsize(e)$.
  The following $\dterms$ provide an example.
  \[
  \begin{array}{l@{\hspace{0.5em}}c@{\hspace{0.5em}}l}
    d & \eeqdef & \D(\D(\D(\D(\D(1,1),1),1),1),\D(1,\D(1,\D(1,\D(1,1))))).\\
    e & \eeqdef & \D(\D(\D(\D(\D(\D(\D(1,1),1),1),1),1),1),1).
  \end{array}
  \]
  It holds that $\csize(d) = 8 > 7 = \csize(e)$ but $\scsize(d) = 27 \not \geq
  28 = \scsize(e)$. The calculations of these values are based on the sets
  of subterms of~$d$ and of~$e$, shown in the following, where the compacted
  size of the respective member is annotated at the right.
  {\small
  \[
  \begin{array}{l@{\hspace{0.5em}}c@{\hspace{0.5em}}ll@{\hspace{0.5em}}c}
  &&&& \csize\\\midrule
  \{f \mid d \suptermq f \} & = & \{ & 1, & 0\\
  &&& \D(1,1), & 1\\
  &&& \D(1,\D(1,1)), & 2\\
  &&& \D(\D(1,1),1), & 2\\
  &&& \D(1,\D(1,\D(1,1))), & 3\\
  &&& \D(\D(\D(1,1),1),1), & 3\\
  &&& \D(1,\D(1,\D(1,\D(1,1)))), & 4\\
  &&& \D(\D(\D(\D(1,1),1),1),1), & 4\\
  &&& \D(\D(\D(\D(\D(1,1),1),1),1),\D(1,\D(1,\D(1,\D(1,1)))))\; \}. & 8\\[1ex]
  \{f \mid e \suptermq f \} & = & \{ & 1, & 0\\
  &&& \D(1,1), & 1\\
  &&& \D(\D(1,1),1), & 2\\
  &&& \D(\D(\D(1,1),1),1), & 3\\
  &&& \D(\D(\D(\D(1,1),1),1),1), & 4\\
  &&& \D(\D(\D(\D(\D(1,1),1),1),1),1), & 5\\
  &&& \D(\D(\D(\D(\D(\D(1,1),1),1),1),1),1), & 6\\
  &&& \D(\D(\D(\D(\D(\D(\D(1,1),1),1),1),1),1),1)\; \}. & 7
  \end{array}
  \]}

  \noindent %
  Hence $\csize(d) = 8$, $\scsize(d) = 0+1+2+2+3+3+4+4+8 = 27$,
  $\csize(e) = 7$ and $\scsize(e) = 0+1+2+3+4+5+6+7 = 28$.
\end{examp}

\subsubsection{Reducing the Compacted Size by Replacing Subproofs}

We are now ready to state the following theorem, the main result of this
subsection.
\begin{thm}[Reducing the Compacted Size by Replacing Subproofs]
  \label{thm-replace-csize-scsize}
  Let $d,d',e,e'$ be \mbox{\dterms} such that $e$ occurs in $d$, and $d' =
  \replace{d}{e}{e'}$. It holds that

  \smallskip

  \subthm{thm-replace-csize} If $e$ is compound and $e \geqc e'$, then
  $\csize(d) \geq \csize(d')$.

  \subthm{thm-replace-scsize} If $e \gtrc e'$, then $\scsize(d) >
  \scsize(d')$.

\end{thm}
\begin{proof}
  \prlReset{thm-replace-csize-scsize} We begin with shared aspects of the
  proofs of both subtheorems. We can assume that the \dterm $e$ is compound,
  which is explicitly stated as precondition for
  Theorem~\ref{thm-replace-csize} and implied by the precondition $e \gtrc e'$
  of Theorem~\ref{thm-replace-scsize}. There must exist a set
  $\{d_1,\ldots,d_n\}$ of compound \Dterms for some $n \geq 0$ such that the
  set $S \xeqdef \DSUBQ(d)$ of compound subterms of $d$ can be characterized
  as the disjoint union of three particular subsets in the following way.
  \[
  \begin{arrayprf}
    \prl{prf-thm-rcs-s-disjoint-union} &
    S = \{e\} \uplus \DSUBSTRICT(e) \uplus \{d_1, \ldots,
    d_n\}.
  \end{arrayprf}
  \]
  Let $T$ be the set of those strict subterms of $e$ that are compound and
  have in~$d$ an occurrence in a position other than as subterm of $e$.
  Clearly $\DSUBSTRICT(e) \supseteq T$. Thus,
  by~\pref{prf-thm-rcs-s-disjoint-union} we can characterize $S$ also as
  \[
  \begin{arrayprf}
    \prl{prf-thm-rcs-s-t} &
    S = \{e\} \cup \DSUBSTRICT(e) \cup T \cup \{d_1, \ldots,
    d_n\}.
  \end{arrayprf}
  \]
  Let $\CompD$ denote the set of all compound \Dterms.
  The set $S' \xeqdef \DSUBQ(d')$ of compound subterms of $d'$ can then be
  characterized as follows.
  \[
  \begin{arrayprf}
    \prl{prf-thm-rcs-sprime} &
    S' = (\{e'\} \cap \CompD) \cup \DSUBSTRICT(e') \cup T\; \cup\\
    & \hphantom{S' =\;} (\{\replace{d_1}{e}{e'}, \ldots,
    \replace{d_n}{e}{e'}\} \cap \CompD).
  \end{arrayprf}
  \]
  From $e \geqc e'$, which is a precondition of
  Theorem~\ref{thm-replace-csize} as well as Theorem~\ref{thm-replace-scsize},
  it follows that $\DSUBSTRICT(e) \supseteq \DSUBSTRICT(e')$. Since
  $\DSUBSTRICT(e) \supseteq T$ we can conclude from~\pref{prf-thm-rcs-sprime}
  that
  \[
  \begin{arrayprf}
    \prl{prf-thm-rcs-sprime-super} &
    (\{e'\} \cap \CompD) \cup \DSUBSTRICT(e) \cup
    \{\replace{d_1}{e}{e'}, \ldots, \replace{d_n}{e}{e'}\}\; \supseteq\; S'.
  \end{arrayprf}
  \]
  We now turn to the two individual subtheorems.

  \smallskip

  (\ref{thm-replace-csize}) Since $\csize(d) = |S|$ and $\csize(d') = |S'|$ we
  have to show that $|S| \geq |S'|$. From~\pref{prf-thm-rcs-sprime-super} it
  follows that $1 + |\DSUBSTRICT(e)| + |\{\replace{d_1}{e}{e'}, \ldots,
  \replace{d_n}{e}{e'}\}| \;\geq\; |S'|$. Since clearly $n \geq
  |\{\replace{d_1}{e}{e'}, \ldots, \replace{d_n}{e}{e'}\}|$ it follows that $1
  + |\DSUBSTRICT(e)| + n \;\geq\; |S'|$.
  Since~\pref{prf-thm-rcs-s-disjoint-union} implies $|S| = 1 +
  |\DSUBSTRICT(e)| + n$, that is, $|S|$ can be characterized as the left side
  of the previous disequation, it follows that $|S| \geq |S'|$, which
  concludes the proof of the subtheorem.

  \smallskip

  (\ref{thm-replace-scsize}) From~\pref{prf-thm-rcs-sprime-super} it follows
  that
  \[
  \begin{arrayprf}
    \prl{rs-1} &
    \csize(e') + \displaystyle\sum_{e \supterm f} \csize(f) + \sum_{i=1}^n
    \csize(\replace{d_i}{e}{e'})\; \geq\; \scsize(d').
  \end{arrayprf}
  \]
  Given the precondition $e \gtrc e'$ we can conclude by
  Theorem~\ref{thm-replace-csize} that for each $i \in \{1,\ldots,n\}$ it
  holds that $\csize(d_i) \geq \csize(\replace{d_i}{e}{e'})$. Hence
  \[
  \begin{arrayprf}
    \prl{prf-thm-rcs-si} &
    \displaystyle\sum_{i=1}^n \csize(d_i)\; \geq\; \sum_{i=1}^n
    \csize(\replace{d_i}{e}{e'}).
  \end{arrayprf}
  \]
  From the precondition $e \gtrc e'$ and Proposition~\ref{prop-cgt-csize} it follows
  that $\csize(e) > \csize(e')$. From~\pref{prf-thm-rcs-sprime-super}
  and~\pref{prf-thm-rcs-si} we can then conclude
  \[
  \begin{arrayprf}
    \prl{prf-thm-rscs-scs-gt} &
    \csize(e) + \displaystyle\sum_{e \supterm f} \csize(f) + \sum_{i=1}^n \csize(d_i)\; >\;
    \scsize(d').
  \end{arrayprf}
  \]
  By~\pref{prf-thm-rcs-s-disjoint-union}, $\scsize(d)$ can be characterized
  as follows.
  \[
  \begin{arrayprf}
    \prl{prf-thm-rscs-scs-s} &
    \scsize(d)\; =\; \csize(e) + \displaystyle\sum_{e \supterm f} \csize(f) + \sum_{i=1}^n
    \csize(d_i).
  \end{arrayprf}
  \]
  Since the right side of~\pref{prf-thm-rscs-scs-s} is identical to left side
  of~\pref{prf-thm-rscs-scs-gt} it follows that $\scsize(d) > \scsize(d')$,
  the conclusion of the subtheorem to be shown.
\end{proof}

Theorem~\ref{thm-replace-csize} states that if $d'$ is the \dterm obtained
from $d$ by simultaneously replacing \emph{all} occurrences of a compound
\dterm~$e$ with a ``c-smaller'' \dterm~$e'$, i.e., $e \geqc e'$, then the
compacted size of $d'$ is less than or equal to that of~$d$. Both,
precondition and conclusion of the theorem involve non-strict comparisons,
such that one may ask whether the strict precondition $e \gtrc e'$ would imply
the strict conclusion $\csize(d) > \csize(d')$. This is, however, not the
case, as demonstrated by Example~\ref{examp-replace-strict} below.
Nevertheless, as stated with the supplementary
Theorem~\ref{thm-replace-scsize}, the \SCsize is a measure that strictly
decreases under the strict precondition $e \gtrc e'$. By this subtheorem, the
number of replacements according to Theorem~\ref{thm-replace-csize-scsize}
that can be performed in succession with strict preconditions $e \gtrc e'$ is
finite. The \SCsize by itself, however, seems not suitable as a size measure
that refines the compacted size because, as already demonstrated by
Example~\ref{examp-sc-size-c}, there are \dterms $d,d'$ with $\csize(d) >
\csize(d')$ but $\scsize(d) < \scsize(d')$. Both of the following two examples
exhibit particularities of subproof replacements according to
Theorem~\ref{thm-replace-csize-scsize}.

\begin{examp}%
  \label{examp-replace-strict} This example shows that
  strengthening the precondition $e \geqc e'$ of
  Theorem~\ref{thm-replace-csize} to $e \gtrc e'$ does not in general permit
  the stronger conclusion $\csize(d) > \csize(d')$. Let
  \[\begin{array}{l@{\hspace{0.5em}}c@{\hspace{0.5em}}l}
  d & \eeqdef & \D(\D(1,\D(1,1)),\D(1,\D(1,\D(1,1)))).\\
  d' & \eeqdef & \D(\D(1,\D(1,1)),\D(\D(1,1),1)).\\
  e & \eeqdef & \D(1,\D(1,\D(1,1))).\\
  e' & \eeqdef & \D(\D(1,1),1).
  \end{array}
  \]
  Then~$e$ occurs in~$d$ and $d' = \replace{d}{e}{e'}$, matching the
  preconditions of Theorem~\ref{thm-replace-csize-scsize}.  Moreover, it holds
  that $e \gtrc e'$. By Theorem~\ref{thm-replace-csize} it follows that
  $\csize(d) \geq \csize(d')$. Indeed, $\csize(d) = \csize(d') = 4$.  By
  Theorem~\ref{thm-replace-scsize} it follows that $\scsize(d) >
  \scsize(d')$. Indeed, $\scsize(d) = 10$ and $\scsize(d') = 9$.  These
  properties and values can be determined on the basis of the following
  intermediate results. That $e \gtrc e'$ follows since
  \[\{f \in D \mid e \supterm f \} =
  \{\D(1,1),\; \D(1,\D(1,1))\} \supset \{\D(1,1)\} = \{f \in D \mid e'
  \supterm f \}.\]
  
\pagebreak %
  The sets $\{f \mid d \suptermq f \}$ and $\{f \mid d'
  \suptermq f\}$ underlying the calculation of $\csize(d)$, $\scsize(d)$,
  $\csize(d')$ and $\scsize(d')$ are as follows, where the compacted size of
  the respective member is annotated at the right.
  {\small
  \[
  \begin{array}{l@{\hspace{0.5em}}c@{\hspace{0.5em}}ll@{\hspace{0.5em}}c}
  &&&& \csize\\\midrule
  \{f \mid d \suptermq f \} & = & \{ & 1, & 0\\
  &&& \D(1,1), & 1\\
  &&& \D(1,\D(1,1)), & 2\\
  &&& \D(1,\D(1,\D(1,1))), & 3\\
  &&& \D(\D(1,\D(1,1)),\D(1,\D(1,\D(1,1))))\;\}. & 4\\[1ex]
  \{f \mid d' \suptermq f \} & = & \{ & 1, & 0\\
  &&& \D(1,1), & 1\\
  &&& \D(1,\D(1,1)), & 2\\
  &&& \D(\D(1,1),1), & 2\\
  &&& \D(\D(1,\D(1,1)),\D(\D(1,1),1))\; \}. & 4
  \end{array}
  \]}
\end{examp}

\begin{examp}%
  \label{examp-replace-twice} This example illustrates that the
  simultaneous replacement of \emph{all} occurrences of $e$ in $d$ by $e'$ is
  essential for Theorem~\ref{thm-replace-csize-scsize} and that $d'$, the
  formula after the replacement, can contain occurrences of~$e$ again. Let
  \[\begin{array}{l@{\hspace{0.5em}}c@{\hspace{0.5em}}l}
  d & \eeqdef &  \D(\D(\D(1,\D(1,\D(1,1))),1),\D(\D(1,\D(1,\D(1,1))),1)).\\
  d' & \eeqdef & \D(\D(\D(1,\D(1,1)),1),\D(\D(1,\D(1,1)),1)).\\
  d'' & \eeqdef & \D(\D(\D(1,\D(1,1)),1),\D(\D(1,\D(1,\D(1,1))),1)).\\
  e & \eeqdef & \D(1,\D(1,1)).\\
  e' & \eeqdef & \D(1,1).
  \end{array}
  \]
  Then~$e$ occurs in~$d$ and $d' = \replace{d}{e}{e'}$, matching the
  preconditions of Theorem~\ref{thm-replace-csize-scsize}. Moreover, it holds
  that $e \gtrc e'$. By Theorem~\ref{thm-replace-csize} it follows that
  $\csize(d) \geq \csize(d')$. Indeed, $\csize(d) = 5$ and $\csize(d') = 4$.
  Notice that $e$ occurs in $d'$, actually twice. The \dterm~$d''$ is obtained
  from $d$ by replacing just a single occurrence of $e$ with $e'$. Its
  compacted size is $\csize(d'') = 6$, thus not less than or equal to that
  of~$d$, $\csize(d) = 5$. The sets of compound subterms of $d$, $d'$
  and~$d''$, which underlie the determination of their compacted size, are as
  follows. {\small
  \[
  \begin{array}{l@{\hspace{0.5em}}c@{\hspace{0.5em}}ll}
  \DSUBQ(d) & = & \{
  & \D(1,1),\\
  &&& \D(1,\D(1,1)),\\
  &&& \D(1,\D(1,\D(1,1))),\\
  &&& \D(\D(1,\D(1,\D(1,1))),1),\\
  &&& \D(\D(\D(1,\D(1,\D(1,1))),1),\D(\D(1,\D(1,\D(1,1))),1))\;\}.\\[1ex]
  \DSUBQ(d') & = & \{
  & \D(1,1),\\
  &&& \D(1,\D(1,1)),\\
  &&& \D(\D(1,\D(1,1)),1),\\
  &&& \D(\D(\D(1,\D(1,1)),1),\D(\D(1,\D(1,1)),1))\;\}.\\[1ex]
  \DSUBQ(d'') & = & \{
  & \D(1,1),\\
  &&& \D(1,\D(1,1)),\\
  &&& \D(1,\D(1,\D(1,1))),\\
  &&& \D(\D(1,\D(1,1)),1),\\
  &&& \D(\D(1,\D(1,\D(1,1))),1),\\
  &&& \D(\D(\D(1,\D(1,1)),1),\D(\D(1,\D(1,\D(1,1))),1))\;\}.
  \end{array}
  \]}
\end{examp}

The following proposition characterizes the number of \dterms that are smaller
than a given \dterm with respect to the compaction ordering~$\geqc$.
\begin{prop}%
  \label{prop-co-number}
  For all compound \dterms $d$ it holds that
  \[\begin{array}{r@{\hspace{0.5em}}l}
  & |\{e \mid d \geqc e \text{ and } \DPRIM(e) \subseteq \DPRIM(d)\}|\\
  = & (\csize(d) - 1 + |\DPRIM(d)|)^2 + |\DPRIM(d)|.
  \end{array}
  \]
\end{prop}
\begin{proof}
  \prlReset{prop-co-number}
  Let $S$ be the set whose cardinality is denoted by the left side of the
  proposition. Then
  \[\begin{arrayprfeq}
  \prl{1} && S\\
  \prl{2} & =  & \{e \mid d \geqc e \text{ and } \DPRIM(e) \subseteq
  \DPRIM(d)\}\\
  \prl{3} & = & \{e \mid
  \DSUBSTRICT(d) \supseteq
  \DSUBSTRICT(e) \text{ and } \DPRIM(e) \subseteq
  \DPRIM(d)\}\\
  \prl{4} & =  & \{\D(d_1,d_2) \mid d \supterm d_1 \text{ and } d \supterm d_2\}
  \uplus \DPRIM(d).
  \end{arrayprfeq}
  \]
  Since $\{e \mid d \supterm e\} = \DSUBSTRICT(d) \uplus \DPRIM(d)$ and
  $\csize(d)$ is defined as $|\DSUBQ(d)|$ it follows that
  \[\begin{arrayprfeq}
  \prl{5} && |\{e \mid d \supterm e\}| = \csize(d) - 1 + |\DPRIM(d)|.
  \end{arrayprfeq}
  \]
  From the representation of $S$ in the form~\pref{4} and~\pref{5} it
  follows that $|S| = (\csize(d) - 1 + |\DPRIM(d)|)^2 +  |\DPRIM(d)|$,
  that is, the proposition statement.
\end{proof}

By Proposition~\ref{prop-co-number}, for a given compound \dterm~$d$, the number of
\dterms~$e$ that are smaller than~$d$ with respect to $\geqc$ is only
quadratically larger than the compacted size of~$d$ and thus also than the
tree size of~$d$. Hence techniques that inspect all these smaller \dterms for
a given \dterm can be used efficiently in practice. For example to find
\Dterms that can be replaced according to
Theorem~\ref{thm-replace-csize-scsize}, that is, in view of the preconditions
of the theorem, finding \Dterms~$e'$ for a given \Dterm~$e$. Or to classify a
\dterm as redundant because there exists a smaller \dterm that proves the
same.

\subsection{Formula-Related Criteria for Subproof Replacement}
\label{subsec-form-red}

According to Theorem~\ref{thm-sem-mgt}, a \CD proof, that is, a \dterm~$d$
together with an axiom assignment~$\alpha$ proves the \MGT of~$d$ for~$\alpha$
along with all instances of the \MGT. If~$d$ is shortened by replacing
subterms, the general objective is that at least these theorems are still
proven. That is, the \MGT of the modified \Dterm subsumes that of the original
one. In this subsection we identify conditions that ensure that subterm
replacement steps yield proofs with a \MGT that subsumes the \MGT before the
replacement. These conditions will be stated as Theorems~\ref{thm-sr-ipt}
and~\ref{thm-sr-mgt}, which are both consequences of a central underlying
property that will be stated as Lemma~\ref{lem-sr-mult}.

\subsubsection{Decomposition of the MGU Associated with a D-Term}

The proof of Lemma~\ref{lem-sr-mult} involves several applications of a
decomposition of the most general unifier ``associated'' with a \Dterm, that
is, the most general unifier of the set of pairings of all its positions, with
respect to a given axiom assignment~$\alpha$. This decomposition is specified
now with Lemma~\ref{lem-dds-mult}, preceded by an auxiliary proposition, which
shows a specific way to pass between sets of pairs of terms and most general
unifiers.
\begin{prop}[{\cite[Lemma~4.6]{eder:subst:1985}}]
  If $M, N$ are sets of pairs of terms and $\sigma$ is a most general unifier
  of $M$, then

  \subprop{prop-mgu-compose-unifiable}
  $M \cup N$ is unifiable if and only if $N\sigma$ is unifiable.

  \subprop{prop-mgu-compose-asymmetric} If $\tau$ is a most general unifier of
  $N\sigma$, then $\sigma\tau$ is a most general unifier of $M \cup N$.
\end{prop}

\begin{lem}[Decomposition of the MGU Associated with a \DTerm]
  \label{lem-dds-mult}
  Let $d$ be a \dterm and let $p_1, \ldots, p_n, q$, where $n \geq 0$, be
  positions in $\pos(d)$ such that for all $i \in \{1,\ldots,n\}$ it holds
  that $p_i \not < q$. Then
  \[\vy_{q}\sigma\; \variant\;
  \vy_{q}\gamma\mgu(\{\{\vy_{p_1}\gamma,\, \vy_{p_1}\tau\gamma\},
  \ldots, \{\{\vy_{p_n}\gamma,\, \vy_{p_n}\tau\gamma\}\}),\] where
    \[\begin{array}{l@{\hspace{0.5em}}c@{\hspace{0.5em}}l@{\hspace{0.5em}}r@{\hspace{0.5em}}c@{\hspace{0.5em}}l}
  \sigma & = & \mgu(\{\pairing_{\alpha}(d, r) \mid r \in \pos(d)\}),\\
  \tau & = & \mgu(\{\pairing_{\alpha}(d, r) \mid r \in \pos(d)
  \text{ and } p_i \leq r \text{ for some } i \in \{1,\ldots,n\}\}), \text{ and}\\
  \gamma & = & \mgu(\{\pairing_{\alpha}(d, r) \mid r \in \pos(d)
  \text{ and } p_i \not \leq r \text{ for all } i \in \{1,\ldots,n\}\}).\\
  \end{array}
    \]
\end{lem}
\begin{proof}
  \prlReset{lem-dds-mult-new}
  Let
  \[\begin{array}{l@{\hspace{0.5em}}c@{\hspace{0.5em}}l@{\hspace{0.5em}}r@{\hspace{0.5em}}c@{\hspace{0.5em}}l}
  S & \xeqdef & \{\pairing_{\alpha}(d, r) \mid r \in \pos(d)\},\\
  T & \xeqdef & \{\pairing_{\alpha}(d, r) \mid r \in \pos(d)
  \text{ and } p_i \leq r \text{ for some } i \in \{1,\ldots,n\}\},\\
  G & \xeqdef & \{\pairing_{\alpha}(d, r) \mid r \in \pos(d)
  \text{ and } p_i \not \leq r \text{ for all } i \in \{1,\ldots,n\}\}.
  \end{array}\]
  Then $\sigma = \mgu(S)$, $\tau = \mgu(T)$, and $\gamma = \mgu(G)$.  From the
  definition of $\pairing$ (Definition~\ref{def-pairing}) and the precondition $p_i
  \not < q$ for all $i \in \{1,\ldots,n\}$ it follows that
  \[
  \begin{arrayprf}
    \prl{varst} & \vars(T) \subseteq \posvar(\{r \mid
    p_i \leq r \text{ for some } i \in \{1,\ldots,n\}\}).\\
    \prl{varsg} & \vars(G) \subseteq \posvar(\{r \mid
    p_i \not\leq r \text{ for all } i \in \{1,\ldots,n\}\}
    \cup \{y_{p_1}, \ldots, y_{p_n}\}).\\
    \prl{varsprime} & y_{q} \in \posvar(\{r \mid
    p_i \not\leq r \text{ for all } i \in \{1,\ldots,n\}\}
    \cup \{y_{p_1}, \ldots, y_{p_n}\}).
  \end{arrayprf}
  \]
  The lemma can now be shown in the following steps, explained below.
  \[
  \begin{arrayprfeq}
  \prl{1} & & \vy_{q}\sigma\\
  \prl{2}  & = & \vy_{q}\mgu(S)\\
  \prl{3}  & = & \vy_{q}\mgu(T \cup G)\\
  \prl{4}  & \variant & \vy_{q}\tau\mgu(G\tau)\\
  \prl{5}  & = & \vy_{q}\tau|_{\{\vy_{p_1},\ldots,\vy_{p_n}\}}\mgu(G\tau|_{\{\vy_{p_1},\ldots,\vy_{p_n}\}})\\
  \prl{6}  & \variant & \vy_{q}\mgu(\{\{\vy_p,\vy_p\tau\}\} \cup G)\\
  \prl{7}  & \variant &
  \vy_{q}\gamma\mgu(\{\{\vy_p\gamma,\vy_p\tau\gamma\}\}).\\
  \end{arrayprfeq}
  \]
  Step~\pref{2} is obtained by expanding the definition of $\sigma$, and
  step~\pref{3} follows since $S = T \cup G$.  Step~\pref{4} is obtained by
  Proposition~\ref{prop-mgu-compose-asymmetric}.  By~\pref{varsg} and~\pref{varst}
  it follows that $\vars(G) \cap \vars(T) \subseteq
  \{\vy_{p_1},\ldots,\vy_{p_n}\}$ and by~\pref{varsprime} and~\pref{varst}
  that $\{\vy_{q}\} \cap \vars(T) \subseteq
  \{\vy_{p_1},\ldots,\vy_{p_n}\}$. Since $\dom(\tau) \subseteq \vars(T)$ we
  can replace $\tau$ in~\pref{4} with its restriction to
  $\{\vy_{p_1},\ldots,\vy_{p_n}\}$ and obtain~\pref{5}.  Step~\pref{6} follows
  from Proposition~\ref{prop-mgu-compose-asymmetric} since $\tau|_{\{\vy_{p_1},
    \ldots, \vy_{p_n}\}} \variant \mgu(\{\{\vy_{p_1}, \vy_{p_1}\tau\},\ldots,
  \{\vy_{p_n}, \vy_{p_n}\tau\}\})$.  Finally, step~\pref{7} is obtained by
  Proposition~\ref{prop-mgu-compose-asymmetric} and the definition of~$\gamma$.
\end{proof}

\subsubsection{The Subproof Replacement Monotonicity Core Lemma}

Lemma~\ref{lem-sr-mult}, stated and proven in this subsubsection, shows how
the subsumption relationship of associated formulas transfers from subterm
occurrences in a \Dterm to the \Dterm itself. The setting of the lemma is
illustrated in Fig.~\ref{fig-lem-sr-mult}.
\begin{figure}
  \centering
  \scalebox{0.8}{
\begin{tikzpicture}
  \tikzset{
    dot/.style = {circle, fill, minimum size=13pt,
      inner sep=0pt, outer sep=0pt, draw=black},
  }
  \draw (0,0) {} -- (5.5,0) -- (2.75,4.76) -- cycle;
  \fill (0.77,0) {} -- (0.77+1.83,0) -- (0.77+0.91,1.58) -- cycle;
  \fill (3.36,0) {} -- (3.36+1.375,0) -- (3.36+0.69,1.19) -- cycle;

  \node [dot,fill=white] at (2.75,4.76) {$\emptypos$};
  \node [dot,fill=black!10] at (0.77+0.91,1.58) {$p_1$};
  \node [dot,fill=black!10] at (3.36+0.69,1.19) {$p_2$};
  \node [dot,fill=white] at (2.38,2.38) {$q$};
\end{tikzpicture}
\hspace{0.5cm}
\begin{tikzpicture}
  \tikzset{
    dot/.style = {circle, fill, minimum size=13pt,
      inner sep=0pt, outer sep=0pt, draw=black},
  }

  \fill [black!20] (0.77+0.91,1.58) {} -- (0.77+0.91-0.55,1.58-0.95) --
  (0.77+0.91+0.55,1.58-0.95) -- cycle;

  \fill [black!20] (3.36+0.69,1.19) {} -- (3.36+0.69-0.55,1.19-0.95) --
  (3.36+0.69+0.55,1.19-0.95) -- cycle;

  \draw (0,0) {} -- (0.77,0) -- (0.77+0.91,1.58)
  -- (0.77+1.83,0) -- (3.36,0) -- (3.36+0.69,1.19)
  -- (3.36+1.375,0) -- (5.5,0) -- (2.75,4.76) -- cycle;

  \node [dot,fill=white] at (2.75,4.76) {$\emptypos$};
  \node [dot,fill=black!10] at (0.77+0.91,1.58) {$p_1$};
  \node [dot,fill=black!10] at (3.36+0.69,1.19) {$p_2$};
  \node [dot,fill=white] at (2.38,2.38) {$q$};

  \node [] at (0.77+0.91,1.58-0.6) {$e$};
  \node [] at (3.36+0.69,1.19-0.6) {$e$};

\end{tikzpicture}
  }
  \vspace{10pt}
  \caption{The setting of Lemma~\ref{lem-sr-mult} for $n=2$.  The left side
  illustrates the \dterm~$d$. Of the positions~$p_1$ and~$p_2$ neither one is
  below the other one.  Position~$q$ must be neither strictly below $p_1$ nor
  strictly below $p_2$. That is, $q$ can be anywhere in the white area
  including $\emptypos$, or it can be one of $p_1$ or $p_2$.  The right side
  illustrates the \dterm $d[e]_{p_1}[e]_{p_2}$ which is obtained from~$d$ by
  replacing the subterms at~$p_1$ and~$p_2$ with occurrences of the
  \dterm~$e$, indicated by smaller gray triangles.}
  \label{fig-lem-sr-mult}  
\end{figure}

\pagebreak
\begin{lem}[Subproof Replacement Monotonicity Core Lemma]
  \label{lem-sr-mult}
  Let $d, e$ be \dterms, let~$\alpha$ be an axiom assignment for $d$ and for
  $e$, and let $p_1, \ldots, p_n, q$, where $n \geq 0$, be positions in
  $\pos(d)$ such that for all $i,j \in \{1,\ldots,n\}$ with $i \neq j$ it
  holds that $p_i \not \leq p_j$ and for all $i \in \{1,\ldots,n\}$ it holds
  that $p_i \not < q$.  If for all $i \in \{1, \ldots, n\}$ it holds that
  \[\ipt_{\alpha}(d, p_i) \subsumedBy \mgt_{\alpha}(e),\] then
  \[\ipt_{\alpha}(d,q) \subsumedBy
  \ipt_{\alpha}(d[e]_{p_1}[e]_{p_2}\ldots [e]_{p_n}, q).\]
\end{lem}
\begin{proof}
  \prlReset{lem-sr-mult}
  Define the shorthand $d' = d[e]_{p_1}[e]_{p_2}\ldots [e]_{p_n}$. That is,
  $d'$ is $d$ with the subterm occurrences at $p_1, \ldots, p_n$ replaced
  by~$e$. Define the following sets of pairs of terms and substitutions.
  \[\begin{array}{l@{\hspace{0.5em}}c@{\hspace{0.5em}}l@{\hspace{0.5em}}r@{\hspace{0.5em}}c@{\hspace{0.5em}}l}
  S &  \xeqdef & \{\pairing_{\alpha}(d, r) \mid r \in \pos(d)\}.\\
  T &  \xeqdef & \{\pairing_{\alpha}(d, r) \mid r \in \pos(d)
  \text{ and } p_i \leq r \text{ for some } i \in \{1,\ldots,n\}\}.\\
  T' &  \xeqdef & \{\pairing_{\alpha}(d', r) \mid r \in \pos(d')
  \text{ and } p_i \leq r \text{ for some } i \in \{1,\ldots,n\}\}.\\
  G &  \xeqdef & \{\pairing_{\alpha}(d, r) \mid r \in \pos(d)
  \text{ and } p_i \not \leq r \text{ for all } i \in \{1,\ldots,n\}\}.\\
  \sigma &  \xeqdef &  \mgu(S).\\
  \tau &  \xeqdef & \mgu(T).\\
  \tau' &  \xeqdef & \mgu(T').\\
  \gamma &   \xeqdef & \mgu(G).\\
  \mu & \xeqdef & \mgu(\{\{\vy_{p_1}\gamma,\vy_{p_1}\ssubold\gamma\},
  \ldots, \{\vy_{p_n}\gamma,\vy_{p_n}\ssubold\gamma\}\}).\\
  \nu & \xeqdef & \mgu(\{\{\vy_{p_1}\gamma,\vy_{p_1}\ssubnew\gamma\},
  \ldots, \{\vy_{p_n}\gamma,\vy_{p_n}\ssubnew\gamma\}\}).
  \end{array}
  \]
  Because the detailed proof is lengthy, we present it modularized into four
  parts, \name{(I)~Conversion of the Preconditions}, \name{(II)~Determining
    the Instantiating Substitution $\rho$}, \textit{(III)~Contexts where
    $\rho$ is Void}, and \textit{(IV)~Deriving the
    Conclusion}. Figure~\ref{fig-lem-sr-mult} may help to get an intuitive
  overview of the parameters of the lemma statement.

  \beforeproofpartskip
  \noindent{\textit{Part I. Conversion of the Preconditions}}
  \afterproofpartskip

  \noindent
  The following step is a precondition of the lemma to be proven.
  \[
  \begin{arrayprf}
    \prl{pre-disjoint} &
    p_i \not \leq p_j, \text{ for all } i,j \in \{1,\ldots,n\} \text{ with } i \neq j.
  \end{arrayprf}
  \]
  The following statements, whose proofs are described below, show that $\sigma$
  when applied to $\vy_q$ and $\vy_{p_i}$ can be decomposed into $\gamma$
  followed by $\mu$.
  \[
  \begin{arrayprf}
    \prl{pi-gamma-mu} & y_{p_i}\sall\; =\; y_{p_i}\gamma\mu,
    \text{ for all } i \in \{1,\ldots,n\}.\\
    \prl{q-gamma-mu} & y_q\sall\; =\; y_q\gamma\mu.\\
  \end{arrayprf}
  \]
  Step~\pref{pi-gamma-mu} follows from Lemma~\ref{lem-dds-mult} with its
  parameters $p_1, \ldots, p_n$ instantiated by the positions of the same name
  in the lemma to be proven but its parameter~$q$ instantiated to $p_i$ for an
  arbitrary $i \in \{1,\ldots,n\}$. The precondition $p_i \not < q$ for all $i
  \in \{1,\ldots,n\}$ of Lemma~\ref{lem-dds-mult} then instantiates to $p_j
  \not < p_i$ for all $j \in \{1,\ldots,n\}$, which follows
  from~\pref{pre-disjoint}.  Step~\pref{q-gamma-mu} follows from
  Lemma~\ref{lem-dds-mult} with all of its parameters $p_1, \ldots, p_n, q$
  instantiated by the positions of the same names in the lemma to be proven.

  Let us consider now the precondition $\ipt_{\alpha}(d, p_i) \subsumedBy
  \mgt_{\alpha}(e)$ for an arbitrary $i \in \{1,\ldots,n\}$. Its left side
  can be converted by expanding and contracting definitions and
  step~\pref{pi-gamma-mu} as follows.
  \[
  \begin{arrayprfeq}
    \prl{x1} && \ipt_{\alpha}(d, p_i)\\
    \prl{x2} & = &  \P(y_{p_i}\mgu(\{\pairing_{\alpha}(d, r) \mid r \in \pos(d)\}))\\
    \prl{x3} & = &  \P(y_{p_i}\mgu(S))\\
    \prl{x4} & = &  \P(y_{p_i}\sigma)\\
    \prl{x5} & = &  \P(y_{p_i}\gamma\mu).\\
  \end{arrayprfeq}
  \]
  The conversion of the right side of the considered precondition is based on
  some auxiliary definitions and statements.  For all $i \in \{1,\ldots,n\}$
  define the following sets of pairs of terms and substitutions.
  \[\begin{array}{l@{\hspace{0.5em}}c@{\hspace{0.5em}}l@{\hspace{0.5em}}r@{\hspace{0.5em}}c@{\hspace{0.5em}}l}
  T'_i & \xeqdef & \{\pairing_{\alpha}(d', r) \mid r \in \pos(d')
  \text{ and } p_i \leq r\}.\\
  \NT'_i & \xeqdef & \bigcup_{j \in \{1,\ldots,n\}\setminus\{i\}} T'_j.
  \end{array}\]
  Then, as explained below, for all $i,j \in \{1,\ldots,n\}$ the
  following holds.
  \[
  \begin{arrayprf}
    \prl{tpi1} & T'_i \cup \NT'_i = T'.\\
    \prl{tpi2} & \vars(T'_i) \subseteq \posvar(\{r \mid p_i \leq r\}).\\
    \prl{tpi3} & \text{If } i \neq j, \text{ then } \vars(T'_i) \cap \vars(T'_j)
    = \emptyset.\\
    \prl{tpi4} & \vars(\NT'_i) \cap \{y_{p_i}\} = \emptyset.\\
    \prl{tpi5} & \vars(\NT'_i) \cap \vars(T'_i) = \emptyset.\\
    \prl{tpi6} & y_{p_i}\tau' = y_{p_i}\mgu(T'_i).\\
    \prl{idx-ssubnew-dom} &
    \text{If } y_{p_i} \in \dom(\ssubnew), \text{ then }
    \vars(y_{p_i}\ssubnew) \subseteq \posvar(\{r \mid p_i < r\}).\\
    \prl{nonoverlap-term} & \text{If } i \neq j, \text{ then }
    \vars(y_{p_i}\ssubnew) \cap \vars(y_{p_j}\ssubnew) = \emptyset.\\
  \end{arrayprf}
  \]
  Step~\pref{tpi1} follows immediately from the definitions of $T'_i$,
  $\NT'_i$ and $T$.  Step~\pref{tpi2} follows from the definition of $T'_i$
  and the definition of $\pairing$ (Definition~\ref{def-pairing}). Step~\pref{tpi3}
  follows from~\pref{tpi2} and~\pref{pre-disjoint}.  Step~\pref{tpi4} follows
  from the definition of $\NT'_i$ and steps~\pref{tpi2} and~\pref{pre-disjoint}.
  Step~\pref{tpi5} follows from the definition of $\NT'_i$ and
  step~\pref{tpi3}.  Step~\pref{tpi6} follows from the definition of~$\tau'$
  and steps~\pref{tpi1}, \pref{tpi4} and~\pref{tpi5}.
  Step~\pref{idx-ssubnew-dom} follows from~\pref{tpi6} and~\pref{tpi2}.
  Step~\pref{nonoverlap-term} follows from~\pref{tpi6}, \pref{tpi3} and
  \pref{pre-disjoint}.

  The right side of the precondition $\ipt_{\alpha}(d, p_i) \subsumedBy
  \mgt_{\alpha}(e)$ can now be converted in the following steps described
  below.
  \[
  \begin{arrayprfeq}
    \prl{y1} && \mgt_{\alpha}(e)\\
    \prl{y2} & = & \P(y_{\emptypos}\mgu(\{\pairing_{\alpha}(e, r) \mid r \in
    \pos(e)\}))\\
    \prl{y3} & \variant
    & \P(y_{\emptypos}\mgu(\{\pairing_{\alpha}(e, r) \mid r \in \pos(e)\})\sshift{p_i})\\
    \prl{y4} & =
    & \P(y_{\emptypos}\mgu(\{\pairing_{\alpha}(d'|_{p_i}, r) \mid r \in \pos(d'|_{p_i})\})\sshift{p_i})\\
    \prl{y5} & =
    & \P(y_{p_i}\mgu(\{\pairing_{\alpha}(d', r) \mid r \in \pos(d') \text{ and } p_i \leq r\}))\\
    \prl{y6} & = & \P(y_{p_i}\mgu(T'_i))\\
    \prl{yx10} & = & \P(y_{p_i}\tau').\\
  \end{arrayprfeq}
  \]
  Step~\pref{y2} is obtained from~\pref{y1} by expanding the definition of
  $\mgt$.  Step~\pref{y3} follows from Proposition~\ref{prop-shift-variant},
  step~\pref{y4} since by the definition of $d'$ it holds that $d'|_{p_i} = e$,
  and step~\pref{y5} from Proposition~\ref{prop-tree-renaming-add}.  Step~\pref{y6}
  is obtained by contracting the definition of $T'_i$. Step~\pref{yx10}
  follows from~\pref{tpi6}. Note that \pref{y1} is independent from~$i$ and
  the conversion of~\pref{y1} to~\pref{yx10} is possible for any $i \in
  \{1,\ldots,n\}$.

  Because~\pref{x1} and~\pref{x5} as well as~\pref{y1} and~\pref{yx10} are
  equal, we can now reformulate the precondition that for all $i \in
  \{1,\ldots,n\}$ it holds that $\ipt_{\alpha}(d, p_i) \subsumedBy
  \mgt_{\alpha}(e)$ as
  \[\begin{arrayprf}
  \prl{pre-v1} & y_{p_i}\gamma\mu \subsumedBy y_{p_i}\tau',
  \text{ for all } i \in \{1,\ldots,n\}.
  \end{arrayprf}\]

  \beforeproofpartskip
  \noindent{\textit{Part II. Determining the Instantiating Substitution $\rho$}}
  \afterproofpartskip

  \noindent
  We show, as explained below, that for all $i \in \{1,\ldots,n\}$ there
  exists a substitution~$\rho_i$ with the following properties.
  \[
  \begin{arrayprf}
    \prl{pre-i} & y_{p_i}\gamma\mu\; =\; y_{p_i}\ssubnew\rho_i.\\
    \prl{pre-rho-dom} &
    \dom(\rho_i) \subseteq \vars(y_{p_i}\ssubnew).\\
    \prl{pre-i-dom-2} &
    \text{If } y_{p_i} \in \dom(\ssubnew), \text{ then }
    \dom(\rho_i) \subseteq \posvar(\{r \mid p_i < r\}).\\
    \prl{pre-i-dom-disjoint} &
    \text{If } i \neq j, \text{ then }
    \dom(\rho_i) \cap \dom(\rho_j) = \emptyset.\\
    \prl{pre-i-dom-cap} &
    \dom(\rho_i) \cap \dom(\ssubnew) = \emptyset.\\
  \end{arrayprf}
  \]
  Steps~\pref{pre-i} and~\pref{pre-rho-dom} follow from~\pref{pre-v1}.
  Step~\pref{pre-i-dom-2} follows from~\pref{pre-rho-dom}
  and~\pref{idx-ssubnew-dom}, step~\pref{pre-i-dom-disjoint}
  from~\pref{pre-rho-dom} and~\pref{nonoverlap-term}.
  Step~\pref{pre-i-dom-cap} follows from~\pref{pre-rho-dom} since the
  idempotence of~$\tau'$ is equivalent to $\dom(\tau') \cap \vrng(\tau') =
  \emptyset$, which implies $\vars(\vy_{p_i}\tau') \cap \dom(\tau') =
  \emptyset$.

  Step~\pref{pre-i-dom-disjoint} justifies to define a substitution~$\rho$,
  which combines the substitutions $\rho_i$ by forming their union:
  \[\rho\; \xeqdef\; \bigcup_{i=1}^{n}\{v \mapsto v\rho_i \mid v \in
  \dom(\rho_i)\}.\]
  The substitution~$\rho$ has the following properties, whose derivation
  is described below.
 \[
 \begin{arrayprf}
   \prl{pre-dom-components} & y_{p_i}\ssubnew\rho \; =\;  y_{p_i}\ssubnew\rho_i,
   \text{ for all } i \in \{1,\ldots,n\}.\\
    \prl{pre} & y_{p_i}\gamma\mu\; =\; y_{p_i}\ssubnew\rho,
    \text{ for all } i \in \{1,\ldots,n\}.\\
    \prl{dom-rho-2} & \dom(\rho) \subseteq \posvar(\{r \mid
    p_i \leq r \text{ for some } i \in \{1,\ldots,n\}\}).\\
   \prl{pre-dom-cap} & \dom(\rho) \cap \dom(\ssubnew) = \emptyset.\\
 \end{arrayprf}
  \]
  Step~\pref{pre-dom-components} follows from the definition of~$\rho$, given
  that for all $i,j \in \{1,\ldots,n\}$ with $i \neq j$ it holds that
  $\vars(\vy_{p_i}\tau') \cap \dom(\rho_j) = \emptyset$, which follows from
  \pref{pre-rho-dom} and \pref{nonoverlap-term}.
  Step~\pref{pre} follows from \pref{pre-dom-components} and~\pref{pre-i}.
  Step~\pref{dom-rho-2} follows from the definition of~$\rho$ and
  steps~\pref{pre-rho-dom}, \pref{tpi6}, and~\pref{tpi2}.
  Step~\pref{pre-dom-cap} follows from the definition of~$\rho$ and
  step~\pref{pre-i-dom-cap}.

  \beforeproofpartskip
  \noindent{\textit{Part III. Contexts where $\rho$ is Void}}
  \afterproofpartskip

  \noindent
  The variables occurring in members of the range of $\gamma$ as well as
  $\vy_q$ are contained in the same set of \positionalvars.
  \[
  \begin{arrayprf}
    \prl{rng-gamma} & \vrng(\gamma) \subseteq \posvar(\{r \mid p_i \not \leq r
    \text{ for all } i \in \{1,\ldots,n\}\}) \cup \{y_{p_1}, \ldots,
    y_{p_n}\}.\\
    \prl{yq-incl} & y_q \in \posvar(\{r \mid p_i \not \leq r
    \text{ for all } i \in \{1,\ldots,n\}\}) \cup \{y_{p_1}, \ldots,
    y_{p_n}\}.
  \end{arrayprf}
  \]
  Step~\pref{rng-gamma} follows from the definitions of~$\gamma$ and $G$ and
  the definition of $\pairing$ (Definition~\ref{def-pairing}).  Step \pref{yq-incl}
  follows from the precondition that for all $i \in \{1,\ldots,n\}$ it holds
  that $p_i \not < q$.
  Now, let $y$ be a positional variable and let $v$ be a variable such
  that
  \[\begin{array}{l}
  y \in \{y_{p_1},\ldots,y_{p_n},y_{p_q}\}, \text{ and}\\
  v \in \vars(y\gamma).
  \end{array}\]
  From~\pref{rng-gamma} and~\pref{yq-incl} it follows that
  \[
  \begin{arrayprf}
    \prl{v-in} & v \in \posvar(\{r \mid p_i \not \leq r
    \text{ for all } i \in \{1,\ldots,n\}\}) \cup \{y_{p_1}, \ldots,
    y_{p_n}\}.
  \end{arrayprf}
  \]
  As proven below, then
  \[
  \begin{arrayprf}
    \prl{aux-vy} & v\mu = v\rho\mu.
  \end{arrayprf}
  \]
  Step~\pref{aux-vy} is proven by considering three cases (the first two
  overlap, the third applies if none of the first two applies):
  \begin{enumerate}
  \item Case $v \notin \{y_{p_1}, \ldots, y_{p_n}\}$.  Then, by~\pref{v-in}
    and~\pref{dom-rho-2}, $v \notin \dom(\rho)$, hence $v\mu = v\rho\mu$.

  \item Case $v \in \dom(\ssubnew)$.  Then, by~\pref{pre-dom-cap}, $v \notin
    \dom(\rho)$, hence $v\mu = v\rho\mu$.

  \item Case $v \in \{y_{p_1}, \ldots, y_{p_n}\} \setminus \dom(\ssubnew)$.
    Then, by \pref{pre}, $v\gamma\mu = v\rho$.  Since $v \in \vars(y\gamma)$
    and $\gamma$ is idempotent it follows that $v = v\gamma$. Hence $v\mu =
    v\rho$, and, since $\mu$ is idempotent, $v\mu = v\rho\mu$.
  \end{enumerate}
  Given the definition of $v$ and $y$ we can instantiate~\pref{aux-vy} to the
  following statements about the $y_{p_i} \in \dom(\tau')$ for $i \in
  \{1,\ldots, n\}$ and $y_q$.
  \[
  \begin{arrayprf}
    \prl{gap-1-a} & y_{p_i}\gamma\mu\; =\; y_{p_i}\gamma\rho\mu,
    \text{ for all } i \in \{1,\ldots,n\}.\\
    \prl{gap-3} & \vy_q\gamma\mu = \vy_q\gamma\rho\mu.\\
  \end{arrayprf}
  \]

  \beforeproofpartskip
  \noindent{\textit{Part IV. Deriving the Conclusion}}
  \afterproofpartskip

  \noindent
  The conclusion of the lemma to be proven, that is, \[\ipt_{\alpha}(d,q)
  \subsumedBy \ipt_{\alpha}(d[e]_{p_1}[e]_{p_2}\ldots [e]_{p_n}, q)\]
  can be reformulated as
  \[
  \begin{arrayprf}
    \prl{goal} & y_{q}\gamma\mu \subsumedBy y_{q}\gamma\nu.
  \end{arrayprf}
  \]
  For the left side, the reformulation follows since $\ipt_{\alpha}(d,q) =
  \P(y_{q}\gamma\mu)$, which can be derived analogously to
  steps~\pref{x1}--\pref{x5}, but by applying~\pref{q-gamma-mu} instead
  of~\pref{pi-gamma-mu}. For the right side it follows since
  $\ipt_{\alpha}(d[e]_{p_1}[e]_{p_2}\ldots [e]_{p_n}, q) =
  \ipt_{\alpha}(d', q) =\linebreak \P(\vy_q\mgu(\{\pairing_{\alpha}(d',r)
  \mid r \in \pos(d')\})) \variant \P(\vy_q\gamma\nu)$, which can be derived
  by expanding definitions and, for the last step, applying
  Lemma~\ref{lem-dds-mult}.

  To prove~\pref{goal}, we need a further auxiliary statement, which is
  derived along with an intermediate step about the domain of~$\gamma$ as
  explained below.
  \[
  \begin{arrayprf}
    \prl{dom-gamma} & \dom(\gamma) \subseteq \posvar(\{r \mid p_i \not \leq r
    \text{ for all } i \in \{1,\ldots,n\}\}) \cup \{y_{p_1}, \ldots,
    y_{p_n}\}.\\
    \prl{gap-2} & y_{p_i}\ssubnew\; =\;
    y_{p_i}\ssubnew\gamma,
    \text{ for all } i \in \{1,\ldots,n\} \text{ s.th. } y_{p_i} \in \dom(\ssubnew).\\

  \end{arrayprf}
  \]
  Step~\pref{dom-gamma} follows from the definitions of~$\gamma$ and $G$ and
  the definition of $\pairing$ (Definition~\ref{def-pairing}).  Step~\pref{gap-2}
  can be shown as follows. Assume $y_{p_i} \in \dom(\ssubnew)$.
  By~\pref{idx-ssubnew-dom} it follows that $\vars(y_{p_i}\ssubnew) \subseteq
  \posvar(\{r \mid p_i < r\})$. With~\pref{dom-gamma} it follows that
  $\vars(y_{p_i}\ssubnew) \cap \dom(\gamma) = \emptyset$, which
  implies~\pref{gap-2}.  We can now proceed to prove the goal~\pref{goal} as
  follows, explained below.
  \[
  \begin{arrayprf}
    \prl{xunif-00} &  y_{p_i}\gamma\rho\mu\; =\ y_{p_i}\ssubnew\rho,
    \text{ for all } i \in \{1,\ldots,n\}.\\
    \prl{unif-00} &  y_{p_i}\gamma\rho\mu\; =\ y_{p_i}\ssubnew\gamma\rho,
    \text{ for all } i \in \{1,\ldots,n\} \text{ s.th. } y_{p_i} \in \dom(\ssubnew).\\
    \prl{unif-01} & y_{p_i}\gamma\rho\mu\mu\; =\ y_{p_i}\ssubnew\gamma\rho\mu,
    \text{ for all } i \in \{1,\ldots,n\} \text{ s.th. } y_{p_i} \in \dom(\ssubnew).\\
    \prl{unif-1-indom} & y_{p_i}\gamma\rho\mu\;
    =\ y_{p_i}\ssubnew\gamma\rho\mu,
    \text{ for all } i \in \{1,\ldots,n\} \text{ s.th. } y_{p_i} \in
    \dom(\ssubnew).\\
    \prl{unif-1-notindom} & y_{p_i}\gamma\rho\mu\;
    =\ y_{p_i}\ssubnew\gamma\rho\mu,
    \text{ for all } i \in \{1,\ldots,n\}
    \text{ s.th. } y_{p_i} \notin \dom(\ssubnew).\\
    \prl{unif-1} & y_{p_i}\gamma\rho\mu\; =\ y_{p_i}\ssubnew\gamma\rho\mu,
    \text{ for all } i \in \{1,\ldots,n\}.\\
    \prl{unif-2} & \rho\mu \text{ is a unifier of }
    \{\{\vy_{p_1}\gamma,\vy_{p_1}\ssubnew\gamma\},
    \ldots, \{\vy_{p_n}\gamma,\vy_{p_n}\ssubnew\gamma\}\}.\\
    \prl{unif-3} & \rho\mu \subsumedBy \nu.\\
    \prl{goal-0} &  y_q\gamma\rho\mu \subsumedBy y_q\gamma\nu.\\
    \prl{goal-1} &  y_q\gamma\mu \subsumedBy y_q\gamma\nu.\\
  \end{arrayprf}
  \]
  Step~\pref{xunif-00} follows from~\pref{gap-1-a} and~\pref{pre}.
  Step~\pref{unif-00} follows from~\pref{xunif-00} and~\pref{gap-2}.
  Step~\pref{unif-01} follows from \pref{unif-00}.  Step~\pref{unif-1-indom}
  follows from~\pref{unif-01} since~$\mu$ is idempotent.
  Step~\pref{unif-1-notindom} holds since if $y_{p_i} \notin \dom(\ssubnew)$,
  then $y_{p_i}\ssubnew = y_{p_i}$.  Step~\pref{unif-1} follows
  from~\pref{unif-1-indom} and~\pref{unif-1-notindom}.  Step~\pref{unif-2}
  follows from~\pref{unif-1}.  Step~\pref{unif-3} follows from~\pref{unif-2}
  and the definition of~$\nu$.  Step~\pref{goal-0} follows from~\pref{unif-3}.
  Finally, step~\pref{goal-1}, which is the goal to be proven listed above
  as~\pref{goal}, follows from~\pref{goal-0} and~\pref{gap-3}.
\end{proof}

\subsubsection{Subproof Replacement Based on IPT and MGT}

Lemma~\ref{lem-sr-mult} is now applied to justify the following two theorems,
which may be practically applied to modify proofs represented by a \dterm
together with an axiom assignment.
\begin{thm}[\IPT-Based Subproof Replacement]
  \label{thm-sr-ipt}
  Let $d, e$ be \dterms, let~$\alpha$ be an axiom assignment for $d$ and for
  $e$, and let $p_1, \ldots, p_n$, where $n \geq 0$, be positions in $\pos(d)$
  such that for all $i,j \in \{1,\ldots,n\}$ with $i \neq j$ it holds that
  $p_i \not \leq p_j$.  If for all $i \in \{1, \ldots, n\}$ it holds that
  \[\ipt_{\alpha}(d, p_i) \subsumedBy \mgt_{\alpha}(e),\] then
  \[\mgt_{\alpha}(d) \subsumedBy
  \mgt_{\alpha}(d[e]_{p_1}[e]_{p_2}\ldots [e]_{p_n}).\]
\end{thm}
\begin{proof}
  The theorem expresses the special case of Lemma~\ref{lem-sr-mult} with $q =
  \emptypos$. The precondition of that lemma that for all $i \in
  \{1,\ldots,n\}$ it holds that $p_i \not < q$ then holds trivially. The
  remaining preconditions are the same as those of Lemma~\ref{lem-sr-mult}.
  The conclusion is obtained from the conclusion of Lemma~\ref{lem-sr-mult} by
  contracting the definition of $\mgt$.
\end{proof}
Theorem~\ref{thm-sr-ipt} states that simultaneously replacing a number of
occurrences of possibly different subterms in a \dterm by the same subterm
with the property that its \MGT subsumes each of the \IPTs of the original
occurrences results in an overall \dterm whose \MGT subsumes that of the
original overall \dterm.
The following theorem is like Theorem~\ref{thm-sr-ipt}, but restricted to the
case of a single replaced subterm occurrence and with a stronger precondition,
which refers to the \MGT of that subterm instead of the
\IPT.

\pagebreak %
\begin{thm}[\MGT-Based Subproof Replacement]
  \label{thm-sr-mgt}
  Let $d, e$ be \dterms and let~$\alpha$ be an axiom assignment for $d$ and
  for $e$.  For all positions $p \in \pos(d)$ it then holds that if
  \[\mgt_{\alpha}(d|_p) \subsumedBy \mgt_{\alpha}(e),\] then
  \[\mgt_{\alpha}(d) \subsumedBy \mgt_{\alpha}(d[e]_p).\]
\end{thm}
\begin{proof}
Follows from Theorem~\ref{thm-sr-ipt} and Proposition~\ref{prop-ipt-subsumedby-mgt}.
\end{proof}

\emph{Simultaneous} replacements of subterm occurrences are essential for
reducing the compacted size of proofs according to
Theorem~\ref{thm-replace-csize-scsize}. For replacements according to
Theorem~\ref{thm-sr-mgt} these can be achieved by successive replacements of
individual occurrences. In Theorem~\ref{thm-sr-ipt} simultaneous replacements
are explicitly considered because the replacement of one occurrence according
to this theorem can invalidate the preconditions of another occurrence.
Specifically, replacing an occurrence at some position~$p_1$ may result in a
value of $\ipt_{\alpha}(d,p_2)$ for another position~$p_2$ that
\emph{subsumes} its original value such that the precondition
$\ipt_{\alpha}(d,p_2) \subsumedBy \mgt_{\alpha}(e)$ then fails. Hence,
Theorem~\ref{thm-sr-mgt} is formulated just for a \emph{single} subterm
occurrence, while in Theorem~\ref{thm-sr-ipt} simultaneous replacement of
multiple occurrences is explicitly taken into account.

The precondition of Theorem~\ref{thm-sr-mgt} is stronger then that of
Theorem~\ref{thm-sr-ipt}, permitting rewriting according to the theorem in
fewer situations. Nevertheless, Theorem~\ref{thm-sr-mgt} can be useful in
practice, in particular because its precondition $\mgt_{\alpha}(d|_p)
\subsumedBy \mgt_{\alpha}(e)$ can be evaluated on the basis of $\alpha$, $e$
and \emph{just the subterm} $d|_p$ of $d$, whereas determining
$\ipt_{\alpha}(d,p)$ for the precondition of Theorem~\ref{thm-sr-ipt} requires
also consideration of the \emph{context} of $d|_p$ in~$d$.

\subsection{Specific Reductions and Regularities}
\label{subsec-regularities}

\name{Regularity} is a well-known important device in tableau-based theorem
proving (see, e.g., \cite{handbook:ar:haehnle}): A clausal tableau is
\emph{regular} if none of its branches contains more than one occurrence of
the same literal.  Regularity is usually considered with respect to
\emph{completeness}, that is, it is shown that if there exists a proof (closed
clausal tableau), then there exists one that is regular. Intuitively, this is
justified because in a proof that is not regular, the subproof attached at the
upper occurrence of the repeated literal can be replaced by the smaller
subproof rooted at the lower one.  Hence, a non-regular proof can be
\emph{reduced} to a smaller proof by replacing a subproof.  From this point of
view, regularity is just the failure of a particular form of reducibility.
In theorem proving, proofs that are not regular can be excluded from the
search space. If the objective is to shorten given proofs, the reductions
associated with regularities can be applied.

On the basis of the tools developed in the previous sections several related
forms of reduction that are suitable for proof shortening and, viewed as
regularities, suitable for proof search can be naturally specified. Some of
these are stronger than others, with weaker ones often suggesting advantage in
ease and efficiency of implementation.

We group the considered reductions into two families, depending on whether
they are based on the replacement of a single subproof occurrence with a
subproof of itself, discussed in Sect.~\ref{subsubsec-red-rs}, or based on
the replacement of all occurrences of a subproof by proofs that are smaller
with respect to the compaction ordering, discussed in
Sect.~\ref{subsubsec-red-co}.

\subsubsection{Reductions Based on Replacement by a Subterm}
\label{subsubsec-red-rs}

We consider the following reductions based on the replacement of a single
subproof occurrence with a subproof of itself.

\begin{defn}%
  \label{def-red-single}
  Let $d$ be a \dterm and let $\alpha$ be an axiom assignment for~$d$. For
  positions $p, p' \in \pos(d)$ such that $p < p'$, we say that
  the \dterm $d' \eqdef d[d|_{p'}]_p$ \defname{is obtained from $d$ for
  $\alpha$ by}
  \smallskip

  \subdef{def-red-rsii} \defname{\IS-reduction}, if $\ipt_{\alpha}(d,p') =
  \ipt_{\alpha}(d,p)$.

  \subdef{def-red-rsmm} \defname{\MS-reduction}, if $\mgt_{\alpha}(d|_p)
  \subsumedBy \mgt_{\alpha}(d|_{p'})$.

  \subdef{def-red-rsim} \defname{\XS-reduction}, if $\ipt_{\alpha}(d,p)
  \subsumedBy \mgt_{\alpha}(d|_{p'})$.

  \smallskip
  \noindent
  The \dterm~$d$ is called \defname{$X$-reducible} (where $X$ is \IS, \MS or
  \XS) for~$\alpha$ if and only if there exist positions $p,p'$ such that
  $d[d|_{p'}]_p$ is obtained by $X$-reduction from~$d$ for~$\alpha$.
  Otherwise, $d$ is called \defname{$X$-regular}.

\end{defn}

In the names of the defined reductions, \name{I} and \name{M} indicate
characterization solely in terms of IPTs and MGTs, respectively, and \name{S}
indicates replacement of a \name{single} subproof occurrence, contrasted
with \name{C} discussed below in Sect.~\ref{subsubsec-red-co}.

\begin{examp}%
  \label{examp-red-simple}
  The \dterm $d \eqdef \D(\D(\D(1,1),1),1)$ when considered for
  \name{Syll-Simp} (see Sect.~\ref{subsec-comparing})
  as axiom is \IS-, \MS- as well as
  \XS-reducible. For all three reductions the respective positions are $p =
  \emptypos$ and $p' = 1.2$. Hence $d|_p = d$ and $d|_{p'} = 1$ and the
  \dterm~$d'$ obtained from the reduction is just~$1$. As a tree in
  indentation representation \cite[Section~2.3, Figure~20c]{knuth:1}~$d$ with
  associated MGTs and IPTs can be depicted as follows.
  {\small
    \[
    \setlength{\fboxsep}{0pt}
    \begin{array}{l}
      \trs{0}\trsbox{21}{\g{CCCpqrCqr}  \trslash  \g{CCCpqrCqr}}\\
      \trs{1}  \g{CpCCCqrsCrs}  \trslash  \g{CCCCstuCtuCCCpqrCqr}\\
      \trs{2}    \g{CpCqp}  \trslash  \g{CCCCpqrCqrCCCC\trskip{stuCtuCCCpqrCqr}}\\
      \trs{3}      \g{CCCpqrCqr}  \trslash  \g{CCCCvCCCstuCtu\trskip{CCCpqrCqrCCCCstuCtuCCCpqrCqrCCCCpqrCqrCCCCstuCtuCCCpqrCqr}}\\
      \trs{3}      \g{CCCpqrCqr}  \trslash  \g{CCCvCCCstuCtuC\trskip{CCpqrCqrCCCCstuCtuCCCpqrCqr}}\\
      \trs{2}    \trsgray{19}{\g{CCCpqrCqr}  \trslash  \g{CCCpqrCqr}}\\
      \trs{1}  \g{CCCpqrCqr}  \trslash  \g{CCCstuCtu}\\
    \end{array}
    \]\par}%
  \noindent The nodes of the tree appear here from top to bottom in the order
  in which they are visited by pre-order traversal. We represent each node at
  a position~$q$ with the argument term of the \MGT of $d|_q$ and, separated
  by a slash, the argument term of the \IPT of~$d$ at~$q$. These argument
  terms are written in \Lukasiewicz's notation. Variables are renamed to
  $p,q,r,\ldots$, (unrelated to the use of $p,q$ as symbols for positions in
  \Dterms) starting freshly in each \MGT and globally (corresponding to the
  notion of \name{rigid} variables -- see, e.g., \cite{handbook:ar:haehnle})
  for the \IPTs. Long terms are only partially presented. The nodes at
  positions~$p$ and~$p'$ are highlighted by framing and gray background,
  respectively. Observe that for~$2$ as position~$p'$, represented by the
  bottom line in the tree presentation, \MS- and \XS-reduction would also be
  applicable, but not \IS-reduction.
\end{examp}

All three reductions specified in Definition~\ref{def-red-single} effect that
a subterm occurrence $d|_p$ of~$d$ is replaced by a \dterm~$d|_{p'}$, which,
because of the precondition $p < p'$, is a strict subterm of $d|_p$.
Concerning structure, it follows that, if $d'$ is obtained from $d$ by one of
these reductions, then $\tsize(d) > \tsize(d')$, $\csize(d) \geq \csize(d')$
and $\scsize(d) \geq \scsize(d')$. Concerning the associated formulas, for all
three reductions it holds that if $\mgt_{\alpha}(d)$ is defined, then also
$\mgt_{\alpha}(d')$ is defined and $\mgt_{\alpha}(d) \subsumedBy
\mgt_{\alpha}(d')$. This subsumption relationship follows for \IS-reduction
from Theorem~\ref{thm-sr-ipt} together with
Proposition~\ref{prop-ipt-subsumedby-mgt}, for \MS-reduction directly from
Theorem~\ref{thm-sr-mgt}, and for \XS-reduction directly from
Theorem~\ref{thm-sr-ipt}.

From Proposition~\ref{prop-ipt-subsumedby-mgt} it follows that \IS-reduction
and \MS-reduction both are special cases of \XS-reduction. It can be shown
with examples that of \IS- and \MS-reduction neither one is more general than
the other. \XS-reduction, however, is \emph{strictly} more general than both
of \IS- and \MS-reduction, as demonstrated with the following example.
\begin{examp}%
\label{examp-red-xs}
The \Dterm $d \eqdef \D(\D(\D(\D(1,\D(1,1)),1),1),\D(1,1))$ when considered
for \name{{\L}ukasie\-wicz} as axiom is \XS-reducible but neither \MS- nor \IS-reducible.
Analogously to Example~\ref{examp-red-simple}, the \Dterm with associated MGTs
and IPTs can be presented as follows.
{\small
\[
\setlength{\fboxsep}{0pt}
\begin{array}{l}
    \trs{0}\g{CpCCqrCrr} \trslash \g{CpCCqrCrr}\\
\trs{1}  \trsbox{28.2}{\g{CCCpqpCrp} \trslash \g{CCCCCqrCrrCrrCCqrCrrCpCCqrCrr}}\\
\trs{2}    \g{CpCCCqrqCsq} \trslash \g{CCCCstuCCusCvs\trskip{CCCCCqrCrrCrrCCqrCrrCpCCqrCrr}}\\
\trs{3}      \trsgray{26.2}{\g{CCCpCqrCCsqCtqCuCCsqCtq} \trslash \g{CCCCCCqrCrrwCC\trskip{CqrCrrCrrCCCCCqrCrrCrrCCqrCrrCpCCqrCrrCCCCstuCCusCvsCCCCCqrCrrCrrCCqrCrrCpCCqrCrr}}}\\
\trs{4}        \g{CCCpqrCCrpCsp} \trslash \g{CCCCCCCCqrCrrC\trskip{rrCCqrCrrCpCCqrCrrCCCqrCrrCrrCCCCqrCrrwCCCqrCrrCrrCCCCCCqrCrrwCCCqrCrrCrrCCCCCqrCrrCrrCCqrCrrCpCCqrCrrCCCCstuCCusCvsCCCCCqrCrrCrrCCqrCrrCpCCqrCrr}}\\
\trs{4}        \g{CCCCpqCrqCqsCtCqs} \trslash \g{CCCCCCCqrCrrCr\trskip{rCCqrCrrCpCCqrCrrCCCqrCrrCrrCCCCqrCrrwCCCqrCrrCrr}}\\
\trs{5}          \g{CCCpqrCCrpCsp} \trslash \g{CCCCCCqrCrrCrr\trskip{CCCqrCrrCrrCCCCCqrCrrCrrCCqrCrrCpCCqrCrrCCCCCCCqrCrrCrrCCqrCrrCpCCqrCrrCCCqrCrrCrrCCCCqrCrrwCCCqrCrrCrr}}\\
\trs{5}          \g{CCCpqrCCrpCsp} \trslash \g{CCCCCqrCrrCrrC\trskip{CCqrCrrCrrCCCCCqrCrrCrrCCqrCrrCpCCqrCrr}}\\
\trs{3}      \g{CCCpqrCCrpCsp} \trslash \g{CCCCCqrCrrwCCC\trskip{qrCrrCrrCCCCCqrCrrCrrCCqrCrrCpCCqrCrr}}\\
\trs{2}    \g{CCCpqrCCrpCsp} \trslash \g{CCCstuCCusCvs}\\
\trs{1}  \g{CCCCpqCrqCqsCtCqs} \trslash \g{CCCCqrCrrCrrCC\trskip{qrCrr}}\\
\trs{2}    \g{CCCpqrCCrpCsp} \trslash \g{CCCCrrqCCqrCrr\trskip{CCCCqrCrrCrrCCqrCrr}}\\
\trs{2}    \g{CCCpqrCCrpCsp} \trslash \g{CCCrrqCCqrCrr}\\
\end{array}
\]\par}
\end{examp}

\IS-reduction corresponds to the well-known notion of regularity for
rigid-variable tableaux (see, e.g., \cite{handbook:ar:haehnle}), while
\MS-reduction corresponds to forms of regularity considered for tableaux with
universal variables such as hypertableaux \cite{hypertableaux}. The strictly
more general \XS-reduction combines aspects of both.

\subsubsection{Reductions Based on the Compaction Ordering}
\label{subsubsec-red-co}

We now turn to the second family of reductions that are based on the
replacement of all occurrences of a subproof by proofs that are smaller with
respect to the compaction ordering. The underlying justifications are
Theorems~\ref{thm-sr-ipt} and~\ref{thm-replace-csize-scsize}. We define the
following notions of reduction and regularity.

\begin{defn}%
  \label{def-red-all}

  Let $d$ be a \dterm, let $e$ be a subterm of $d$ and let $\alpha$ be an
  axiom assignment for~$d$. For \dterms $e'$, we say that the \dterm
  $d' \eqdef \replace{d}{e}{e'}$ \defname{is obtained from $d$
  for $\alpha$ by}

  \smallskip

  \subdef{def-red-racm} \defname{\MC-reduction}, if $e \gtrc e'$,
  $\mgt_{\alpha}(e')$ is defined and $\mgt_{\alpha}(e) \subsumedBy
  \mgt_{\alpha}(e')$.

  \subdef{def-red-raci} \defname{\XC-reduction}, if $e \gtrc e'$,
  $\mgt_{\alpha}(e')$ is defined, and for all positions $p \in \pos(d)$ such
  that $d|_p = e$ it holds that $\ipt_{\alpha}(d,p) \subsumedBy
  \mgt_{\alpha}(e')$.

  \smallskip
  \noindent
  The \dterm~$d$ is called \defname{$X$-reducible} (where $X$ is \MC or \XC)
  for~$\alpha$ if and only if there exists a \dterm~$e'$ such that
  $\replace{d}{e}{e'}$ is obtained by $X$-reduction from~$d$ for~$\alpha$.
  Otherwise, $d$ is called \defname{$X$-regular}.
\end{defn}

In the names of the defined reductions, \name{M} indicates, as
for \defname{\MS-reduction}, characterization solely in terms of MGTs,
and \name{C} indicates replacement based on the compaction ordering.
While \MC-reduction and \XC-reduction are similar to \MS-reduction and
\XS-reduction in that they compare two MGTs or an IPT with an MGT,
respectively, they differ from these in that they are not based on the
replacement of a single subproof by a subproof of itself, but on the
replacement of \emph{all} occurrences of a subproof by a subproof that is
smaller with respect to the compaction ordering. They aim at reducing the
compacted size. Differently from the \IS-, \MS- and \XS-reductions, they do
not transfer from subterms to containing \dterms. It is, for example, possible
that a subterm of a \dterm is \XC-reducible while the \dterm itself is not
\XC-reducible. This does not come as a surprise, because a proof with smallest
compacted size among all proofs of the same theorem may have a subproof of a
lemma that has not the smallest compacted size among all proofs of the lemma.
A possibility to implement \MC- and \XC-reduction is by enumerating the set
$\{f \mid e \gtrc f\}$ as indicated in the proof of
Proposition~\ref{prop-co-number}.

If $d'$ is obtained from~$d$ by \MC- or \XC-reduction, then by
Theorem~\ref{thm-sr-mgt} or Theorem~\ref{thm-sr-ipt}, respectively, it follows
that $\mgt_{\alpha}(d) \subsumedBy \mgt_{\alpha}(d')$. Concerning structural
properties, by Theorem~\ref{thm-replace-csize-scsize} it follows that
$\csize(d) \geq \csize(d')$ and $\scsize(d) > \scsize(d')$. Combining this
with the structural effects of the reductions from
Definition~\ref{def-red-single}, we can conclude that for all the reductions
specified in Definitions~\ref{def-red-single} and~\ref{def-red-all} it holds
that
\[\la \csize(d), \scsize(d), \tsize(d) \ra > \la \csize(d'), \scsize(d'),
\tsize(d') \ra,\] where the triples of numbers are compared lexically. Hence
any succession of replacement steps with these reductions, intermingling them
arbitrarily, terminates after a finite number of steps.

\subsection{Removing Irrelevant Minor Premises: N-Simplification}
\label{subsec-simp-n}

Proofs may involve applications of \Det where the conclusion $\P y$ is
actually independent from the minor premise $\P x$. Any axiom can then serve
as a trivial minor premise. Meredith expresses this with the symbol~$\n$ as
second argument of the respective \dterm. The following function $\simpn$
specifies a simplification of \dterms with respect to an axiom
assignment~$\alpha$ that replaces subterms with $\n$ accordingly on the basis
of the preservation of the \MGT.

\begin{defn}%
  \label{def-simp-n}
  Let $d$ be a \dterm and let $\alpha$ be an axiom assignment for~$d$. Then
  the \defname{n-simplification of~$d$ with respect to~$\alpha$} is the \dterm
  $\simpn_{\alpha}(d)$, where $\simpn$ is the following function.
  \[
  \begin{array}{l@{\hspace{0.5em}}c@{\hspace{0.8em}}ll}
    \simpn_{\alpha}(d) & \eqdef & d, \text{ if } d \text{ is a \dconstant},\\
    \simpn_{\alpha}(\D(d_1,d_2)) & \eqdef & \D(\simpn_{\alpha}(d_1),\n),
     \text{if } \mgt_{\alpha}(d_1) \text{ is a variable
       or is of}\\
     && \text{the form } \i(x,t) \text{ with } x \text{ a variable
     not in } \vars(t),\\
  \simpn_{\alpha}(\D(d_1,d_2)) & \eqdef &
  \D(\simpn_{\alpha}(d_1),\simpn_{\alpha}(d_2)), \text{ else.}
  \end{array}
  \]
\end{defn}

N-simplification preserves the \MGT of subterms in all positions, except of
those that are replaced by $\n$. That is, if $d' = \simpn_{\alpha}(d)$, then
for all positions $p \in \pos(d')$ such that $d'|_p \neq \n$ it holds
$\mgt_{\alpha}(d'|_p) \variant \mgt_{\alpha}(d|_p)$. The particular effect of
n-simplification is that occurrences of complex subterms of a \dterm may be
replaced by the \dconstant $\n$, resulting in a shortened proof. We will see
examples of the effect of n-simplification in
Sects.~\ref{subsec-prop-struct} and~\ref{subsec-prover9}.

In some applications it is undesirable to have $\n$ as a special primitive
\Dterm symbol. For example, if there is originally a single proper axiom like
\Luk, the \dterms then can have two different leaf symbols, altering
combinatoric properties such as the number of different \dterms of a given
tree size or compacted size. This can be addressed by using instead of $\n$
just other primitive \Dterms that identify an arbitrary axiom, such as the
numeral~1 in previously considered example proofs. The size reduction achieved
by n-simplification is then retained, only the explicit marking of
independence from the minor premise expressed by~$\n$ is lost. When required,
however, this marking can easily be restored with an application of
conventional n-simplification, which then has just the effect of replacing
occurrences of primitive \Dterms by~$\n$.

\section{Inspecting \Lukasiewicz's Proof and its Variation by Meredith}
\label{sec-inspecting}

As noted in Sect.~\ref{subsec-luk-shortest-axiom}, \Lukasiewicz
\cite{luk:1948} has formally proven that his axiom \Luk entails \Syll, \Peirce
and \Simp, and Meredith \cite{meredith:notes:1963} presented a variation of
this proof in his framework of CD, reproduced here as
Fig.~\ref{fig-proof-mer} (p.~\pageref{fig-proof-mer}). Can we learn
something from these proofs that helps to improve ATP? Developed with only
human resources, do they lie among the vast combinatory possibilities within
some smaller space that can be characterized by certain features, regular
patterns and size restrictions of involved components? To approach these
questions, we take a close look at these proofs, inspecting each of their
subproofs for various properties.

This section provides a comprehensive analysis of these historic proofs. It
takes into account the accumulated knowledge from nearly a century of research
as well as new insights. The latter as well as the entire analysis have been
made possible due to the formal basis established in the preceding sections.
The results of our analysis are presented in a condensed form in the
Tables~\ref{tab-bigmer} and~\ref{tab-bigluk}, which are discussed in detail
throughout this section.

\subsection{The Considered Proofs}
\label{subsec-considered-proofs}

We call the two proofs considered in this section $\DMER$ and $\DLUK$.
Basically, each proof can be understood as a set of three \Dterms, one \Dterm
for each of the goal theorems, \Syll, \Peirce and \Simp, which are proven from
the axiom \Luk. The set of the three trees is represented by a single DAG with
three roots, one for each goal theorem. The DAG represents the proofs of the
three goals simultaneously such that subproofs used for more than a single
goal can be shared. In the \TPTP the three goals appear separated as problems
\tp{LCL038-1}, \tp{LCL083-1} and \tp{LCL082-1}, respectively. On occasion, we
consider information about the modularization of the original proof
presentations by Meredith and \Lukasiewicz, which is not captured by the set
of \Dterms and the respective minimal DAG alone, but would be rendered
formally by a representation as \compDterm, as discussed in
Sect.~\ref{subsubsec-modularization}.

Proof~$\DMER$ is Meredith's variation \cite{meredith:notes:1963} of
\Lukasiewicz's proof \cite{luk:1948} and is expressed with \CD.
Figure~\ref{fig-proof-mer} reproduces the presentation by Meredith.
Proof~$\DLUK$ is a \CD proof that results from a conversion of \Lukasiewicz's
proof \cite{luk:1948}, originally expressed by the method of substitution and
detachment, with explicitly annotated formula substitutions. We first
converted \Lukasiewicz's original proof straightforwardly to \CD. In the
result the structure of the detachment applications is strictly retained,
while the formula substitutions are considered only implicitly by unification
with most general unifiers. The lemma formulas of the intermediate stages, or
``theses'' \cite{luk:1948}, are then most general theorems of the respective
subproofs. These differ slightly from \Lukasiewicz’s original theses: in most
cases both are identical modulo renaming of variables, and in some cases
\Lukasiewicz's thesis is a strict instance of the most general theorem. As a
second conversion step we applied n-simplification
(Sect.~\ref{subsec-simp-n}) to eliminate trivial redundancies.
Figure~\ref{fig-proof-luk} shows the resulting proof~$\DLUK$, in the notation
by Meredith \cite{meredith:notes:1963}, arranged into intermediate steps that
match \Lukasiewicz's original presentation. Figure~\ref{fig-lukcd-ordering}
shows the label dependency ordering of that proof.

\begin{figure}
  \small
  \centering
\begin{tabular}{r@{\hspace{0.5em}}L{24em}@{\hspace{0.0em}}R{3em}@{\hspace{0.5em}}R{3em}}  
1. & $\g{CCCpqrCCrpCsp}$ & \luk{1} & \mer{1}\\
2. & $\g{CCCCpqCrqCqsCtCqs} \mereq \f{D}11$ & \subsumesluk{2} & \\
3. & $\g{CCCpCqrCCsqCtqCuCCsqCtq} \mereq \f{D}12$ & \subsumesluk{3} & \\
4. & $\g{CCCpqpCrp} \mereq \f{D}\f{D}31\mathrm{n}$ & \luk{4} & \mer{2}\\
5. & $\g{CCCpqCqrCsCqr} \mereq \f{D}14$ & \luk{5} & \\
6. & $\g{CCCpCqrCsqCtCsq} \mereq \f{D}15$ & \luk{6} & \\
7. & $\g{CCCpCqrCsCrtCuCsCrt} \mereq \f{D}16$ & \luk{7} & \\
8. & $\g{CCCpqrCqr} \mereq \f{D}\f{D}71\mathrm{n}$ & \luk{8} & \mer{3}\\
9. & $\g{CpCCpqCrq} \mereq \f{D}81$ & \luk{9} & \mer{4}\\
10. & $\g{CCCCCpqrCsrpCtp} \mereq \f{D}19$ & \luk{10} & \\
11. & $\g{CCCpqCCCqrsCtsCuCCCqrsCts} \mereq \f{D}1.10$ & \luk{11} & \\
12. & $\g{CCCpCCCqrsCtsCuqCvCuq} \mereq \f{D}1.11$ & \luk{12} & \\
13. & $\g{CCCpCqrCsCCCrtuCvuCwCsCCCrtuCvu} \mereq \f{D}1.12$ & \luk{13} & \\
14. & $\g{CCCpqCrsCCCqtsCrs} \mereq \f{D}\f{D}13.1\mathrm{n}$ & \luk{14} & \mer{5}\\
15. & $\g{CCCpqCrsCCpsCrs} \mereq \f{D}14.1$ & \luk{15} & \mer{6}\\
16. & $\g{CCpCqrCCCpsrCqr} \mereq \f{D}15.9$ & \luk{16} & \mer{7}\\
17. & $\g{CCCCCpqrsCtpCCrpCtp} \mereq \f{D}16.1$ & \luk{17} & \mer{8}\\
18. & $\g{CCCCpqCrqCCCqsptCuCCCqspt} \mereq \f{D}1.17$ & \luk{18} & \mer{10}\\
19. & $\g{CCCCpqrCsqCCCqtsCpq} \mereq \f{D}\f{D}18.18.\mathrm{n}$ & \subsumesluk{19} & \mer{11}\\
20. & $\g{CCCCpqrCsqCCCqtpCsq} \mereq \f{D}14.19$ & \subsumesluk{20} & \mer{12}\\
21. & $\g{CCCCpqrsCCsqCpq} \mereq \f{D}20.15$ & \subsumesluk{21} & \mer{13}\\
22. & $\g{Cpp} \mereq \f{D}\f{D}54\mathrm{n}$ & \luk{22} & \subsumesmer{9}\\
23. & $\g{CCCpqrCCrpp} \mereq \f{D}20.22$ & \luk{23} & \mer{14}\\
24. & $\g{CpCCpqq} \mereq \f{D}8.23$ & \luk{24} & \mer{15}\\
25. & $\g{CCpqCCCprqq} \mereq \f{D}15.24$ & \luk{25} & \mer{16}\\
26. & $\g{CCCCpqrCCCpsqqCCCpsqq} \mereq \f{D}25.25$ & \subsumesluk{26} & \mer{16\pr}\\
* 27. & $\g{CpCqp} \mereq \f{D}88$ & \luk{27} & \mer{19}\\
* 28. & $\g{CCCpqpp} \mereq \f{D}25.22$ & \luk{28} & \mer{18}\\
* 29. & $\g{CCpqCCqrCpr} \mereq \f{D}\f{D}21.26.21$ & \luk{29} & \mer{17}\\
\end{tabular}
  \vspace{6pt}
  \caption{The proof~$\DLUK$, that is, \Lukasiewicz's proof from
    his 1948 paper \cite{luk:1948} converted to \CD and n-simplified, broken into 29 steps
    matching the original presentation. The notation follows Meredith
    \cite{meredith:notes:1963}. The shown lemmas are MGTs of the respective
    subproofs. The two right columns indicate corresponding proof steps in
    \Lukasiewicz's original presentation \cite{luk:1948} and in Meredith's
    variation \cite{meredith:notes:1963}, reproduced here as
    Fig.~\ref{fig-proof-mer} (p.~\pageref{fig-proof-mer}). The meaning of
    the labels is explained in Sect.~\ref{subsubsec-mer-luk-labels}.}
  \label{fig-proof-luk}
\end{figure}

\begin{figure}
  \centering
  \begin{tikzpicture}[>=latex',line join=bevel,scale=0.455]
      \pgfsetlinewidth{1bp}
\pgfsetcolor{black}
  \draw [] (29.891bp,44.0bp) .. controls (27.385bp,44.0bp) and (24.698bp,44.0bp)  .. (22.188bp,44.0bp);
  \draw [] (59.891bp,44.0bp) .. controls (57.385bp,44.0bp) and (54.698bp,44.0bp)  .. (52.188bp,44.0bp);
  \draw [] (89.891bp,44.0bp) .. controls (87.385bp,44.0bp) and (84.698bp,44.0bp)  .. (82.188bp,44.0bp);
  \draw [] (119.89bp,44.0bp) .. controls (117.39bp,44.0bp) and (114.7bp,44.0bp)  .. (112.19bp,44.0bp);
  \draw [] (149.89bp,58.696bp) .. controls (147.39bp,56.638bp) and (144.7bp,54.43bp)  .. (142.19bp,52.369bp);
  \draw [] (179.89bp,67.0bp) .. controls (177.39bp,67.0bp) and (174.7bp,67.0bp)  .. (172.19bp,67.0bp);
  \draw [] (209.89bp,67.0bp) .. controls (207.39bp,67.0bp) and (204.7bp,67.0bp)  .. (202.19bp,67.0bp);
  \draw [] (240.52bp,102.72bp) .. controls (237.6bp,97.088bp) and (234.39bp,90.892bp)  .. (231.47bp,85.263bp);
  \draw [] (269.89bp,120.36bp) .. controls (267.39bp,120.45bp) and (264.7bp,120.55bp)  .. (262.19bp,120.64bp);
  \draw [] (299.89bp,120.0bp) .. controls (297.39bp,120.0bp) and (294.7bp,120.0bp)  .. (292.19bp,120.0bp);
  \draw [] (329.89bp,120.0bp) .. controls (327.39bp,120.0bp) and (324.7bp,120.0bp)  .. (322.19bp,120.0bp);
  \draw [] (359.89bp,120.0bp) .. controls (357.39bp,120.0bp) and (354.7bp,120.0bp)  .. (352.19bp,120.0bp);
  \draw [] (389.89bp,120.0bp) .. controls (387.39bp,120.0bp) and (384.7bp,120.0bp)  .. (382.19bp,120.0bp);
  \draw [] (419.89bp,120.0bp) .. controls (417.39bp,120.0bp) and (414.7bp,120.0bp)  .. (412.19bp,120.0bp);
  \draw [] (449.89bp,120.0bp) .. controls (447.39bp,120.0bp) and (444.7bp,120.0bp)  .. (442.19bp,120.0bp);
  \draw [] (479.89bp,120.0bp) .. controls (477.39bp,120.0bp) and (474.7bp,120.0bp)  .. (472.19bp,120.0bp);
  \draw [] (509.89bp,119.36bp) .. controls (507.39bp,119.45bp) and (504.7bp,119.55bp)  .. (502.19bp,119.64bp);
  \draw [] (539.89bp,119.0bp) .. controls (537.39bp,119.0bp) and (534.7bp,119.0bp)  .. (532.19bp,119.0bp);
  \draw [] (569.89bp,118.36bp) .. controls (567.39bp,118.45bp) and (564.7bp,118.55bp)  .. (562.19bp,118.64bp);
  \draw [] (659.91bp,118.0bp) .. controls (643.02bp,118.0bp) and (609.23bp,118.0bp)  .. (592.25bp,118.0bp);
  \draw [] (539.79bp,37.224bp) .. controls (528.57bp,30.439bp) and (509.76bp,21.0bp)  .. (492.0bp,21.0bp) .. controls (492.0bp,21.0bp) and (492.0bp,21.0bp)  .. (190.0bp,21.0bp) .. controls (172.24bp,21.0bp) and (153.43bp,30.439bp)  .. (142.21bp,37.224bp);
  \draw [] (601.14bp,79.035bp) .. controls (597.87bp,85.689bp) and (594.17bp,93.226bp)  .. (590.9bp,99.887bp);
  \draw [] (599.9bp,58.039bp) .. controls (589.38bp,54.957bp) and (572.99bp,50.153bp)  .. (562.38bp,47.042bp);
  \draw [] (629.89bp,61.0bp) .. controls (627.39bp,61.0bp) and (624.7bp,61.0bp)  .. (622.19bp,61.0bp);
  \draw [] (659.89bp,61.0bp) .. controls (657.39bp,61.0bp) and (654.7bp,61.0bp)  .. (652.19bp,61.0bp);
  \draw [] (689.89bp,68.029bp) .. controls (687.39bp,67.044bp) and (684.7bp,65.988bp)  .. (682.19bp,65.002bp);
  \draw [] (239.89bp,67.0bp) .. controls (237.39bp,67.0bp) and (234.7bp,67.0bp)  .. (232.19bp,67.0bp);
  \draw [] (689.89bp,33.524bp) .. controls (687.39bp,37.373bp) and (684.7bp,41.5bp)  .. (682.19bp,45.354bp);
  \draw [] (719.9bp,99.006bp) .. controls (709.38bp,103.18bp) and (692.99bp,109.67bp)  .. (682.38bp,113.88bp);
  \draw [] (719.89bp,86.696bp) .. controls (717.39bp,84.638bp) and (714.7bp,82.43bp)  .. (712.19bp,80.369bp);
\begin{scope}
  \definecolor{strokecol}{rgb}{0.0,0.0,0.0}
  \pgfsetstrokecolor{strokecol}
  \draw (11.0bp,44.0bp) node {1};
\end{scope}
\begin{scope}
  \definecolor{strokecol}{rgb}{0.0,0.0,0.0}
  \pgfsetstrokecolor{strokecol}
  \draw (41.0bp,44.0bp) node {2};
\end{scope}
\begin{scope}
  \definecolor{strokecol}{rgb}{0.0,0.0,0.0}
  \pgfsetstrokecolor{strokecol}
  \draw (71.0bp,44.0bp) node {3};
\end{scope}
\begin{scope}
  \definecolor{strokecol}{rgb}{0.0,0.0,0.0}
  \pgfsetstrokecolor{strokecol}
  \draw (101.0bp,44.0bp) node {4};
\end{scope}
\begin{scope}
  \definecolor{strokecol}{rgb}{0.0,0.0,0.0}
  \pgfsetstrokecolor{strokecol}
  \draw (131.0bp,44.0bp) node {5};
\end{scope}
\begin{scope}
  \definecolor{strokecol}{rgb}{0.0,0.0,0.0}
  \pgfsetstrokecolor{strokecol}
  \draw (161.0bp,67.0bp) node {6};
\end{scope}
\begin{scope}
  \definecolor{strokecol}{rgb}{0.0,0.0,0.0}
  \pgfsetstrokecolor{strokecol}
  \draw (191.0bp,67.0bp) node {7};
\end{scope}
\begin{scope}
  \definecolor{strokecol}{rgb}{0.0,0.0,0.0}
  \pgfsetstrokecolor{strokecol}
  \draw (221.0bp,67.0bp) node {8};
\end{scope}
\begin{scope}
  \definecolor{strokecol}{rgb}{0.0,0.0,0.0}
  \pgfsetstrokecolor{strokecol}
  \draw (251.0bp,121.0bp) node {9};
\end{scope}
\begin{scope}
  \definecolor{strokecol}{rgb}{0.0,0.0,0.0}
  \pgfsetstrokecolor{strokecol}
  \draw (281.0bp,120.0bp) node {10};
\end{scope}
\begin{scope}
  \definecolor{strokecol}{rgb}{0.0,0.0,0.0}
  \pgfsetstrokecolor{strokecol}
  \draw (311.0bp,120.0bp) node {11};
\end{scope}
\begin{scope}
  \definecolor{strokecol}{rgb}{0.0,0.0,0.0}
  \pgfsetstrokecolor{strokecol}
  \draw (341.0bp,120.0bp) node {12};
\end{scope}
\begin{scope}
  \definecolor{strokecol}{rgb}{0.0,0.0,0.0}
  \pgfsetstrokecolor{strokecol}
  \draw (371.0bp,120.0bp) node {13};
\end{scope}
\begin{scope}
  \definecolor{strokecol}{rgb}{0.0,0.0,0.0}
  \pgfsetstrokecolor{strokecol}
  \draw (401.0bp,120.0bp) node {14};
\end{scope}
\begin{scope}
  \definecolor{strokecol}{rgb}{0.0,0.0,0.0}
  \pgfsetstrokecolor{strokecol}
  \draw (431.0bp,120.0bp) node {15};
\end{scope}
\begin{scope}
  \definecolor{strokecol}{rgb}{0.0,0.0,0.0}
  \pgfsetstrokecolor{strokecol}
  \draw (461.0bp,120.0bp) node {16};
\end{scope}
\begin{scope}
  \definecolor{strokecol}{rgb}{0.0,0.0,0.0}
  \pgfsetstrokecolor{strokecol}
  \draw (491.0bp,120.0bp) node {17};
\end{scope}
\begin{scope}
  \definecolor{strokecol}{rgb}{0.0,0.0,0.0}
  \pgfsetstrokecolor{strokecol}
  \draw (521.0bp,119.0bp) node {18};
\end{scope}
\begin{scope}
  \definecolor{strokecol}{rgb}{0.0,0.0,0.0}
  \pgfsetstrokecolor{strokecol}
  \draw (551.0bp,119.0bp) node {19};
\end{scope}
\begin{scope}
  \definecolor{strokecol}{rgb}{0.0,0.0,0.0}
  \pgfsetstrokecolor{strokecol}
  \draw (581.0bp,118.0bp) node {20};
\end{scope}
\begin{scope}
  \definecolor{strokecol}{rgb}{0.0,0.0,0.0}
  \pgfsetstrokecolor{strokecol}
  \draw (671.0bp,118.0bp) node {21};
\end{scope}
\begin{scope}
  \definecolor{strokecol}{rgb}{0.0,0.0,0.0}
  \pgfsetstrokecolor{strokecol}
  \draw (551.0bp,44.0bp) node {22};
\end{scope}
\begin{scope}
  \definecolor{strokecol}{rgb}{0.0,0.0,0.0}
  \pgfsetstrokecolor{strokecol}
  \draw (611.0bp,61.0bp) node {23};
\end{scope}
\begin{scope}
  \definecolor{strokecol}{rgb}{0.0,0.0,0.0}
  \pgfsetstrokecolor{strokecol}
  \draw (641.0bp,61.0bp) node {24};
\end{scope}
\begin{scope}
  \definecolor{strokecol}{rgb}{0.0,0.0,0.0}
  \pgfsetstrokecolor{strokecol}
  \draw (671.0bp,61.0bp) node {25};
\end{scope}
\begin{scope}
  \definecolor{strokecol}{rgb}{0.0,0.0,0.0}
  \pgfsetstrokecolor{strokecol}
  \draw (701.0bp,72.0bp) node {26};
\end{scope}
\begin{scope}
  \definecolor{strokecol}{rgb}{0.0,0.0,0.0}
  \pgfsetstrokecolor{strokecol}
  \draw (251.0bp,67.0bp) node {27};
\end{scope}
\begin{scope}
  \definecolor{strokecol}{rgb}{0.0,0.0,0.0}
  \pgfsetstrokecolor{strokecol}
  \draw (701.0bp,18.0bp) node {28};
\end{scope}
\begin{scope}
  \definecolor{strokecol}{rgb}{0.0,0.0,0.0}
  \pgfsetstrokecolor{strokecol}
  \draw (731.0bp,95.0bp) node {29};
\end{scope}
  \end{tikzpicture}%
  \vspace{2pt}
  \caption{The label dependency ordering $<_{\compd}$ of \Lukasiewicz's proof
    \cite{luk:1948}. In other words, the label dependency ordering of
    proof~$\DLUK$ when viewed as \compDterm~$\compd$ with a domain that corresponds
    to \Lukasiewicz's original presentation.}
  \label{fig-lukcd-ordering}  
\end{figure}

\subsection{Examined Properties}
\label{subsec-properties}

In the following subsections we are going to inspect each subproof of $\DMER$
and $\DLUK$ in view of various properties. We use the term \name{property}
there informally in a generic sense. More precisely, we consider properties of
the subproof's structure, properties of the formula proven by the subproof and
properties which take contexts into account, specifically the embedding of
occurrences of the subproof into the overall proof and global contexts such as
other proofs of the formula proven by the subproof or uses of this formula in
the relevant literature.

The considered properties can be grouped into several families. We start with
discussing aspects around labeling and naming: which lemmas are explicitly
exposed and which are taken as implicit intermediate step; what cross
correspondences are among proofs and with formulas well-known in other
contexts. Next we examine structural properties of the \DTerm and then
syntactic properties of the MGTs, that is, the lemmas proven by the subproofs.
We continue with considering properties that relate a subproof to all possible
proofs of its MGT, for example, to compare with a minimal \Dterm measure such
as compacted size required to prove the MGT. Then we will look at regularity
properties as discussed in Sect.~\ref{subsec-regularities} and finally at
properties of the \IPTs, which are associated with each occurrence of a
subproof in the overall proof when viewed as tree.

Values of the properties are shown for $\DMER$ in Table~\ref{tab-bigmer} and
for $\DLUK$ in Table~\ref{tab-bigluk}. These tables contain a row for each
distinct subproof, that is, for each subterm of at least one of the three
\Dterms corresponding to the three goal theorems. Even if a subproof is
referenced multiple times in the proof, it is represented just by a single
line. The number of rows is thus the compacted size of the set of the three
\Dterms, plus one for the axiom.

In the following subsections we specify these properties and discuss their
values for $\DMER$ and $\DLUK$. Because in both proofs we have \Luk as a
single axiom, we quietly assume the corresponding axiom assignments, not
distinguishing between proofs and \Dterms, their structural component. For
reference and use in the table headers, each property is given a short
identifier.

\subsection{Labels and Names of Formulas}

Properties concerning labels and names of formulas refer to the MGT proven by
the subproof, independently of the subproof itself. We consider the
concordance with the presentations by Meredith and \Lukasiewicz as well as
appearances of the MGT formula in the literature, independently from the
particular proofs.

\attitem{\xatt{MER}, \xatt{\LUK}: Corresponding Step in Meredith's and
  \Lukasiewicz's Proof Presentation}
\label{subsubsec-mer-luk-labels}

Properties \xatt{MER} and \xatt{\LUK} show the number of the corresponding
step in the original proof presentations by Meredith
\cite{meredith:notes:1963}, indicated by the prefix \textit{M}, and
\Lukasiewicz \cite{luk:1948}, indicated by the prefix \textit{\L}. In some
cases the referenced lemma in the original presentation is not the MGT but
just a strict instance of the MGT, which is indicated by prefixing the
reference with $\strictlySubsumes$.\footnote{The use of the symbol
$\strictlySubsumes$ here is an adaptation of $t \subsumedBy s$, which stands
for $s$ \name{subsumes}~$t$ (Sect.~\ref{subsec-notation}).}

Through presence or absence of an entry in the table, the respective columns
show for which of the subproofs the proof goal is explicitly displayed as a
lemma in the original presentation and which are considered just as implicit
``unnamed'' intermediate steps, as discussed in
Sect.~\ref{subsubsec-modularization}.

To indicate a specific correspondence of two steps in $\DMER$ and $\DLUK$ we
use \mer{16\pr} as an additional reference to the most general theorem of
subproof $\f{D}16.16$ of the proof of step~17 in $\DMER$
(Fig.~\ref{fig-proof-mer}, p.~\pageref{fig-proof-mer}), which is strictly
more general than \Lukasiewicz's \cite{luk:1948} thesis~26. \mer{16\pr}
appears as the most general theorem of subproofs~$30$ and~$31$ in
Tables~\ref{tab-bigmer} and~\ref{tab-bigluk}, respectively.

The cross reference columns $\xatt{\LUK}$ in Table~\ref{tab-bigmer} and
$\xatt{MER}$ in \ref{tab-bigluk} include gray bullets to indicate that the MGT
of the respective row is also the MGT of some subproof of the referenced
proof, but is not made explicit with a label there. Actually all fields of the
cross reference columns that do not contain a label are filled with a gray
bullet, with exception of subproof 26 of $\DLUK$, whose MGT does not appear in
$\DMER$.

\begin{sidewaystable}[p]
  \centering
  \small
  \caption{Properties of all subproofs of $\DMER$ (Fig.~\ref{fig-proof-mer}),
    Meredith's variation \cite{meredith:notes:1963} of \Lukasiewicz's proof
    \cite{luk:1948}.}
  \label{tab-bigmer}
  \setlength{\tabcolsep}{3.1pt}
  \renewcommand{\arraystretch}{0.89}
  \rowcolors{2}{gray!20}{white}
  \scalebox{0.97}{\begin{tabular}{rlrrrrrrrrccccrrrcrrccrrrrr}
 & \tabatt{} & \tabatt{MER} & \tabatt{{\L}UK} & \tabatt{NN} & \tabatt{DC} & \tabatt{DT} & \tabatt{DH} & \tabatt{DI} & \tabatt{DR} & \tabatt{DS} & \tabatt{DP} & \tabatt{DK$_{L}$} & \tabatt{DK$_{R}$} & \tabatt{FC} & \tabatt{FT} & \tabatt{FH} & \tabatt{FO} & \tabatt{MC} & \tabatt{MT} & \tabatt{RS} & \tabatt{RC} & \tabatt{IT$_{U}$} & \tabatt{IT$_{M}$} & \tabatt{IH$_{U}$} & \tabatt{IH$_{M}$} & \tabatt{}\\\midrule
1. & $1$ & \mer{1} & \luk{1} & \spec{1} & 0 & 0 & 0 & 17 & 554 & -- & \textbullet & 0 & 0 & 6 & 6 & 3 & \textbullet & 0 & 0 & \textbullet & \textbullet & 4451 & 203 & 18 & 11 & 1.\\
2. & $\f{D}11$ &  & \subsumesluk{2} &  & 1 & 1 & 1 & 1 & 45 & $1{=}1$ & \textbullet & 1 & 1 & 7 & 8 & 4 & \textbullet & 1 & 1 & \textbullet & \textbullet & 1640 & 220 & 17 & 12 & 2.\\
3. & $\f{D}12$ &  & \subsumesluk{3} &  & 2 & 2 & 2 & 1 & 45 & $1\lhd$ & \textbullet & 1 & 2 & 8 & 11 & 4 & \textbullet & 2 & 2 & \textbullet & \textbullet & 1881 & 252 & 17 & 12 & 3.\\
4. & $\f{D}31$ &  & \textcolor{black!30}{\textbullet} &  & 3 & 3 & 3 & 1 & 45 & $\rhd1$ & \textbullet & 2 & 2 & 5 & 5 & 4 & \textcolor{black!30}{\textbullet} & 3 & 3 & \textbullet & \textbullet & 689 & 92 & 16 & 11 & 4.\\
5. & $\f{D}4\mathrm{n}$ & \mer{2} & \luk{4} &  & 4 & 4 & 4 & 1 & 45 & $\rhd\n$ & \textbullet & 3 & 2 & 4 & 4 & 3 & \textbullet & 4 & 4 & \textbullet & \textbullet & 688 & 91 & 15 & 10 & 5.\\
6. & $\f{D}15$ &  & \luk{5} &  & 5 & 5 & 5 & 1 & 45 & $1\lhd$ & \textbullet & 3 & 2 & 5 & 6 & 3 & \textbullet & 5 & 5 & \textbullet & \textbullet & 1667 & 198 & 15 & 10 & 6.\\
7. & $\f{D}16$ &  & \luk{6} &  & 6 & 6 & 6 & 1 & 45 & $1\lhd$ & \textbullet & 3 & 3 & 6 & 7 & 4 & \textbullet & 6 & 6 & \textbullet & \textbullet & 1802 & 208 & 16 & 11 & 7.\\
8. & $\f{D}17$ &  & \luk{7} &  & 7 & 7 & 7 & 1 & 45 & $1\lhd$ & \textbullet & 3 & 4 & 7 & 9 & 4 & \textbullet & 7 & 7 & \textbullet & \textbullet & 2648 & 303 & 16 & 11 & 8.\\
9. & $\f{D}81$ &  & \textcolor{black!30}{\textbullet} & \spec{2} & 8 & 8 & 8 & 1 & 45 & $\rhd1$ & \textbullet & 3 & 4 & 5 & 5 & 4 & \textcolor{black!30}{\textbullet} & 8 & 8 & \textbullet & \textbullet & 1032 & 119 & 15 & 10 & 9.\\
10. & $\f{D}9\mathrm{n}$ & \mer{3} & \luk{8} & \spec{3} & 9 & 9 & 9 & 5 & 45 & $\rhd\n$ & \textbullet & 3 & 4 & 4 & 4 & 3 & \textbullet & 9 & 9 & \textbullet & \textbullet & 1031 & 118 & 14 & 9 & 10.\\
11. & $\f{D}10.1$ & \mer{4} & \luk{9} & \spec{4} & 10 & 10 & 10 & 2 & 37 & $\rhd1$ & \textbullet & 4 & 4 & 4 & 4 & 3 & \textbullet & 10 & 10 & \textbullet & \textbullet & 448 & 60 & 13 & 9 & 11.\\
12. & $\f{D}1.11$ &  & \luk{10} &  & 11 & 11 & 11 & 1 & 23 & $1\lhd$ & \textbullet & 4 & 4 & 7 & 7 & 5 & \textbullet & 11 & 11 & \textbullet & \textbullet & 498 & 73 & 14 & 10 & 12.\\
13. & $\f{D}1.12$ &  & \luk{11} &  & 12 & 12 & 12 & 1 & 23 & $1\lhd$ & \textbullet & 4 & 4 & 8 & 12 & 5 & \textbullet & 12 & 12 & \textbullet & \textbullet & 1157 & 168 & 14 & 10 & 13.\\
14. & $\f{D}1.13$ &  & \luk{12} &  & 13 & 13 & 13 & 1 & 23 & $1\lhd$ & \textbullet & 4 & 4 & 9 & 10 & 6 & \textbullet & $\iv{12}{13}$ & 13 & \textbullet & \textbullet & 1050 & 159 & 15 & 11 & 14.\\
15. & $\f{D}1.14$ &  & \luk{13} &  & 14 & 14 & 14 & 1 & 23 & $1\lhd$ & \textbullet & 4 & 5 & 10 & 15 & 6 & \textbullet & $\iv{12}{14}$ & 14 & \textbullet & \textbullet & 1657 & 246 & 15 & 11 & 15.\\
16. & $\f{D}15.1$ &  & \textcolor{black!30}{\textbullet} &  & 15 & 15 & 15 & 1 & 23 & $\rhd1$ & \textbullet & 4 & 5 & 8 & 9 & 5 & \textcolor{black!30}{\textbullet} & $\iv{12}{15}$ & 15 & \textbullet & \textbullet & 684 & 100 & 14 & 10 & 16.\\
17. & $\f{D}16.\mathrm{n}$ & \mer{5} & \luk{14} &  & 16 & 16 & 16 & 2 & 23 & $\rhd\n$ & \textbullet & 4 & 5 & 7 & 8 & 4 & \textbullet & $\iv{12}{16}$ & 16 & \textbullet & \textbullet & 683 & 99 & 13 & 9 & 17.\\
18. & $\f{D}17.1$ & \mer{6} & \luk{15} &  & 17 & 17 & 17 & 3 & 18 & $\rhd1$ & \textbullet & 4 & 5 & 6 & 7 & 3 & \textbullet & $\iv{12}{17}$ & 17 & \textbullet & \textbullet & 395 & 56 & 12 & 8 & 18.\\
19. & $\f{D}18.11$ & \mer{7} & \luk{16} &  & 18 & 28 & 18 & 1 & 14 & $\rhd$ & -- & 5 & 5 & 6 & 7 & 4 & \textbullet & $\iv{12}{14}$ & 14 & \textbullet & \textbullet & 209 & 61 & 11 & 9 & 19.\\
20. & $\f{D}19.1$ & \mer{8} & \luk{17} &  & 19 & 29 & 19 & 2 & 14 & $\rhd1$ & -- & 6 & 5 & 8 & 9 & 5 & \textbullet & $\iv{12}{15}$ & 15 & \textbullet & \textbullet & 132 & 38 & 10 & 8 & 20.\\
21. & $\f{D}1.20$ & \mer{10} & \luk{18} &  & 20 & 30 & 20 & 2 & 10 & $1\lhd$ & -- & 6 & 5 & 9 & 12 & 5 & \textbullet & $\iv{12}{16}$ & 16 & \textbullet & \textbullet & 158 & 47 & 10 & 8 & 21.\\
22. & $\f{D}21.21$ &  & \textcolor{black!30}{\textbullet} &  & 21 & 61 & 21 & 1 & 5 & $=$ & -- & 6 & 5 & 9 & 10 & 5 & \textcolor{black!30}{\textbullet} & $\iv{12}{17}$ & $\iv{24}{33}$ & \textbullet & \textbullet & 53 & 16 & 9 & 7 & 22.\\
23. & $\f{D}22.\mathrm{n}$ & \mer{11} & \subsumesluk{19} &  & 22 & 62 & 22 & 1 & 5 & $\rhd\n$ & -- & 6 & 5 & 8 & 9 & 4 & \textbullet & $\iv{12}{18}$ & $\iv{24}{34}$ & \textbullet & \textbullet & 52 & 15 & 8 & 6 & 23.\\
24. & $\f{D}17.23$ & \mer{12} & \subsumesluk{20} &  & 23 & 79 & 23 & 2 & 5 & $\lhd$ & -- & 6 & 5 & 8 & 9 & 4 & \textbullet & $\iv{12}{23}$ & $\iv{23}{51}$ & \textbullet & \textbullet & 57 & 16 & 7 & 5 & 24.\\
25. & $\f{D}24.18$ & \mer{13} & \subsumesluk{21} &  & 24 & 97 & 24 & 2 & 2 & $\rhd$ & -- & 6 & 5 & 6 & 7 & 4 & \textbullet & $\iv{12}{24}$ & $\iv{24}{69}$ & \textbullet & \textbullet & 27 & 17 & 6 & 5 & 25.\\
26. & $\f{D}20.10$ & \mer{9} & \subsumedByluk{22} & \spec{5} & 20 & 39 & 20 & 2 & 4 & $\rhd$ & -- & 7 & 5 & 2 & 3 & 2 & \textbullet & 6 & 8 & \textbullet & -- & 27 & 7 & 6 & 4 & 26.\\
27. & $\f{D}24.26$ & \mer{14} & \luk{23} & \spec{8} & 25 & 119 & 24 & 2 & 3 & $>_{\mathrm{c}}$ & -- & 7 & 5 & 5 & 5 & 3 & \textbullet & $\iv{13}{25}$ & $\iv{24}{91}$ & \textbullet & \textbullet & 24 & 7 & 6 & 4 & 27.\\
28. & $\f{D}10.27$ & \mer{15} & \luk{24} & \spec{9} & 26 & 129 & 25 & 1 & 2 & $\lhd$ & -- & 7 & 5 & 3 & 3 & 3 & \textbullet & $\iv{13}{26}$ & $\iv{24}{101}$ & \textbullet & \textbullet & 19 & 12 & 6 & 5 & 28.\\
29. & $\f{D}18.28$ & \mer{16} & \luk{25} & \spec{10} & 27 & 147 & 26 & 2 & 2 & $\lhd$ & -- & 7 & 5 & 5 & 5 & 4 & \textbullet & $\iv{12}{26}$ & $\iv{24}{36}$ & \textbullet & \textbullet & 19 & 12 & 6 & 5 & 29.\\
30. & $\f{D}29.29$ & \mer{16\pr} & \subsumesluk{26} &  & 28 & 295 & 27 & 1 & 1 & $=$ & -- & 7 & 6 & 7 & 10 & 5 & \textbullet & $\iv{12}{27}$ & $\iv{24}{239}$ & \textbullet & \textbullet & 13 & 13 & 5 & 5 & 30.\\
31. & $\f{D}25.30$ &  & \textcolor{black!30}{\textbullet} &  & 30 & 393 & 28 & 1 & 1 & $<_{\mathrm{c}}$ & -- & 7 & 7 & 7 & 7 & 5 & \textbullet & $\iv{12}{29}$ & $\iv{23}{121}$ & \textbullet & \textbullet & 13 & 13 & 5 & 5 & 31.\\
32. & $\f{D}31.25$ & \mer{17} & \luk{29} & \spec{11} & 31 & 491 & 29 & 0 & 1 & $\rhd$ & -- & 7 & 7 & 5 & 5 & 3 & \textbullet & $\iv{12}{22}$ & $\iv{24}{64}$ & \textbullet & \textbullet & 5 & 5 & 3 & 3 & 32.\\
33. & $\f{D}27.26$ & \mer{18} & \luk{28} & \spec{12} & 26 & 159 & 25 & 0 & 1 & $\rhd$ & -- & 7 & 5 & 3 & 3 & 3 & \textbullet & 11 & 15 & -- & \textbullet & 3 & 3 & 3 & 3 & 33.\\
34. & $\f{D}10.10$ & \mer{19} & \luk{27} & \spec{13} & 10 & 19 & 10 & 0 & 1 & $=$ & -- & 4 & 4 & 2 & 2 & 2 & \textbullet & 6 & 7 & \textbullet & \textbullet & 2 & 2 & 2 & 2 & 34.\\
\end{tabular}
}
\end{sidewaystable}

\attitem{\xatt{NN}: Pointer to Nicknames if it is a  Generally Often Used Formula}
\label{subsubsec-xatt-nicknames}

As noted in Sect.~\ref{sec-background}, the \Lukasiewicz school introduced
nicknames for important and often referenced formulas
\cite[p.~319]{prior:formal:logic:1962}, \cite{ulrich:legacy:2001}. With regard
to ATP it was conjectured that these may be of special importance for guiding
proof search \cite[p.~112]{ulrich:legacy:2001}. Table~\ref{tab-formula-names}
lists all those MGTs of subproofs of $\DMER$ and $\DLUK$ that are known under
such a name.
The names considered there include those collected by Ulrich
\cite{ulrich:legacy:2001}. In addition, beyond the combinators appearing
already there, also short combinator terms with
well-known combinators are listed here as names of their principal
type-scheme.\footnote{The \label{footnote-comb-labels} precise restriction was
to terms formed from up to five occurrences of the well-known combinators
$\comb{I},\comb{K},\comb{B},\comb{C},\comb{S},\comb{W}$. Combinators
$\comb{S}$ and $\comb{W}$ and terms with five occurrences were not among the
characterizations of the considered MGTs.} As a further source of ``names'' we
took \name{Thesis~1--68} from a textbook by \Lukasiewicz \cite{luk:book}.
These were often used for experiments in ATP
\cite{wos:bledsoe:91,mccune:wos:cd:1992,wos:resonance:95,wos:combining:96,cw:cdtools:2022}.
If the MGT of a subproof is a named formula in this sense, this is indicated
as property~\xatt{NN}, whose value points to the respective row of
Table~\ref{tab-formula-names}.

\begin{table}
  \centering
  \caption{``Named'' formulas that occur as MGTs of subproofs of $\DMER$ and $\DLUK$.}
  \label{tab-formula-names}
  \begin{tabular}{rll}
    Id & Formula & Names\\\midrule
    $\spec{1}$ & $\g{CCCpqrCCrpCsp}$ & \name{\Luk}\\
    $\spec{2}$ & $\g{CpCCCqrsCrs}$ & $\comb{K (C B K)}$\\
    $\spec{3}$ & $\g{CCCpqrCqr}$ & \name{Syll-Simp}, $\comb{CBK}$, \thes{19}\\
    $\spec{4}$ & $\g{CpCCpqCrq}$ & $\comb{C(BK)}$\\
    $\spec{5}$ & $\g{CCpqCpq}$ & $\comb{BI}$\\
    $\spec{6}$ & $\g{CpCqq}$ & \name{Simp*}, \name{Irrel},
    $\comb{K^*}$, $\comb{KI}$, $\comb{CK}$, \thes{63}\\
    $\spec{7}$ & $\g{Cpp}$ & \name{Id}, $\comb{I}$, \thes{16}\\
    $\spec{8}$ & $\g{CCCpqrCCrpp}$ & \name{Roll}, \thes{26}\\
    $\spec{9}$ & $\g{CpCCpqq}$ & \name{Pon}, \name{Aff}, \name{Assertion}, $\comb{I'}$, $\comb{CI}$, \thes{20}\\
    $\spec{10}$ & $\g{CCpqCCCprqq}$ & \name{Comm-Tarski}, \name{Henkin}\\
    $\spec{11}$ & $\g{CCpqCCqrCpr}$ & \name{Syll}, \name{Syl}, \name{Suffixing}, $\comb{B'}$, $\comb{CB}$, \thes{1}\\
    $\spec{12}$ & $\g{CCCpqpp}$ & \name{Peirce}, \thes{24}\\
    $\spec{13}$ & $\g{CpCqp}$ & \name{Simp}, $\comb{K}$, \thes{18} \\\midrule
  \end{tabular}
\end{table}

\subsection{Structural Properties of the \DTerm}
\label{subsec-prop-struct}

Structural Properties of the \DTerm refer to the respective subproof as \dterm
or full binary tree.

\attitem{\xatt{DC}, \xatt{DT}, \xatt{DH}: Compacted Size, Tree Size and Height}

The properties \xatt{DC}, \xatt{DT}, \xatt{DH} describe the basic dimensions
of the subproof's structure: compacted size, tree size and height. For the
proofs of $\Syll$ (subproof~32 of $\DMER$ and subproof~33 of $\DLUK$), the
compacted size of $\DMER$ improves with~31 upon that of $\DLUK$, which is~32,
by one. Both proofs have the same height, 29. However, with respect to tree
size $\DMER$ with~491 is larger than $\DLUK$ with only 435.

N-simplification has no reducing effect on $\DMER$, but was applied to obtain
$\DLUK$ from the straightforward conversion of \Lukasiewicz's original proof
to \CD. There, n-simplification effected on the subproof of
$\Syll$ a reduction of the tree size from 563 to 435, while compacted size and
height remained unchanged, 32 and~29, respectively.

Dimensions for the whole proof of the three goal theorems $\Syll$, $\Peirce$
and $\Simp$ as a set of the three \Dterms can be determined as follows. The
compacted size is the total number of compound subterms, that is, the number
of rows in the respective tables, not counting the first row, which represents
the primitive \Dterm $1$ that corresponds to the axiom. Thus the compacted
size of $\DMER$ is~33, improving by one on that of $\DLUK$, which is~34.

As tree size of the overall proof, the set of the three \Dterms for the goal
theorems, we take the sum over the tree sizes of its members, which is
$491+159+19 = 669$ for $\DMER$ and $435+131+19 = 585$ for $\DLUK$. As height
of the overall proof we take the maximum of the height of its members, which
is $\f{max}(\{29,25,10\}) = 29$ for $\DMER$ and $\f{max}(\{29,27,10\}) = 29$
for $\DLUK$. For the \CD conversion of \Lukasiewicz's proof before
n-simplification, the overall dimensions are: compacted size~34, tree
size~751, and height~29.

\begin{sidewaystable}[p]
  \centering
  \small
  \caption{Properties of all subproofs of $\DLUK$ (Fig.~\ref{fig-proof-luk}),
    \Lukasiewicz's proof \cite{luk:1948}, converted to \CD and
    n-simplified.}
  \label{tab-bigluk}
  \setlength{\tabcolsep}{3.1pt}
  \renewcommand{\arraystretch}{0.89}
  \rowcolors{2}{gray!20}{white}
  \scalebox{0.9}{\begin{tabular}{rlrrrrrrrrccccrrrcrrccrrrrr}
 & \tabatt{} & \tabatt{{\L}UK} & \tabatt{MER} & \tabatt{NN} & \tabatt{DC} & \tabatt{DT} & \tabatt{DH} & \tabatt{DI} & \tabatt{DR} & \tabatt{DS} & \tabatt{DP} & \tabatt{DK$_{L}$} & \tabatt{DK$_{R}$} & \tabatt{FC} & \tabatt{FT} & \tabatt{FH} & \tabatt{FO} & \tabatt{MC} & \tabatt{MT} & \tabatt{RS} & \tabatt{RC} & \tabatt{IT$_{U}$} & \tabatt{IT$_{M}$} & \tabatt{IH$_{U}$} & \tabatt{IH$_{M}$} & \tabatt{}\\\midrule
1. & $1$ & \luk{1} & \mer{1} & \spec{1} & 0 & 0 & 0 & 17 & 481 & -- & \textbullet & 0 & 0 & 6 & 6 & 3 & \textbullet & 0 & 0 & \textbullet & \textbullet & 4451 & 183 & 18 & 11 & 1.\\
2. & $\f{D}11$ & \subsumesluk{2} & \textcolor{black!30}{\textbullet} &  & 1 & 1 & 1 & 1 & 43 & $1{=}1$ & \textbullet & 1 & 1 & 7 & 8 & 4 & \textbullet & 1 & 1 & \textbullet & \textbullet & 1640 & 141 & 17 & 11 & 2.\\
3. & $\f{D}12$ & \subsumesluk{3} & \textcolor{black!30}{\textbullet} &  & 2 & 2 & 2 & 1 & 43 & $1\lhd$ & \textbullet & 1 & 2 & 8 & 11 & 4 & \textbullet & 2 & 2 & \textbullet & \textbullet & 1881 & 185 & 17 & 11 & 3.\\
4. & $\f{D}31$ &  & \textcolor{black!30}{\textbullet} &  & 3 & 3 & 3 & 1 & 43 & $\rhd1$ & \textbullet & 2 & 2 & 5 & 5 & 4 & \textcolor{black!30}{\textbullet} & 3 & 3 & \textbullet & \textbullet & 689 & 67 & 16 & 10 & 4.\\
5. & $\f{D}4\mathrm{n}$ & \luk{4} & \mer{2} &  & 4 & 4 & 4 & 2 & 43 & $\rhd\n$ & \textbullet & 3 & 2 & 4 & 4 & 3 & \textbullet & 4 & 4 & \textbullet & \textbullet & 688 & 66 & 15 & 9 & 5.\\
6. & $\f{D}15$ & \luk{5} & \textcolor{black!30}{\textbullet} &  & 5 & 5 & 5 & 2 & 39 & $1\lhd$ & \textbullet & 3 & 2 & 5 & 6 & 3 & \textbullet & 5 & 5 & \textbullet & \textbullet & 1667 & 173 & 15 & 10 & 6.\\
7. & $\f{D}16$ & \luk{6} & \textcolor{black!30}{\textbullet} &  & 6 & 6 & 6 & 1 & 35 & $1\lhd$ & \textbullet & 3 & 3 & 6 & 7 & 4 & \textbullet & 6 & 6 & \textbullet & \textbullet & 1802 & 208 & 16 & 11 & 7.\\
8. & $\f{D}17$ & \luk{7} & \textcolor{black!30}{\textbullet} &  & 7 & 7 & 7 & 1 & 35 & $1\lhd$ & \textbullet & 3 & 4 & 7 & 9 & 4 & \textbullet & 7 & 7 & \textbullet & \textbullet & 2648 & 303 & 16 & 11 & 8.\\
9. & $\f{D}81$ &  & \textcolor{black!30}{\textbullet} & \spec{2} & 8 & 8 & 8 & 1 & 35 & $\rhd1$ & \textbullet & 3 & 4 & 5 & 5 & 4 & \textcolor{black!30}{\textbullet} & 8 & 8 & \textbullet & \textbullet & 1032 & 119 & 15 & 10 & 9.\\
10. & $\f{D}9\mathrm{n}$ & \luk{8} & \mer{3} & \spec{3} & 9 & 9 & 9 & 4 & 35 & $\rhd\n$ & \textbullet & 3 & 4 & 4 & 4 & 3 & \textbullet & 9 & 9 & \textbullet & \textbullet & 1031 & 118 & 14 & 9 & 10.\\
11. & $\f{D}10.1$ & \luk{9} & \mer{4} & \spec{4} & 10 & 10 & 10 & 2 & 30 & $\rhd1$ & \textbullet & 4 & 4 & 4 & 4 & 3 & \textbullet & 10 & 10 & \textbullet & \textbullet & 448 & 64 & 13 & 9 & 11.\\
12. & $\f{D}1.11$ & \luk{10} & \textcolor{black!30}{\textbullet} &  & 11 & 11 & 11 & 1 & 20 & $1\lhd$ & \textbullet & 4 & 4 & 7 & 7 & 5 & \textbullet & 11 & 11 & \textbullet & \textbullet & 498 & 72 & 14 & 10 & 12.\\
13. & $\f{D}1.12$ & \luk{11} & \textcolor{black!30}{\textbullet} &  & 12 & 12 & 12 & 1 & 20 & $1\lhd$ & \textbullet & 4 & 4 & 8 & 12 & 5 & \textbullet & 12 & 12 & \textbullet & \textbullet & 1157 & 167 & 14 & 10 & 13.\\
14. & $\f{D}1.13$ & \luk{12} & \textcolor{black!30}{\textbullet} &  & 13 & 13 & 13 & 1 & 20 & $1\lhd$ & \textbullet & 4 & 4 & 9 & 10 & 6 & \textbullet & $\iv{12}{13}$ & 13 & \textbullet & \textbullet & 1050 & 153 & 15 & 11 & 14.\\
15. & $\f{D}1.14$ & \luk{13} & \textcolor{black!30}{\textbullet} &  & 14 & 14 & 14 & 1 & 20 & $1\lhd$ & \textbullet & 4 & 5 & 10 & 15 & 6 & \textbullet & $\iv{12}{14}$ & 14 & \textbullet & \textbullet & 1657 & 241 & 15 & 11 & 15.\\
16. & $\f{D}15.1$ &  & \textcolor{black!30}{\textbullet} &  & 15 & 15 & 15 & 1 & 20 & $\rhd1$ & \textbullet & 4 & 5 & 8 & 9 & 5 & \textcolor{black!30}{\textbullet} & $\iv{12}{15}$ & 15 & \textbullet & \textbullet & 684 & 99 & 14 & 10 & 16.\\
17. & $\f{D}16.\mathrm{n}$ & \luk{14} & \mer{5} &  & 16 & 16 & 16 & 2 & 20 & $\rhd\n$ & \textbullet & 4 & 5 & 7 & 8 & 4 & \textbullet & $\iv{12}{16}$ & 16 & \textbullet & \textbullet & 683 & 98 & 13 & 9 & 17.\\
18. & $\f{D}17.1$ & \luk{15} & \mer{6} &  & 17 & 17 & 17 & 3 & 15 & $\rhd1$ & \textbullet & 4 & 5 & 6 & 7 & 3 & \textbullet & $\iv{12}{17}$ & 17 & \textbullet & \textbullet & 395 & 56 & 12 & 8 & 18.\\
19. & $\f{D}18.11$ & \luk{16} & \mer{7} &  & 18 & 28 & 18 & 1 & 10 & $\rhd$ & -- & 5 & 5 & 6 & 7 & 4 & \textbullet & $\iv{12}{14}$ & 14 & \textbullet & \textbullet & 209 & 61 & 11 & 9 & 19.\\
20. & $\f{D}19.1$ & \luk{17} & \mer{8} &  & 19 & 29 & 19 & 1 & 10 & $\rhd1$ & -- & 6 & 5 & 8 & 9 & 5 & \textbullet & $\iv{12}{15}$ & 15 & \textbullet & \textbullet & 132 & 38 & 10 & 8 & 20.\\
21. & $\f{D}1.20$ & \luk{18} & \mer{10} &  & 20 & 30 & 20 & 2 & 10 & $1\lhd$ & -- & 6 & 5 & 9 & 12 & 5 & \textbullet & $\iv{12}{16}$ & 16 & \textbullet & \textbullet & 158 & 47 & 10 & 8 & 21.\\
22. & $\f{D}21.21$ &  & \textcolor{black!30}{\textbullet} &  & 21 & 61 & 21 & 1 & 5 & $=$ & -- & 6 & 5 & 9 & 10 & 5 & \textcolor{black!30}{\textbullet} & $\iv{12}{17}$ & $\iv{24}{33}$ & \textbullet & \textbullet & 53 & 16 & 9 & 7 & 22.\\
23. & $\f{D}22.\mathrm{n}$ & \subsumesluk{19} & \mer{11} &  & 22 & 62 & 22 & 1 & 5 & $\rhd\n$ & -- & 6 & 5 & 8 & 9 & 4 & \textbullet & $\iv{12}{18}$ & $\iv{24}{34}$ & \textbullet & \textbullet & 52 & 15 & 8 & 6 & 23.\\
24. & $\f{D}17.23$ & \subsumesluk{20} & \mer{12} &  & 23 & 79 & 23 & 2 & 5 & $\lhd$ & -- & 6 & 5 & 8 & 9 & 4 & \textbullet & $\iv{12}{23}$ & $\iv{23}{51}$ & \textbullet & \textbullet & 57 & 16 & 7 & 5 & 24.\\
25. & $\f{D}24.18$ & \subsumesluk{21} & \mer{13} &  & 24 & 97 & 24 & 2 & 2 & $\rhd$ & -- & 6 & 5 & 6 & 7 & 4 & \textbullet & $\iv{12}{24}$ & $\iv{24}{69}$ & \textbullet & \textbullet & 27 & 17 & 6 & 5 & 25.\\
26. & $\f{D}65$ &  &  & \spec{6} & 6 & 10 & 6 & 1 & 4 & $\rhd$ & -- & 3 & 2 & 2 & 2 & 2 & \textcolor{black!30}{\textbullet} & 5 & 9 & \textbullet & -- & 28 & 8 & 7 & 5 & 26.\\
27. & $\f{D}26.\mathrm{n}$ & \luk{22} & \subsumesmer{9} & \spec{7} & 7 & 11 & 7 & 2 & 4 & $\rhd\n$ & -- & 3 & 2 & 1 & 1 & 1 & \textbullet & 6 & 10 & \textbullet & -- & 27 & 7 & 6 & 4 & 27.\\
28. & $\f{D}24.27$ & \luk{23} & \mer{14} & \spec{8} & 26 & 91 & 24 & 1 & 3 & -- & -- & 6 & 5 & 5 & 5 & 3 & \textbullet & $\iv{13}{25}$ & $\iv{24}{91}$ & \textbullet & -- & 24 & 7 & 6 & 4 & 28.\\
29. & $\f{D}10.28$ & \luk{24} & \mer{15} & \spec{9} & 27 & 101 & 25 & 1 & 3 & $\lhd$ & -- & 6 & 5 & 3 & 3 & 3 & \textbullet & $\iv{13}{26}$ & $\iv{24}{101}$ & \textbullet & -- & 19 & 5 & 6 & 4 & 29.\\
30. & $\f{D}18.29$ & \luk{25} & \mer{16} & \spec{10} & 28 & 119 & 26 & 3 & 3 & $\lhd$ & -- & 6 & 5 & 5 & 5 & 4 & \textbullet & $\iv{12}{26}$ & $\iv{24}{36}$ & \textbullet & -- & 19 & 5 & 6 & 4 & 30.\\
31. & $\f{D}30.30$ & \subsumesluk{26} & \mer{16\pr} &  & 29 & 239 & 27 & 1 & 1 & $=$ & -- & 6 & 5 & 7 & 10 & 5 & \textbullet & $\iv{12}{27}$ & $\iv{24}{239}$ & \textbullet & -- & 13 & 13 & 5 & 5 & 31.\\
32. & $\f{D}25.31$ &  & \textcolor{black!30}{\textbullet} &  & 31 & 337 & 28 & 1 & 1 & $<_{\mathrm{c}}$ & -- & 6 & 6 & 7 & 7 & 5 & \textbullet & $\iv{12}{29}$ & $\iv{23}{121}$ & \textbullet & -- & 13 & 13 & 5 & 5 & 32.\\
33. & $\f{D}32.25$ & \luk{29} & \mer{17} & \spec{11} & 32 & 435 & 29 & 0 & 1 & $\rhd$ & -- & 7 & 6 & 5 & 5 & 3 & \textbullet & $\iv{12}{22}$ & $\iv{24}{64}$ & \textbullet & -- & 5 & 5 & 3 & 3 & 33.\\
34. & $\f{D}30.27$ & \luk{28} & \mer{18} & \spec{12} & 29 & 131 & 27 & 0 & 1 & $\rhd$ & -- & 6 & 5 & 3 & 3 & 3 & \textbullet & 11 & 15 & -- & -- & 3 & 3 & 3 & 3 & 34.\\
35. & $\f{D}10.10$ & \luk{27} & \mer{19} & \spec{13} & 10 & 19 & 10 & 0 & 1 & $=$ & -- & 4 & 4 & 2 & 2 & 2 & \textbullet & 6 & 7 & \textbullet & \textbullet & 2 & 2 & 2 & 2 & 35.\\
\end{tabular}
}
\end{sidewaystable}

\attitem{\xatt{DI}: Number of Incoming DAG Edges}

With \xatt{DI} we refer to the number of incoming edges in the DAG
representation of the overall proof of all theorems.
The roots of the DAG, corresponding to the goal theorems, can be identified by
the \xatt{DI} value~0. In both tables it can be observed that there are rows
such as the row for subproof~5 of $\DMER$ in Table~\ref{tab-bigmer} where the
\xatt{DI} value is 1 and there is nevertheless an entry in the column with the
labels of the original proof presentation, $\xatt{MER}$ for
Table~\ref{tab-bigmer} and $\xatt{\LUK}$ for Table~\ref{tab-bigluk}. These
rows exemplify the use of labels by Meredith and \Lukasiewicz to modularize
proofs as addressed in Sect.~\ref{subsubsec-modularization}.

\attitem{\xatt{DR}: Repeats}
\label{subsubsec-dr}

\xatt{DR} denotes the total number of occurrences in the set of expanded trees
of all roots of the DAG. Because leaves labeled by n-simplification with $\n$
are not considered here, the number of occurrences of the primitive subproof 1
shown in the tables is smaller than the total number of leaves of the three
trees, which is the overall tree size plus one, that is, 670 for $\DMER$ and
586 for $\DLUK$.

\attitem{\xatt{DS}: Structural Relationship between the Subproofs of Major
  and Minor Premise}

\xatt{DS} describes special cases of the structural relationship between the
subproofs of major and minor premise. Possible values are identity, expressed
with $=$, the strict subterm and superterm relationships expressed with $\lhd$
and $\rhd$, respectively, and the strict compaction ordering relationship (if
none of the other relationships holds) expressed with
$\mathrel{<_{\mathrm{c}}}$ and~$\mathrel{>_{\mathrm{c}}}$.\footnote{Cases
where the compaction ordering applied only non-strictly did not occur in the
investigated proofs.} In addition, it is indicated if a premise is the axiom
or is $\n$.

Consideration of this property was motivated by the empiric observation that
for most subproofs of $\DMER$ and $\DLUK$ the subproofs of both premises are
related by the subterm relationship. In fact, in each of the proofs $\DMER$
and $\DLUK$ the value of $\xatt{DS}$ is for all compound subproofs with
exception of two ones either $=$, $\lhd$ or $\rhd$. This observed pattern can
actually be reversed into a proof construction method that succeeds for
many \CD problems, also with multiple axioms, and leads in some cases to
proofs with small compacted size where the exhaustive search for proofs with
guaranteed smallest compacted size appears unfeasible \cite{cw:cdtools:2022}. In
Sect.~\ref{subsec-psp} we will specify this method and show a
particularly short proof of \Syll from \Luk obtained with it.

\attitem{\xatt{DP}: Is Prime}

\xatt{DP} expresses that \xatt{DT} and \xatt{DC} are the same. We call \Dterms
with this property \emph{prime}, because they do not have repeated subterms
that can be ``factored'' in a DAG representation. Assuming a singleton set
$\DPRIMSET = \{1\}$ of primitive \Dterms, the property can be characterized in
different ways: (1)~\xatt{DT} and \xatt{DC} are the same. (2)~\xatt{DT} and
\xatt{DH} are the same. (3)~Every compound subterm of the given \Dterm has
only a single occurrence in it. (4)~The given \Dterm is a member of
$\bigcup_{i = 0}^\omega \primelevel(i)$. where for natural numbers $n \geq 0$
the set $\primelevel(n)$ of \Dterms is specified inductively as

\beforelistskip
\begin{enumerate}
\item $\primelevel(0)\;\eqdef\; \{1\}$.
\item $\primelevel(1)\;\eqdef\; \{\D(1,1)\}$.
\item $\displaystyle \primelevel(n+2)\; \eqdef\!\!\! \bigcup_{d \in
  \primelevel(n+1)}\!\!\!\! \{\D(1,d)\} \cup \{\D(d,1)\}$.
\end{enumerate}  

Characterization~(3) suggests to perform proof search by enumerating \Dterms
for increasing values of $\primelevel$. Members of $\primelevel(n)$ have size
(compacted size, tree size or height, which are identical for them)~$n$. The
number of distinct prime \dterms of a given size~$n$ grows by the sequence
\OEISNUM{A011782} of integers \cite{oeis}
(Table~\ref{tab-oeis-numbers-prime}), i.e.,~$1$ for $n=0$ and $2^{n-1}$ for $n
> 0$, which is much slower than the growth for compacted size, tree size or
height shown in Table~\ref{tab-oeis-numbers} on p.~\pageref{tab-oeis-numbers}.

\begin{table}[t]
  \centering
  \caption{The number of distinct \Dterms for a single axiom (or full binary
    trees) in $\primelevel(n)$.}
  \label{tab-oeis-numbers-prime}
  \begin{tabular}{llcrrrrrrrrrr}
    $n$ && 0 & 1 & 2 & 3 & 4 & 5 & 6 & 7 & 8 & 9\\\midrule    
    $|\primelevel(n)|$ & \OEISNUM{A011782}
    & 1 & 1 & 2 & 4  & 8   & 16    &  32 & 64 & 128 & 256
  \end{tabular}
\end{table}

For $\DMER$ we observe in Table~\ref{tab-bigmer} that the subproofs~1--18 are
exactly those that are prime. Moreover, all these prime proofs in $\DMER$ are
a subproof of a single subproof, subproof~18. This suggests that proof search
may be decomposed into two phases. First, identifying a small number of
``maximal prime proofs'' or ``prime cores'' \cite{cwwb:lukas:2021}, such as
subproof~18 in $\DMER$ for axiom \Luk. This is in a search space that --
narrowed through the prime property and possibly further properties --
relatively quickly leads beyond small proof sizes for which all structures can
be trivially explored. Second, further search with the MGTs of the prime cores
available as proven lemma formulas. Such experiments were performed for
deriving \Syll from \Luk with \ProverN \cite{prover9} as prover for the second
phase, leading to proofs with much smaller compacted size (44) than obtained
by \ProverN alone (80--94, see Sect.~\ref{subsec-prover9})
\cite{cwwb:lukas:2021,cwwb:lukas:2021:extended}. Yet above the size of the
human-made proofs (31--32) and a machine proof obtained with another technique
(22) described in Sect.~\ref{subsec-psp}.

\pagebreak
\attitem{\xatt{DK$_L$}, \xatt{DK$_R$}: Left and Right Successive Height}

\xatt{DK$_L$}, \xatt{DK$_R$} are the maximal number of successive edges going
to the left and right, respectively, on any path from the root to a leaf.
These properties were motivated by the observation that in $\DMER$ and $\DLUK$
these values are relatively low compared to the height of the subproof. This
suggests that limiting them could restrict the number of candidate structures
during proof search. Both proofs would, for example, satisfy the constraint
$\xatt{DK$_L$}^2 \leq 2.5*\xatt{DH}$ \textit{and} $\xatt{DK$_R$}^2 \leq
2.5*\xatt{DH}$. Whether such restrictions can indeed be successfully used in
proof search has not yet been settled. Empirical observations obtained in our
experiments suggest that with structure enumeration for increasing
tree size they lead to a linear reduction of the number of considered trees.
Namely, while the numbers of full binary trees of tree sizes 13 and~14 are
742,900 and 2,674,440, respectively (\OEISNUM{A000108}), with the above
constraints these numbers are roughly halved to 385,234 and 1,405,546,
respectively. With enumeration for increasing height the reduction seems
stronger: The number of binary trees of heights~4 and~5 are 651 and 457,653,
respectively (\OEISNUM{A001699}). With the above constraints, the numbers are
reduced to~231 and 9,153, respectively.

\subsection{Properties of the \MGT}
\label{subsec-prop-mgt}

Here we discuss properties of the argument term $f$ of the \MGT~$\P(f)$ of the
respective subproof.

\attitem{\xatt{FC}, \xatt{FT}, \xatt{FH}: Compacted Size, Tree Size and
  Height}
\label{subsubsect-fcth}

The properties \xatt{FC}, \xatt{FT}, \xatt{FH} describe the basic dimensions
of $f$. They are defined now for terms in full analogy to the respective
measures for \Dterms (Definitions~\ref{def-csize}
and~\ref{def-treesize-height}): The compacted size \xatt{FC} is the number of
inner nodes of the minimal DAG representing the tree; the tree size \xatt{FT}
is the number of inner nodes, in other words, the number of occurrences of
function symbols of arity larger than $0$; the height \xatt{FH} is the length
(number of edges) of the longest downward path from the root to a leaf. In the
literature, the term height is also called \name{term depth}.

The maximal \xatt{FT} value in $\DMER$ as well as in $\DLUK$ is~15. It
pertains in both proofs to the same formula, which, moreover, happens to
appear in both proofs as MGT of the respective subproof number~15. In
Meredith's presentation it is just an implicit intermediate formula, indicated
by the empty value of $\xatt{MER}$ in Table~\ref{tab-bigmer}, whereas in
\Lukasiewicz's presentation it is made explicit as thesis number~\luk{13}.
With respect to the tree size, this formula stands out: the next largest value
of $\xatt{FT}$ is 12, which pertains in both proofs to two subproofs. The
maximal value of \xatt{FH} in both proofs is 6 and pertains in each of the
proofs to two subproofs, including that with the $\xatt{FT}$ value~15.

Deleting inferred formulas whose tree size or height exceeds a threshold are
basic techniques to restrict the search space of resolution provers.
Corresponding \ProverN options are for example \texttt{max\_weight} and
\texttt{max\_depth} \cite{prover9}. The default measure used as term weight by
\ProverN is linearly related to the tree size as defined here. \CD problems
are processed by \ProverN in default settings with positive hyperresolution.
The inferred resolvents are then actually MGTs of \Dterms that can be
associated with the hyperresolution derivations. In contrast, clausal tableau
provers with rigid variables do not explicitly construct these MGTs; they only
construct the deeper instantiated IPTs associated with particular nodes of the
tableau tree. Hence, restricting the search space by limiting term dimensions
of MGTs is usually not available for clausal tableau provers.

Blending goal-driven structure enumeration with axiom-driven structure
enumeration that permits the application of heuristic limitations to MGTs was
recently studied for \CD problems; it led to a drastic improvement compared to
conventional clausal tableau provers \cite{cw:sgcd}.

\attitem{\xatt{FV}: Number of Distinct Variables}

Like \xatt{FT} and \xatt{FH}, the property \xatt{FV}, that is, the number of
distinct variables, is commonly used in resolution provers as a threshold to
delete inferred formulas that exceed it. In \ProverN this threshold can be
specified with the \texttt{max\_vars} option. The discussion in
Sect.~\ref{subsubsect-fcth} on the availability of MGTs for heuristic
restrictions applies here as well.

\attitem{\xatt{FO}: Is [Weakly] Organic}

The \name{organic} property \xatt{FO} of a propositional formula, with respect
to a set of axioms, says that it has no strict subformula that is itself a
theorem entailed by the axioms. With our wrapper predicate $\P$ this means
that an MGT $\P(f)$ is organic if $f$ has no strict subterm~$f'$ such that
$\forall \P(f')$ is entailed by the given axioms. \Lukasiewicz and his
collaborators aimed at finding axiomatizations of propositional logics with
axioms that are organic \cite{luk:tarski:aussagenkalkuel:1930,luk:1948}. For
axiomatizations of fragments of propositional logic, the \name{organic}
property can be checked by a SAT solver. In the proofs $\DMER$ and $\DLUK$ we
observe that with a few exceptions the MGTs of all subproofs actually are
organic. The exceptions can, however, be ascribed a weakened form of
\name{organic} that is specified as follows: We call an atomic formula~$\P(f)$
\emph{weakly organic} if it is not organic and $f$ is an implication $\i(p,g)$
(or $\g{Cpg}$ in \Lukasiewicz's notation) where $p$ is a variable that does
not occur in $g$ and $\P(g)$ is organic. The \name{weakly organic} property is
indicated in the property tables by a gray bullet.

\subsection{Comparisons with all Proofs of the \MGT}

The properties considered in this subsection apply to all proofs of the MGT of
the respective subproof, regarded as a set of D-terms.

\attitem{\xatt{MC}, \xatt{MT}: Minimal Compacted and Tree Size of a Proof}

The values of \xatt{MC} and \xatt{MT} are the minimal compacted size of a
proof of the MGT and the minimal tree size of a proof of the MGT,
respectively. These values may be hard to determine such that they often can
only be narrowed down to an integer interval. Values of these properties were
found with the provers \name{CCS} \cite{cw:ccs} and \SGCD \cite{cw:sgcd} in
configurations that exhaustively search for proofs with a given compacted size
or tree size, respectively.

In particular for the goals \Peirce and \Simp (subproofs~33 and~34 in
Table~\ref{tab-bigmer}, subproofs~34 and~35 in Table~\ref{tab-bigluk}) it can
be observed that the compacted size \xatt{DC} and tree size \xatt{DT} are much
larger than the respective minimal values \xatt{MC} and \xatt{MT}. This is
understandable because the apparent aim of Meredith and \Lukasiewicz was to
reduce the overall compacted size. \Peirce and \Simp are thus proven in
$\DMER$ and $\DLUK$ not as standalone problems but as side results from the
given proof of $\Syll$. Subproofs of that proof are permitted to be re-used
there without increase of the overall compacted size.

\subsection{Regularity}

The regularity properties hold for the respective subproof as \dterm.

\attitem{\xatt{RS}, \xatt{RC}: Is \XS-Regular, Is \XC-Regular}

These properties are regularities as specified in
Definitions~\ref{def-red-rsim} and~\ref{def-red-raci}. In $\DMER$ and $\DLUK$
all subproofs are \XS-regular, with the exception of a single subproof that
derives \Peirce as a side result. In $\DMER$ there is just a single subproof
that is not \XC-regular, while in $\DLUK$ \XC-regularity fails for nine
subproofs, indicating a greater redundancy.

\subsection{Properties of Occurrences of the \IPTs}

The respective subproof has \xatt{DR} (see Sect.~\ref{subsubsec-dr})
occurrences in the overall proof as a set of trees. The following properties
refer to the multiset of the arguments $f$ of the IPTs~$\P(f)$ of all these
occurrences.

\attitem{\xatt{IT$_U$}, \xatt{IT$_M$}: Tree Size of the IPTs -- Maximum and
  Rounded Median}

\xatt{IT$_U$} and \xatt{IT$_M$} indicate the tree size of the members of the
considered multiset by the values of the maximum and the rounded median. These
values may be compared with \xatt{FT}, the tree size of (the argument term of)
the MGT. In particular for subproofs that appear at deeper levels in the
overall proof, \xatt{IT$_U$} and \xatt{IT$_M$} are much larger than \xatt{FT},
illustrating Proposition~\ref{prop-ipt-subsumedby-mgt}. The largest tree size
of (the argument term of) an IPT in $\DMER$ as well as $\DLUK$ is 4451. It is
the value of an instance of the axiom, where the tree size of (the argument
term of) the MGT, that is, the axiom formula itself, is just 6.

\attitem{\xatt{IH$_U$}, \xatt{IH$_M$}: Height of the IPTs -- Maximum and Rounded
  Median}

\xatt{IH$_U$} and \xatt{IH$_M$} indicate the height of the members of the
considered multiset by the values of the maximum and the rounded median.
Compared with \xatt{FH}, the height of (the argument term of) the MGT, they
are similarly as in the comparison of \xatt{IT$_U$} and \xatt{IT$_M$} with
\xatt{FT} much higher for subproofs appearing at deeper levels, however on a
quite different scale: The largest height of (the argument term of) an IPT in
$\DMER$ as well as $\DLUK$ is 18, for an instance of the axiom, where the
height of (the argument term of) the MGT is~3.

\section{Proofs of \Syll from \Luk by ATP Systems}
\label{sec-atp}

Deriving \Syll from \Luk and \Det, that is, showing the validity of~\PSYLL
(Sect.~\ref{sec-background}), or solving \TPTP problem \TP{LCL038-1}, which
was achieved without a computer by \Lukasiewicz \cite{luk:1948}, was brought
up as a challenge problem for ATP by Frank Pfenning in 1988
\cite{pfenning:single:1988}. In this section we summarize the achievements of
ATP systems on the problem since then and report the dimensions of proofs
found by \ProverN \cite{prover9}, which in essence are \CD proofs. For the
proofs by \ProverN we show the effects of the novel reductions introduced in
Sect.~\ref{sec-red}. Finally we present a new proof, which is much shorter
than all known ones. It has been obtained with a novel technique inspired by
observations made at the investigation of the human-made proofs.

\subsection{From a Challenge Problem to a Not-That-Easy Zero-Rated Problem}
\label{subsec-challenge-zero}

According to Larry Wos et al. \cite{wos:contributes:1990} \Syll, \Peirce
and \Simp could be derived in 1990 by \Otter \cite{otter} in about 11 hours.
Techniques were weighting formulas by symbol count and hyperresolution as
inference rule. In 1992 \Otter needed about 8 hours, generating 6.7 million
clauses and keeping about 20 thousands clauses to derive \Syll, while the
parallel prover
\name{Roo} achieved a nearly linear speedup for the problem, solving it with
24 processes in about 21 minutes \cite{roo:parallel:1992}. The inference rule
was hyperresolution, and forward subsumption (but not back subsumption) was
applied. In addition, to conserve memory, generated clauses with more than 20
symbols were discarded. Also in 1992 strategies for \CD with Otter were
compared \cite{mccune:wos:cd:1992}. Depending on the strategy, \Otter could
derive \Syll in about 2--4 hours. As mentioned there, proving \Syll from \Luk
was the first truly difficult \CD theorem proved by \Otter and has been used
extensively as a benchmark for parallel deduction programs. \name{CODE}
\cite{fuchs:code:1997}, a dedicated solver for \CD from 1997 apparently could
also solve the problem.

Branden Fitelson and Wos \cite{fitelson:missing:2001} studied various classes
of ``missing'' proofs. \Lukasiewicz's proof is there the leading example of a
proof with
\emph{omissions}, where subproofs of some steps are missing. \Lukasiewicz's
presentation shows 28 steps. The objective of Fitelson and Wos was to produce
from these displayed steps a proof that contains all of these, but is entirely
formed by the more fine-grained \CD steps. Otter succeeded, finding a proof of
length (i.e., compacted size) 36. Actually, our proof~$\DLUK$
(Fig.~\ref{fig-proof-luk}) is another such completion, but was obtained
without proof search just from a detailed transcription of
\Lukasiewicz's presentation, as described in
Sect.~\ref{subsec-considered-proofs}. Its compacted size is~34.

The problem of deriving \Syll from \Luk and \Det entered the \TPTP as
\TP{LCL038-1}. Its first documented difficulty rating in \TPTP version 2.0.0,
1997, is 1.00, meaning that the problem is hard because no state-of-the-art
ATP system in a specific sense \cite{tptp-rating} can solve it. A value of
0.00, meaning that the problem is easy, or all ``state-of-the-art ATP
systems'' can solve the problem, first appeared with version 3.2.0 in 2006.
Since then the difficulty rating fluctuated between 0.00 and 0.81. Its current
value in version 9.0.0 is 0.60.

According to the
\name{ProblemAndSolutionStatistics} file of \TPTP~9.0.0 from 2024 the two
well-known powerful provers \E \cite{eprover} and \Vampire \cite{vampire} fail
on it in their recent versions 3.2.0 and 4.9, respectively. Nevertheless, in
earlier versions they succeed, as documented in the
\name{ProblemAndSolutionStatistics} file of \TPTP~7.5.0 from 2021 and
replicable with versions downloadable from the systems' Web
pages.\footnote{\url{http://www.eprover.org/} and
\url{https://vprover.github.io/}, accessed Jan 15, 2023.}
\E~2.6\footnote{Invoked with flags \texttt{-s --print-statistics
  --proof-object=1.}} finds a proof with 88 steps and
\Vampire~4.5.1\footnote{Invoked with flags \texttt{--time\_limit 600 --mode
  casc.}} a proof with 148 steps (in both cases not counting the three initial
clauses as steps). It is not evident how these proofs would be translated to
\CD proofs and thus how their size actually compares to that of the human
proofs. For a rough estimate, however, we can observe that the compacted size,
which is 32 and 31 for the proof by \Lukasiewicz and Meredith's variation,
respectively, is the exact number of positive hyperresolution steps to build
the proof. If the hyperresolution is modeled by binary resolution, the number
of steps doubles to 64 or 62, respectively.

For the goal-driven first-order provers such as \leancop \cite{leancop},
\SETHEO \cite{setheo:92} or \PTTP \cite{pttp}, which may described as based on
clausal tableaux \cite{letz:habil}, the CM
\cite{bibel:atp:1982,bibel:otten:2020} or model elimination
\cite{loveland:1978}, the problem remains out of reach. This is not
surprising, given that these systems in essence enumerate tree structures
whose size is linearly related to the tree size of \Dterms, 435 and 491 for
\Lukasiewicz's proof and Meredith's variation, respectively, and 64 as
currently known smallest value (Sect.~\ref{subsec-psp}). The only known solutions
of the problem with this approach are with a recent generalization where the
goal-driven structure enumeration is interwoven with heuristically restricted
axiom-driven structure enumeration \cite{cw:sgcd}. We will discuss a proof
obtained in this way below in Sect.~\ref{subsec-psp}.

\subsection{Prover9's Proofs and Reductions by Replacing Subproofs}
\label{subsec-prover9}

\ProverN, like \Otter \cite{otter}, succeeds on \TP{LCL038-1}. Moreover, by
default it applies positive hyperresolution to \CD problems, where proofs
directly translate to \CD proofs, that is, \Dterms. It appears that in
applications with axiomatizations of logics it is often desired to have \CD
proofs in contrast to arbitrary resolution proofs \cite{veroff:cd:2011}.
\CDTools \cite{cw:cdtools:2022}, a \SWIProlog library to support experimenting with
\CD, provides a conversion of \ProverN's hyperresolution proofs to \Dterms.
This is implemented using \name{Prooftrans}, a proof conversion tool, which
comes with \ProverN. The availability of \ProverN's proofs as \Dterms permits
to compare their dimensions with those of the human proofs and to experiment
with the reductions introduced in Sect.~\ref{subsec-regularities}.

\ProverN in default settings returns for \TP{LC038-1} different proofs,
although of roughly similar size, depending on whether in the clause \Det the
major premise appears before the minor premise, as in the original \TPTP
problem file, or \Det is reordered such that the major premise appears after
the minor premise.\footnote{Provers that are more sensitive to the ordering of
literals in a clause typically determine this ordering on the basics of
heuristics, independently from the ordering in the input, e.g.,
\cite[Sect.~5.3]{setheo:92}.} Tables~\ref{tab-prover9}
and~\ref{tab-prover9-major-minor} show properties of the respective proofs:
(1) in its original form as obtained from \ProverN; (2) after n-simplification
(Definition~\ref{def-simp-n}); (3) and (4) after exhaustively applying
\XS-reduction (\ref{def-red-rsim}) and \XC-reduction (\ref{def-red-raci}),
respectively, to (2); and (5) after applying \XC-reduction to (3). The
proofs~(4) and~(5) within each table are identical.

The shown properties are as those specified in Sect.~\ref{subsec-properties} with
the following additions. $\xatt{DX}$ is the \SCsize
(Definition~\ref{def-scsize}) of the \Dterm. \xatt{FT}$_{\g{Max}}$ and
\xatt{FH}$_{\g{Max}}$ are the maximal values of $\xatt{FT}$ and $\xatt{FH}$
among all subproofs of the given proof, i.e., the maximal tree size and
maximal height of the MGT of a subproof. \textit{Red.} indicates the number of
reduction steps performed to obtain the proof as described in the \name{Source
  of the \Dterm} column. Specifically, for n-simplification \textit{Red.}
shows the number of occurrences of $\n$ in the \Dterm and for \XS- and
\XC-reduction it shows the actual number of rewriting steps according to
Definitions~\ref{def-red-rsim} and~\ref{def-red-raci}, respectively.

We also experimented with configuring \ProverN such that it continues to
search for further proofs after a proof was found, but this did not lead to
finding a second proof within several minutes. In another experiment we tried
\ProverN with increasing values of \texttt{max\_depth}, which limits
\xatt{FH}$_{\g{Max}}$. The lowest number where it succeeds is~7, corresponding
in our scale, not counting the predicate, to term height~6. The prover then
succeeds very quickly, in 7~s, compared to 44~s without \texttt{max\_depth}
restriction, but the proofs are larger, with compacted size~110 (tree size
315,246, height~50) if the major premise of \Det appears after the minor
premise, and compacted size~131 (tree size 400,792, height~50) if it appears
before. Also the value of \xatt{FT}$_{\g{Max}}$ with 14 is in both cases
larger.

\begin{table}
  \centering
  \caption{Properties of the proof \TP{LCL038-1} found by \ProverN in default
    settings if in input clause \Det the major premise appears after the minor
    premise along with the effects of reductions on the proof.}
  \label{tab-prover9}
\setlength{\tabcolsep}{3pt}
\begin{tabular}{rlrrrrrrccr}
  & \textit{Source of \Dterm} & \xatt{DC} & \xatt{DT} & \xatt{DH} & \xatt{DX} &
  \xatt{FT$_{\g{Max}}$} &  \xatt{FH$_{\g{Max}}$} &
  \xatt{RS} & \xatt{RC} & \textit{Red.}\\\midrule
  (1) & From \ProverN & 94 & 304,890 & 40 & 3,247 & 11 & 7 &  \tabno & \tabno\\
  (2) & From (1) by n-simp. & 83 & 8,217 & 38 & 2,485 & 11 & 7 & \tabno & \tabno & 1,708\\
  (3) & From (2) by \XS-red. & 80 & 7,058 & 38 & 2,311 & 13 & 7 & \tabyes & \tabno & 61\\
  (4) & From (2) by \XC-red. & 80 & 5,746 & 36 & 2,290 & 13 & 7 & \tabyes & \tabyes &  2\\
  (5) & From (3) by \XC-red. & 80 & 5,746 & 36 & 2,290 & 13 & 7 & \tabyes & \tabyes &  1
\end{tabular}
\end{table}

\begin{table}
  \centering
  \caption{Proof properties and effects of reductions as in
    Table~\ref{tab-prover9}, but for the case where in clause \Det the major
    premise appears before the minor premise, as in the original \TPTP problem
    file.}
    \label{tab-prover9-major-minor}
\setlength{\tabcolsep}{3pt}
\begin{tabular}{rlrrrrrrccr}
  & \textit{Source of \Dterm} & \xatt{DC} & \xatt{DT} & \xatt{DH} & \xatt{DX}
  & \xatt{FT$_{\g{Max}}$} & \xatt{FH$_{\g{Max}}$} & \xatt{RS} & \xatt{RC} &
  \textit{Red.}\\\midrule
  (1) & From \ProverN & 93 & 216,094 & 40 & 3,011 & 11 & 7 & \tabno & \tabno\\
  (2) & From (1) by n-simp. &  91 & 18,261 & 38 & 2,870 & 11 & 7 & \tabno & \tabno & 3,700\\
  (3) & From (2) by \XS-red. & 88 & 12,922 & 38 & 2,669 & 13 & 7 & \tabyes & \tabno & 281\\
  (4) & From (2) by \XC-red. & 84 & 8,200 & 36 & 2,410  & 13 & 7 & \tabyes & \tabyes & 6\\
  (5) & From (3) by \XC-red. & 84 & 8,200 & 36 & 2,410 & 13 & 7 & \tabyes & \tabyes & 5
\end{tabular}
\end{table}

The most striking values in Tables~\ref{tab-prover9}
of~\ref{tab-prover9-major-minor} are the vast tree sizes \xatt{DT} of the
original proofs, which are drastically reduced by n-simplification. It is not
clear whether this apparent redundancy has a negative effect on proof search.

Actually, tree size seems to be not much taken into consideration in the
context of resolution. Being closely related to the \name{multiplicity} of a
clause in a proof, it may be seen as a fundamental measure for clausal
tableaux with rigid variables. While it is considered by Veroff
as \name{CDcount} in the investigation of finding shortest
proofs \cite{veroff:shortest:2001}, it is, in contrast to compacted size and
height, not even mentioned in a CD-related work by
Wos \cite{wos:combining:96}. On the other hand, it appears that compacted size
-- underlying DAGs as proof structures -- is considered in the context of
clausal tableaux only rarely, for example in \cite{eder:cs:1989,cw:ccs}. The
deeper reason for these preferences lies in the fact that any resolvent may be
regarded as a lemma. The use of lemmas leads to DAGs, hence the focus on these
in resolution.

\subsection{PSP Level Enumeration and a Short Proof}
\label{subsec-psp}

Column \xatt{DS} in Tables~\ref{tab-bigmer} and~\ref{tab-bigluk} shows that
steps in the human-made proofs can often be described in a proof-structural
way as a \Dterm $\D(d,d')$ where either~$d$ is the proof of some previously
proven lemma and~$d'$ is a subterm of~$d$, or vice versa. The question is then
whether this observed pattern can be turned into a proof construction method
that is useful for proof search. As a basis for such a method we define an
inductive characterization of sets of \Dterms by \name{PSP level}, with
\name{``PSP''} suggesting \name{``Proof-SubProof''}.

\begin{defn}
  \label{def-psp}
  We assume a singleton set $\DPRIMSET = \{1\}$ of primitive \Dterms. For
  natural numbers $n \geq 0$, the \defname{PSP level} of $n$, in symbols
  $\psplevel(n)$, is a set of \Dterms specified inductively as
  \beforelistskip
  \begin{enumerate}
  \item $\psplevel(0)\; \eqdef\; \{1\}$.
  \item $\displaystyle \psplevel(n+1)\; \eqdef
    \!\!\!\bigcup_{d \in \psplevel(n)}\!\!\!\!
    \{\D(d, d') \mid d \suptermq d'\}
    \cup 
    \{\D(d', d) \mid d \supterm d'\} ).$
  \end{enumerate}
\end{defn}

Assuming a procedure that enumerates the subterms of a given \Dterm, we can
associate with Definition~\ref{def-psp} straightforwardly a procedure that
enumerates \Dterms interwoven with unification in an axiom-driven way for
increasing PSP levels. The procedure may be improved by caching computed PSP
levels instead of recomputing them.

PSP levels are disjoint. All \Dterms in PSP level $n$ have compacted size~$n$.
However, the cardinality of \Dterms at PSP level $n$ grows slower than that of
\Dterms of compacted size $n$, according to the sequence \OEISNUM{A001147}
\cite{oeis} of integers in contrast to \OEISNUM{A254789}.
Table~\ref{tab-oeis-numbers-psp} shows the initial values of both sequences.
It follows that the enumeration of \Dterms according to the PSP level is
``incomplete'', that is, there are \Dterms that are not a member of any PSP
level.

\begin{table}[t]
  \centering
  \caption{The numbers of distinct \Dterms for a single axiom (or full binary
    trees) in PSP level~$n$ and of compacted size~$n$.}
  \label{tab-oeis-numbers-psp}
  \setlength{\tabcolsep}{4pt}
  \begin{tabular}{llcrrrrrrrrr}
    $n$ && 0 & 1 & 2 & 3 & 4 & 5 & 6 & 7 & 8\\\midrule
    $|\psplevel(n)|$ & \OEISNUM{A001147} &
    1 & 1 & 3 & 15 & 105 & 945 & 10,395 & 135,135 & 2,027,025\\
    Compacted size & \OEISNUM{A254789}
    & 1 & 1 & 3 & 15 & 111 & 1,119 & 14,487 & 230,943 & 4,395,855\\
  \end{tabular}
  \vspace{-10pt} %
\end{table}

Enumeration by PSP level is not just growing slower than by compacted size,
but also apparently simpler to realize. In contrast to DAG enumeration based
on variations of the \name{value-number method}
\cite{aho:compilers:86,cw:ccs}, enumeration by PSP level does not require an
interplay of rigid variables with copies of MGTs \cite{cw:ccs} or forgetting
of variables \cite{eder:cs:1989}. For enumeration by PSP level it is
straightforward to maintain just MGTs.

Most importantly for proof search, the maintenance of MGTs permits simple
incorporation of heuristic restrictions based on their properties as discussed
in Sect.~\ref{subsec-prop-mgt}. This includes discarding \Dterms whose MGT
dimensions exceed configured thresholds, discarding \Dterms whose MGT already
appeared as MGT of a \Dterm produced earlier in the enumeration, and limiting
the overall size of cached solutions by deleting entries according to
heuristic criteria based on properties of the MGTs.

Experiments showed that the enumeration of \Dterms by PSP level indeed
succeeds on many \CD problems. For problems with more than a single axiom, the
definition of $\psplevel(n+1)$ was there extended to include also $\D(d,a)$
and $\D(a,d)$ for $d \in \psplevel(n)$ and arbitrary axiom identifiers $a \in
\DPRIMSET$, not just those occurring in $d$. \SGCD \cite{cw:sgcd} can operate
with enumeration by PSP level. In five such configurations with different
heuristic restrictions, \SGCD enumeration succeeded for 153 of the 196
``basic'' \CD problems in \TPTP~8.0.0\footnote{These ``basic'' \CD problems
are all \CD problems in TPTP~8.0.0 with exception of two with status
\name{satisfiable}, five with a form of detachment that is based on
implication represented by disjunction and negation, and three with a
non-atomic goal theorem.} \cite{cw:cdtools:2022}. Among the 196 problems of the corpus
there are 189 rated $< 1.00$. Among the 153 solutions obtained
with enumeration by PSP level there are 12~problems rated~0.25 and two
rated~0.50.\footnote{For \label{foot-psp-table} details, see
\url{http://cs.christophwernhard.com/cdtools/exp-tptpcd-2022-07/table_4.html}.}
The proofs obtained with enumeration by PSP level tend to have small compacted
size, also for problems where exhaustive enumeration by compacted size to find
a proof with minimal compacted size appears not feasible. The \CCS system
\cite{cw:ccs}, for example, succeeds in finding solutions with minimal
compacted size for only 86 problems.\footnote{Details are included in the
table referenced in footnote~\ref{foot-psp-table}.}

Lemmas obtained from \SGCD with enumeration by PSP level can substantially
increase the performance of first-order provers, including the leading
system \Vampire \cite{vampire}, on CD problems \cite{mrcwzzwb:lemmas:2023}.
Moreover, \TP{LCL073-1}, a problem known as really hard for automated provers,
can be solved by \SGCD in a setting based on enumeration by PSP
level \cite{mrcwzzwb:lemmas:2023}. \SGCD is invoked there twice, for lemma
generation by PSP level and for proving with a combination of enumeration by
PSP level and by height. Both phases use different heuristic restrictions. The
problem is rated 1.00, continuously since ratings were introduced in the \TPTP
in 1997. Mechanically, it was so far proven only once, in 2000 by
Wos \cite{wos:meredith} with transferring outputs and insights between several
invocations of \Otter.

\begin{figure}[p]
  \centering
      {\normalfont \scalebox{0.315}{\input{img6/dg_luk_22}}}
      \vspace{6pt}
  \caption{The DAG representation of the proof of \Syll from \Luk
    (\TP{LCL038-1}) obtained by \SGCD configured to enumeration by PSP level.
    Its compacted size is~22. A dashed arrow indicates that the actual formula
    used as minor premise plays no role to determine the conclusion, which is
    indicated by ``$\n$'' in Meredith's notation (see
    Sect.~\ref{subsec-simp-n}). Node numbers provide correspondence to
    Fig.~\ref{fig-proof-short} below.}
  \label{fig-luk-short-graph}
\end{figure}

For deriving \Syll from \Luk, problem \TP{LCL038-1}, \SGCD with enumeration by
PSP level finds in a few seconds a proof that is substantially smaller than
the proof by \Lukasiewicz and its variation by Meredith: The proof has
compacted size~22, tree size~64 and height~22.
Figure~\ref{fig-luk-short-graph} shows it as a DAG.

\begin{figure}[t] %
  \small
  \centering
\begin{tabular}{r@{\hspace{0.5em}}l@{\hspace{0.0em}}R{3em}@{\hspace{0.5em}}R{3em}}  
1. & $\g{CCCpqrCCrpCsp}$ & \mer{1} & \luk{1}\\
2. & $\g{CCCCpqCrqCqsCtCqs} \mereq \f{D}11$ &  & \subsumesluk{2}\\
3. & $\g{CCCpCqrCCsqCtqCuCCsqCtq} \mereq \f{D}12$ &  & \subsumesluk{3}\\
4. & $\g{CCCpCqrCstCCqtCst} \mereq \f{D}\f{D}\f{D}\f{D}1\f{D}1\f{D}1\f{D}1\f{D}\f{D}\f{D}\f{D}131\mathrm{n}11\mathrm{n}1$ &  & \\
5. & $\g{CCCCpqCrqCCCsCptCrquCvCCCsCptCrqu} \mereq \f{D}1\f{D}\f{D}414$ &  & \\
6. & $\g{CCCpqpCrp} \mereq \f{D}\f{D}31\mathrm{n}$ & \mer{2} & \luk{4}\\
* 7. & $\g{CCpqCCqrCpr} \mereq \f{D}\f{D}\f{D}\f{D}1\f{D}\f{D}55\mathrm{n}1\mathrm{n}1$ & \mer{17} & \luk{29}\\
* 8. & $\g{CCCpqpp} \mereq \f{D}\f{D}\f{D}426\mathrm{n}$ & \mer{18} & \luk{28}\\
* 9. & $\g{CpCqp} \mereq \f{D}\f{D}26\mathrm{n}$ & \mer{19} & \luk{27}\\
\end{tabular}
  \vspace{5pt} %
  \caption{The proof $\DSHORT$ in Meredith's notation
    \cite{meredith:notes:1963}. The two right columns indicate corresponding
    proof steps in \Lukasiewicz's original \cite{luk:1948} and in Meredith's
    variation \cite{meredith:notes:1963} (Fig.~\ref{fig-proof-mer}), as
    explained in Sect.~\ref{subsubsec-mer-luk-labels}.}
  \label{fig-proof-short}
\end{figure}

\begin{figure}[h!] %
  \centering
  \vspace{-5pt}
  \begin{tikzpicture}[>=latex',line join=bevel,scale=0.455]
      \pgfsetlinewidth{1bp}
\pgfsetcolor{black}
  \draw [] (29.891bp,72.0bp) .. controls (27.385bp,72.0bp) and (24.698bp,72.0bp)  .. (22.188bp,72.0bp);
  \draw [] (59.891bp,72.0bp) .. controls (57.385bp,72.0bp) and (54.698bp,72.0bp)  .. (52.188bp,72.0bp);
  \draw [] (89.891bp,89.252bp) .. controls (87.385bp,86.836bp) and (84.698bp,84.244bp)  .. (82.188bp,81.824bp);
  \draw [] (119.89bp,116.25bp) .. controls (117.39bp,113.84bp) and (114.7bp,111.24bp)  .. (112.19bp,108.82bp);
  \draw [] (89.891bp,54.748bp) .. controls (87.385bp,57.164bp) and (84.698bp,59.756bp)  .. (82.188bp,62.176bp);
  \draw [] (149.89bp,126.0bp) .. controls (147.39bp,126.0bp) and (144.7bp,126.0bp)  .. (142.19bp,126.0bp);
  \draw [] (119.89bp,81.748bp) .. controls (117.39bp,84.164bp) and (114.7bp,86.756bp)  .. (112.19bp,89.176bp);
  \draw [] (119.89bp,62.252bp) .. controls (117.39bp,59.836bp) and (114.7bp,57.244bp)  .. (112.19bp,54.824bp);
  \draw [] (119.89bp,27.748bp) .. controls (117.39bp,30.164bp) and (114.7bp,32.756bp)  .. (112.19bp,35.176bp);
\begin{scope}
  \definecolor{strokecol}{rgb}{0.0,0.0,0.0}
  \pgfsetstrokecolor{strokecol}
  \draw (11.0bp,72.0bp) node {1};
\end{scope}
\begin{scope}
  \definecolor{strokecol}{rgb}{0.0,0.0,0.0}
  \pgfsetstrokecolor{strokecol}
  \draw (41.0bp,72.0bp) node {2};
\end{scope}
\begin{scope}
  \definecolor{strokecol}{rgb}{0.0,0.0,0.0}
  \pgfsetstrokecolor{strokecol}
  \draw (71.0bp,72.0bp) node {3};
\end{scope}
\begin{scope}
  \definecolor{strokecol}{rgb}{0.0,0.0,0.0}
  \pgfsetstrokecolor{strokecol}
  \draw (101.0bp,99.0bp) node {4};
\end{scope}
\begin{scope}
  \definecolor{strokecol}{rgb}{0.0,0.0,0.0}
  \pgfsetstrokecolor{strokecol}
  \draw (131.0bp,126.0bp) node {5};
\end{scope}
\begin{scope}
  \definecolor{strokecol}{rgb}{0.0,0.0,0.0}
  \pgfsetstrokecolor{strokecol}
  \draw (101.0bp,45.0bp) node {6};
\end{scope}
\begin{scope}
  \definecolor{strokecol}{rgb}{0.0,0.0,0.0}
  \pgfsetstrokecolor{strokecol}
  \draw (161.0bp,126.0bp) node {7};
\end{scope}
\begin{scope}
  \definecolor{strokecol}{rgb}{0.0,0.0,0.0}
  \pgfsetstrokecolor{strokecol}
  \draw (131.0bp,72.0bp) node {8};
\end{scope}
\begin{scope}
  \definecolor{strokecol}{rgb}{0.0,0.0,0.0}
  \pgfsetstrokecolor{strokecol}
  \draw (131.0bp,18.0bp) node {9};
\end{scope}
  \end{tikzpicture}%
    \vspace{2pt}
  \caption{The label dependency ordering $<_{\compd}$ of proof $\DSHORT$
    as presented in Fig.~\ref{fig-proof-short}.}
  \label{fig-short-ordering}  
\end{figure}

\enlargethispage{13pt} %

This proof of \Syll was supplemented with enumeration techniques to derive also
\Peirce and \Simp \cite{cw:cdtools:2022}. We call the overall proof of the three goal
theorems, whose compacted size is 29, $\DSHORT$. Figure~\ref{fig-proof-short}
shows it in Meredith's notation, where labeled intermediate steps are only
introduced for nodes with multiple incoming edges.
Figure~\ref{fig-short-ordering} shows the corresponding label dependency
ordering.

\begin{sidewaystable}
  \centering
  \small
  \caption{Properties as specified in Sect.~\ref{subsec-properties} of all
    subproofs of $\DSHORT$ (Fig.~\ref{fig-proof-short}).}
  \label{tab-bigshort}
    \setlength{\tabcolsep}{3.1pt}
    \renewcommand{\arraystretch}{0.89}
    \rowcolors{2}{gray!20}{white}
  \scalebox{0.97}{
\begin{tabular}{rlrrrrrrrrrccccrrrcrrccrrrrr}
 & \tabatt{} & \tabatt{$\DSHORT$} & \tabatt{MER} & \tabatt{{\L}UK} & \tabatt{NN} & \tabatt{DC} & \tabatt{DT} & \tabatt{DH} & \tabatt{DI} & \tabatt{DR} & \tabatt{DS} & \tabatt{DP} & \tabatt{DK$_{L}$} & \tabatt{DK$_{R}$} & \tabatt{FC} & \tabatt{FT} & \tabatt{FH} & \tabatt{FO} & \tabatt{MC} & \tabatt{MT} & \tabatt{RS} & \tabatt{RC} & \tabatt{IT$_{U}$} & \tabatt{IT$_{M}$} & \tabatt{IH$_{U}$} & \tabatt{IH$_{M}$} & \tabatt{}\\\midrule
1. & $1$ & \textit{P1} & \mer{1} & \luk{1} & \spec{1} & 0 & 0 & 0 & 18 & 79 & -- & \textbullet & 0 & 0 & 6 & 6 & 3 & \textbullet & 0 & 0 & \textbullet & \textbullet & 2495 & 103 & 15 & 9 & 1.\\
2. & $\f{D}11$ & \textit{P2} & \textcolor{black!30}{\textbullet} & \subsumesluk{2} &  & 1 & 1 & 1 & 3 & 9 & $1{=}1$ & \textbullet & 1 & 1 & 7 & 8 & 4 & \textbullet & 1 & 1 & \textbullet & \textbullet & 929 & 57 & 13 & 7 & 2.\\
3. & $\f{D}12$ & \textit{P3} & \textcolor{black!30}{\textbullet} & \subsumesluk{3} &  & 2 & 2 & 2 & 2 & 7 & $1\lhd$ & \textbullet & 1 & 2 & 8 & 11 & 4 & \textbullet & 2 & 2 & \textbullet & \textbullet & 1032 & 154 & 14 & 10 & 3.\\
4. & $\f{D}13$ &  &  &  &  & 3 & 3 & 3 & 1 & 5 & $1\lhd$ & \textbullet & 1 & 3 & 9 & 11 & 5 & \textbullet & 3 & 3 & \textbullet & \textbullet & 1462 & 290 & 14 & 11 & 4.\\
5. & $\f{D}41$ &  &  &  &  & 4 & 4 & 4 & 1 & 5 & $\rhd1$ & \textbullet & 2 & 3 & 7 & 7 & 5 & \textcolor{black!30}{\textbullet} & 4 & 4 & \textbullet & \textbullet & 564 & 111 & 13 & 10 & 5.\\
6. & $\f{D}5\mathrm{n}$ &  &  &  &  & 5 & 5 & 5 & 1 & 5 & $\rhd\n$ & \textbullet & 3 & 3 & 6 & 6 & 4 & \textbullet & 5 & 5 & \textbullet & \textbullet & 563 & 110 & 12 & 9 & 6.\\
7. & $\f{D}61$ &  &  &  & \spec{14} & 6 & 6 & 6 & 1 & 5 & $\rhd1$ & \textbullet & 4 & 3 & 5 & 5 & 4 & \textbullet & 6 & 6 & \textbullet & \textbullet & 230 & 43 & 11 & 8 & 7.\\
8. & $\f{D}17$ &  &  &  &  & 7 & 7 & 7 & 1 & 5 & $1\lhd$ & \textbullet & 4 & 3 & 8 & 8 & 6 & \textbullet & 7 & 7 & \textbullet & \textbullet & 259 & 45 & 12 & 9 & 8.\\
9. & $\f{D}18$ &  &  &  &  & 8 & 8 & 8 & 1 & 5 & $1\lhd$ & \textbullet & 4 & 3 & 9 & 14 & 6 & \textbullet & 8 & 8 & \textbullet & \textbullet & 569 & 102 & 12 & 9 & 9.\\
10. & $\f{D}19$ &  &  &  &  & 9 & 9 & 9 & 1 & 5 & $1\lhd$ & \textbullet & 4 & 4 & 10 & 11 & 7 & \textbullet & 9 & 9 & \textbullet & \textbullet & 507 & 86 & 13 & 10 & 10.\\
11. & $\f{D}1.10$ &  &  &  &  & 10 & 10 & 10 & 1 & 5 & $1\lhd$ & \textbullet & 4 & 5 & 11 & 17 & 7 & \textbullet & 10 & 10 & \textbullet & \textbullet & 802 & 140 & 13 & 10 & 11.\\
12. & $\f{D}11.1$ &  &  &  &  & 11 & 11 & 11 & 1 & 5 & $\rhd1$ & \textbullet & 4 & 5 & 9 & 10 & 6 & \textcolor{black!30}{\textbullet} & 11 & 11 & \textbullet & \textbullet & 333 & 59 & 12 & 9 & 12.\\
13. & $\f{D}12.\mathrm{n}$ &  &  &  &  & 12 & 12 & 12 & 1 & 5 & $\rhd\n$ & \textbullet & 4 & 5 & 8 & 9 & 5 & \textbullet & 12 & 12 & \textbullet & \textbullet & 332 & 58 & 11 & 8 & 13.\\
14. & $\f{D}31$ &  & \textcolor{black!30}{\textbullet} & \textcolor{black!30}{\textbullet} &  & 3 & 3 & 3 & 1 & 2 & $\rhd1$ & \textbullet & 2 & 2 & 5 & 5 & 4 & \textcolor{black!30}{\textbullet} & 3 & 3 & \textbullet & \textbullet & 8 & 8 & 5 & 5 & 14.\\
15. & $\f{D}13.1$ & \textit{P4} &  &  &  & 13 & 13 & 13 & 3 & 5 & $\rhd1$ & \textbullet & 4 & 5 & 7 & 8 & 4 & \textbullet & $\iv{12}{13}$ & 13 & \textbullet & \textbullet & 197 & 37 & 10 & 7 & 15.\\
16. & $\f{D}15.2$ &  &  &  & \spec{15} & 14 & 15 & 14 & 1 & 1 & $\rhd$ & -- & 5 & 5 & 4 & 5 & 3 & \textbullet & 11 & 13 & \textbullet & \textbullet & 11 & 11 & 5 & 5 & 16.\\
17. & $\f{D}15.1$ &  & \mer{7} & \luk{16} &  & 14 & 14 & 14 & 1 & 2 & $\rhd1$ & \textbullet & 5 & 5 & 6 & 7 & 4 & \textbullet & $\iv{12}{14}$ & 14 & \textbullet & \textbullet & 95 & 62 & 9 & 8 & 17.\\
18. & $\f{D}17.15$ &  &  &  &  & 15 & 28 & 15 & 1 & 2 & $\rhd$ & -- & 6 & 5 & 9 & 11 & 6 & \textbullet & $\iv{12}{15}$ & $\iv{23}{28}$ & \textbullet & \textbullet & 57 & 37 & 8 & 7 & 18.\\
19. & $\f{D}1.18$ & \textit{P5} &  &  &  & 16 & 29 & 16 & 2 & 2 & $1\lhd$ & -- & 6 & 5 & 10 & 16 & 6 & \textbullet & $\iv{12}{16}$ & $\iv{23}{29}$ & \textbullet & \textbullet & 81 & 64 & 8 & 8 & 19.\\
20. & $\f{D}19.19$ &  &  &  &  & 17 & 59 & 17 & 1 & 1 & $=$ & -- & 6 & 5 & 11 & 18 & 7 & \textcolor{black!30}{\textbullet} & $\iv{12}{17}$ & $\iv{24}{59}$ & \textbullet & \textbullet & 33 & 33 & 7 & 7 & 20.\\
21. & $\f{D}20.\mathrm{n}$ &  &  &  &  & 18 & 60 & 18 & 1 & 1 & $\rhd\n$ & -- & 6 & 5 & 10 & 17 & 6 & \textbullet & $\iv{12}{18}$ & $\iv{24}{60}$ & \textbullet & \textbullet & 32 & 32 & 6 & 6 & 21.\\
22. & $\f{D}1.21$ &  &  &  &  & 19 & 61 & 19 & 1 & 1 & $1\lhd$ & -- & 6 & 5 & 11 & 15 & 6 & \textbullet & $\iv{12}{19}$ & $\iv{23}{61}$ & \textbullet & \textbullet & 36 & 36 & 6 & 6 & 22.\\
23. & $\f{D}22.1$ &  &  &  &  & 20 & 62 & 20 & 1 & 1 & $\rhd1$ & -- & 6 & 5 & 8 & 10 & 5 & \textcolor{black!30}{\textbullet} & $\iv{12}{20}$ & $\iv{24}{62}$ & \textbullet & \textbullet & 13 & 13 & 5 & 5 & 23.\\
24. & $\f{D}23.\mathrm{n}$ &  &  &  &  & 21 & 63 & 21 & 1 & 1 & $\rhd\n$ & -- & 6 & 5 & 7 & 9 & 4 & \textbullet & $\iv{12}{21}$ & $\iv{23}{63}$ & \textbullet & \textbullet & 12 & 12 & 4 & 4 & 24.\\
25. & $\f{D}14.\mathrm{n}$ & \textit{P6} & \mer{2} & \luk{4} &  & 4 & 4 & 4 & 2 & 2 & $\rhd\n$ & \textbullet & 3 & 2 & 4 & 4 & 3 & \textbullet & 4 & 4 & \textbullet & \textbullet & 7 & 7 & 4 & 4 & 25.\\
26. & $\f{D}16.25$ &  &  &  &  & 17 & 20 & 15 & 1 & 1 & -- & -- & 6 & 5 & 4 & 4 & 4 & \textcolor{black!30}{\textbullet} & 10 & 14 & \textbullet & \textbullet & 4 & 4 & 4 & 4 & 26.\\
27. & $\f{D}2.25$ &  &  &  & \spec{16} & 5 & 6 & 5 & 1 & 1 & $\lhd$ & -- & 3 & 2 & 3 & 3 & 3 & \textcolor{black!30}{\textbullet} & 5 & 6 & \textbullet & \textbullet & 3 & 3 & 3 & 3 & 27.\\
28. & $\f{D}24.1$ & \textit{P7} & \mer{17} & \luk{29} & \spec{11} & 22 & 64 & 22 & 0 & 1 & $\rhd1$ & -- & 6 & 5 & 5 & 5 & 3 & \textbullet & $\iv{12}{22}$ & $\iv{24}{64}$ & \textbullet & \textbullet & 5 & 5 & 3 & 3 & 28.\\
29. & $\f{D}26.\mathrm{n}$ & \textit{P8} & \mer{18} & \luk{28} & \spec{12} & 18 & 21 & 16 & 0 & 1 & $\rhd\n$ & -- & 7 & 5 & 3 & 3 & 3 & \textbullet & 11 & 15 & \textbullet & \textbullet & 3 & 3 & 3 & 3 & 29.\\
30. & $\f{D}27.\mathrm{n}$ & \textit{P9} & \mer{19} & \luk{27} & \spec{13} & 6 & 7 & 6 & 0 & 1 & $\rhd\n$ & -- & 3 & 2 & 2 & 2 & 2 & \textbullet & 6 & 7 & \textbullet & \textbullet & 2 & 2 & 2 & 2 & 30.\\
\end{tabular}
}
\end{sidewaystable}

\begin{table}[h!] %
  \centering
  \caption{``Named'' formulas that occur as MGTs of subproofs of $\DSHORT$
    and are not listed in Table~\ref{tab-formula-names}.}
  \label{tab-named-shortproof}
  \begin{tabular}{rll}
    Id & Formula & Names\\\midrule
    $\spec{14}$ & $\g{CpCCCqprCsr}$ & $\comb{B (C (B K)) K}$\\
    $\spec{15}$ & $\g{CCpCpqCrCpq}$ & $\comb{C (K W)}$, $\comb{B K W}$\\
    $\spec{16}$ & $\g{CpCqCrq}$ & $\comb{K K}$\\
  \end{tabular}
\end{table}

The criteria on combinator terms as formula names from
footnote~\ref{footnote-comb-labels} (p.~\pageref{footnote-comb-labels}) lead
for $\DSHORT$ to three additional ``named'' formulas, which are shown in
Table~\ref{tab-named-shortproof}. Properties of all subproofs of $\DSHORT$ are
shown in Table~\ref{tab-bigshort}, in analogy to Tables~\ref{tab-bigmer}
and~\ref{tab-bigluk}.

The small \Dterm size apparently comes for the price of a slight extension of
the maximal size of MGTs: For $\DSHORT$ the maximal value of \xatt{FT} is 17
and the maximal value of \xatt{FH} is 7, compared to 15 and~6 for $\DMER$ and
$\DLUK$, respectively. Subproof~26 is the only one where the \xatt{DS} column
has an empty value, indicating that it cannot be obtained by a PSP induction
step from some subproof appearing at a row further above. Subproof~26 does not
belong to the subproof of \Syll, subproof~27, which was obtained purely by PSP
level enumeration, but to the supplements to prove \Peirce and \Simp.
Figure~\ref{fig-luk-short-graph} shows just the proof of \Syll as a DAG.

A further size reduction of our proof of \TP{LCL038-1} from
Fig.~\ref{fig-luk-short-graph} can be achieved with combinatory
compression \cite{cw:ccs}: The tree grammar obtained from a grammar-based tree
compression tool \cite{lohrey:treerepair:2013} for the proof can be converted
to a generalized form of \Dterm that permits leaves labeled by combinators
expressing proof structure transformations of the original \Dterm. It has
compacted size~19, height~15, but tree size~119
\cite{cw:cdtools:2022}.

\section{Conclusion}
\label{sec-conclusion}

Our leading motivation has been improving proof search in ATP by the incorporation
of operations that are more global than extending a set of formulas by an
inferred formula. A comparative analysis of proof systems seemed necessary to
this end. Our focus here was on Meredith's system known as condensed
detachment (CD). For it we have elaborated a new formal reconstruction as a
special case of the connection method (CM).

Our reconstruction preserves an important aspect of CD, the reification of
proofs as terms, more specifically \Dterms, which may be regarded as full
binary trees. The underlying ATP model is the CM, where structures formed by
connections attached to the formula provide the key concept. \Dterms then
are one way to represent such structures for problems of a certain restricted
class.

The incorporation of lemmas belongs to the key global operations for reducing
the amount of search and the size of proofs. We specifically considered a form
of lemmas that corresponds to the repeated use of a substructure -- subtree or
subterm -- in a proof. Or, in other words, the interplay of trees as proof
structures and their representation as DAGs. Lemmas are then characterized by
way of \Dterms, along with various measures and properties concerning the
proof structure as well as the proven formulas.

The resulting formalism has opened the door towards enhancement of ATP systems
by taking into account global features within the proof search, suggesting
various techniques that are immediately applicable in practice. First
experiments on the restricted kind of problems considered in the paper, which,
among others, include the 196 CD problems in the TPTP problem collection, are
promising and encourage future work.

On the basis of the formalism we analyzed and compared the remarkable historic
proofs by {\L}ukasiewicz and Meredith of a problem stated by the former. The
problems played a historic role also in ATP, which is surveyed in the paper.
However, an in-depth analysis of its human-made proofs has not been undertaken
before. In a particular experiment we ``learned'' from the human-made proofs
by converting an observed structural feature into a novel method for proof
search. It finds short proofs for many problems for which a systematic search for
shortest proofs appears unfeasible. In particular, for {\L}ukasiewicz's
problem, it quickly yields a particularly short proof, shorter than the
human-made model proofs, and drastically shorter than all known proofs by ATP
systems.

In the longer run, our approach lends itself towards supporting ATP by machine
learning (see, e.g., \cite{faerber:2021:mlct,enigma:2020,mrcwzzwb:lemmas:2023}). This is because
the reification of proof structures provides information that can be exploited
in the learning process and is not available within other ATP approaches.

\paragraph{Acknowledgments.}

We thank Michael Rawson and Zsolt Zombori as well as anonymous reviewers of
CADE~2021 and of JAR for helpful comments and suggestions that led to
significant improvements of the presentation. Funded by the Deutsche
Forschungsgemeinschaft (DFG, German Research Foundation) --
Project-ID~457292495. The work was supported by the North-German
Supercomputing Alliance (HLRN).

\nocite{rezus:2020:witness} %
\nocite{luk:selected:1970} %
\phantomsection
\fancyhead[CO]{\small References}
\fancyhead[CE]{\small References}
\addcontentsline{toc}{section}{References}
\nocite{bibsettings}
\bibliography{biblukas202}


\begin{thebibliography}{85}
\ifx \bisbn   \undefined \def \bisbn  #1{ISBN #1}\fi
\ifx \binits  \undefined \def \binits#1{#1}\fi
\ifx \bauthor  \undefined \def \bauthor#1{#1}\fi
\ifx \batitle  \undefined \def \batitle#1{#1}\fi
\ifx \bjtitle  \undefined \def \bjtitle#1{#1}\fi
\ifx \bvolume  \undefined \def \bvolume#1{\textbf{#1}}\fi
\ifx \byear  \undefined \def \byear#1{#1}\fi
\ifx \bissue  \undefined \def \bissue#1{#1}\fi
\ifx \bfpage  \undefined \def \bfpage#1{#1}\fi
\ifx \blpage  \undefined \def \blpage #1{#1}\fi
\ifx \burl  \undefined \def \burl#1{\textsf{#1}}\fi
\ifx \doiurl  \undefined \def \doiurl#1{\url{https://doi.org/#1}}\fi
\ifx \betal  \undefined \def \betal{\textit{et al.}}\fi
\ifx \binstitute  \undefined \def \binstitute#1{#1}\fi
\ifx \binstitutionaled  \undefined \def \binstitutionaled#1{#1}\fi
\ifx \bctitle  \undefined \def \bctitle#1{#1}\fi
\ifx \beditor  \undefined \def \beditor#1{#1}\fi
\ifx \bpublisher  \undefined \def \bpublisher#1{#1}\fi
\ifx \bbtitle  \undefined \def \bbtitle#1{#1}\fi
\ifx \bedition  \undefined \def \bedition#1{#1}\fi
\ifx \bseriesno  \undefined \def \bseriesno#1{#1}\fi
\ifx \blocation  \undefined \def \blocation#1{#1}\fi
\ifx \bsertitle  \undefined \def \bsertitle#1{#1}\fi
\ifx \bsnm \undefined \def \bsnm#1{#1}\fi
\ifx \bsuffix \undefined \def \bsuffix#1{#1}\fi
\ifx \bparticle \undefined \def \bparticle#1{#1}\fi
\ifx \barticle \undefined \def \barticle#1{#1}\fi
\bibcommenthead
\ifx \bconfdate \undefined \def \bconfdate #1{#1}\fi
\ifx \botherref \undefined \def \botherref #1{#1}\fi
\ifx \url \undefined \def \url#1{\textsf{#1}}\fi
\ifx \bchapter \undefined \def \bchapter#1{#1}\fi
\ifx \bbook \undefined \def \bbook#1{#1}\fi
\ifx \bcomment \undefined \def \bcomment#1{#1}\fi
\ifx \oauthor \undefined \def \oauthor#1{#1}\fi
\ifx \citeauthoryear \undefined \def \citeauthoryear#1{#1}\fi
\ifx \endbibitem  \undefined \def \endbibitem {}\fi
\ifx \bconflocation  \undefined \def \bconflocation#1{#1}\fi
\ifx \arxivurl  \undefined \def \arxivurl#1{\textsf{#1}}\fi
\csname PreBibitemsHook\endcsname

%%% 1
\bibitem[\protect\citeauthoryear{Aho et~al.}{1986}]{aho:compilers:86}
\begin{bbook}
\bauthor{\bsnm{Aho}, \binits{A.V.}},
\bauthor{\bsnm{Sethi}, \binits{R.}},
\bauthor{\bsnm{Ullman}, \binits{J.D.}}:
\bbtitle{Compilers -- Principles, Techniques, and Tools}.
\bpublisher{Addison-Wesley},
\blocation{Reading, MA}
(\byear{1986})
\end{bbook}
\endbibitem

%%% 2
\bibitem[\protect\citeauthoryear{Astrachan and
  Stickel}{1992}]{astrachan:stickel:caching:1992}
\begin{bchapter}
\bauthor{\bsnm{Astrachan}, \binits{O.L.}},
\bauthor{\bsnm{Stickel}, \binits{M.E.}}:
\bctitle{Caching and lemmaizing in model elimination theorem provers}.
In: \beditor{\bsnm{Kapur}, \binits{D.}} (ed.)
\bbtitle{CADE-11},
pp. \bfpage{224}--\blpage{238}.
\bpublisher{Springer},
\blocation{Berlin}
(\byear{1992}).
\doiurl{10.1007/3-540-55602-8_168}
\end{bchapter}
\endbibitem

%%% 3
\bibitem[\protect\citeauthoryear{Baumgartner et~al.}{1996}]{hypertableaux}
\begin{bchapter}
\bauthor{\bsnm{Baumgartner}, \binits{P.}},
\bauthor{\bsnm{Furbach}, \binits{U.}},
\bauthor{\bsnm{Niemel{\"a}}, \binits{I.}}:
\bctitle{Hyper tableaux}.
In: \beditor{\bsnm{Alferes}, \binits{J.J.}},
\beditor{\bsnm{Pereira}, \binits{L.M.}},
\beditor{\bsnm{Orlowska}, \binits{E.}} (eds.)
\bbtitle{JELIA'96}.
\bsertitle{LNCS (LNAI)},
vol. \bseriesno{1126},
pp. \bfpage{1}--\blpage{17}.
\bpublisher{Springer},
\blocation{Berlin}
(\byear{1996}).
\doiurl{10.1007/3-540-61630-6_1}
\end{bchapter}
\endbibitem

%%% 4
\bibitem[\protect\citeauthoryear{Bibel}{1982}]{bibel:atp:1982}
\begin{bbook}
\bauthor{\bsnm{Bibel}, \binits{W.}}:
\bbtitle{Automated Theorem Proving}.
\bpublisher{Vieweg},
\blocation{Braunschweig}
(\byear{1982}).
\doiurl{10.1007/978-3-322-90102-6}.
\bcomment{Second edition 1987}
\end{bbook}
\endbibitem

%%% 5
\bibitem[\protect\citeauthoryear{Bibel}{1993}]{bibel:deduction:1993}
\begin{bbook}
\bauthor{\bsnm{Bibel}, \binits{W.}}:
\bbtitle{Deduction: Automated Logic}.
\bpublisher{Academic Press},
\blocation{London}
(\byear{1993})
\end{bbook}
\endbibitem

%%% 6
\bibitem[\protect\citeauthoryear{Bibel}{2024a}]{wb:comparison:2024}
\begin{bchapter}
\bauthor{\bsnm{Bibel}, \binits{W.}}:
\bctitle{Comparison of proof methods}.
In: \beditor{\bsnm{Otten}, \binits{J.}},
\beditor{\bsnm{Bibel}, \binits{W.}} (eds.)
\bbtitle{AReCCa 2023}.
\bsertitle{CEUR Workshop Proc.},
vol. \bseriesno{3613},
pp. \bfpage{119}--\blpage{132}.
\bpublisher{CEUR-WS.org},
\blocation{Aachen}
(\byear{2024})
\end{bchapter}
\endbibitem

%%% 7
\bibitem[\protect\citeauthoryear{Bibel}{2024b}]{wb:conjecture:2024}
\begin{botherref}
\oauthor{\bsnm{Bibel}, \binits{W.}}:
A conjecture for {ATP} research.
CoRR
\textbf{abs/2403.10334}
(2024).
\doiurl{10.48550/2403.10334}
\end{botherref}
\endbibitem

%%% 8
\bibitem[\protect\citeauthoryear{Bibel and Otten}{2020}]{bibel:otten:2020}
\begin{bchapter}
\bauthor{\bsnm{Bibel}, \binits{W.}},
\bauthor{\bsnm{Otten}, \binits{J.}}:
\bctitle{From {S}ch\"utte's formal systems to modern automated deduction}.
In: \beditor{\bsnm{Kahle}, \binits{R.}},
\beditor{\bsnm{Rathjen}, \binits{M.}} (eds.)
\bbtitle{The Legacy of Kurt Sch\"utte},
pp. \bfpage{215}--\blpage{249}.
\bpublisher{Springer},
\blocation{Cham}
(\byear{2020}).
\bcomment{Chap. 13}.
\doiurl{10.1007/978-3-030-49424-7\_13}
\end{bchapter}
\endbibitem

%%% 9
\bibitem[\protect\citeauthoryear{Bull and Cubrinovska}{2018}]{bull:interview}
\begin{botherref}
\oauthor{\bsnm{Bull}, \binits{R.}},
\oauthor{\bsnm{Cubrinovska}, \binits{A.}}:
Interview with {R}obert {B}ull.
online: \textit{Popper and Prior in New Zealand},
  \url{http://popper-prior.nz/items/show/255}, accessed Jul 09, 2024
(2018)
\end{botherref}
\endbibitem

%%% 10
\bibitem[\protect\citeauthoryear{Bunder}{1995}]{bunder:cd:1995}
\begin{barticle}
\bauthor{\bsnm{Bunder}, \binits{M.W.}}:
\batitle{A simplified form of condensed detachment}.
\bjtitle{J. Log., Lang. Inf.}
\bvolume{4}(\bissue{2}),
\bfpage{169}--\blpage{173}
(\byear{1995}).
\doiurl{10.1007/BF01048619}
\end{barticle}
\endbibitem

%%% 11
\bibitem[\protect\citeauthoryear{Claessen and
  Smallbone}{2018}]{claessen:smallbone:2018}
\begin{bchapter}
\bauthor{\bsnm{Claessen}, \binits{K.}},
\bauthor{\bsnm{Smallbone}, \binits{N.}}:
\bctitle{Efficient encodings of first-order {H}orn formulas in equational
  logic}.
In: \beditor{\bsnm{Galmiche}, \binits{D.}},
\beditor{\bsnm{Schulz}, \binits{S.}},
\beditor{\bsnm{Sebastiani}, \binits{R.}} (eds.)
\bbtitle{IJCAR 2018}.
\bsertitle{LNCS (LNAI)},
vol. \bseriesno{10900},
pp. \bfpage{388}--\blpage{404}.
\bpublisher{Springer},
\blocation{Cham}
(\byear{2018}).
\doiurl{10.1007/978-3-319-94205-6\_26}
\end{bchapter}
\endbibitem

%%% 12
\bibitem[\protect\citeauthoryear{Dershowitz and
  Jouannaud}{1991}]{dershowitz:notations:1991}
\begin{barticle}
\bauthor{\bsnm{Dershowitz}, \binits{N.}},
\bauthor{\bsnm{Jouannaud}, \binits{J.}}:
\batitle{Notations for rewriting}.
\bjtitle{Bull. {EATCS}}
\bvolume{43},
\bfpage{162}--\blpage{174}
(\byear{1991})
\end{barticle}
\endbibitem

%%% 13
\bibitem[\protect\citeauthoryear{Downey et~al.}{1980}]{downey:variations:1980}
\begin{barticle}
\bauthor{\bsnm{Downey}, \binits{P.J.}},
\bauthor{\bsnm{Sethi}, \binits{R.}},
\bauthor{\bsnm{Tarjan}, \binits{R.E.}}:
\batitle{Variations on the common subexpression problem}.
\bjtitle{JACM}
\bvolume{27}(\bissue{4}),
\bfpage{758}--\blpage{771}
(\byear{1980}).
\doiurl{10.1145/322217.322228}
\end{barticle}
\endbibitem

%%% 14
\bibitem[\protect\citeauthoryear{Eder}{1985}]{eder:subst:1985}
\begin{barticle}
\bauthor{\bsnm{Eder}, \binits{E.}}:
\batitle{Properties of substitutions and unification}.
\bjtitle{J. Symb. Comput.}
\bvolume{1}(\bissue{1}),
\bfpage{31}--\blpage{46}
(\byear{1985}).
\doiurl{10.1016/S0747-7171(85)80027-4}
\end{barticle}
\endbibitem

%%% 15
\bibitem[\protect\citeauthoryear{Eder}{1989}]{eder:cs:1989}
\begin{bchapter}
\bauthor{\bsnm{Eder}, \binits{E.}}:
\bctitle{A comparison of the resolution calculus and the connection method, and
  a new calculus generalizing both methods}.
In: \beditor{\bsnm{B{\"o}rger}, \binits{E.}},
\beditor{\bsnm{{Kleine B{\"u}ning}}, \binits{H.}},
\beditor{\bsnm{Richter}, \binits{M.M.}} (eds.)
\bbtitle{CSL~'88}.
\bsertitle{LNCS},
vol. \bseriesno{385},
pp. \bfpage{80}--\blpage{98}.
\bpublisher{Springer},
\blocation{Berlin}
(\byear{1989}).
\doiurl{10.1007/BFb0026296}
\end{bchapter}
\endbibitem

%%% 16
\bibitem[\protect\citeauthoryear{Eder}{1992}]{eder:relative:1992}
\begin{bbook}
\bauthor{\bsnm{Eder}, \binits{E.}}:
\bbtitle{Relative Complexities of First Order Calculi}.
\bpublisher{Vieweg},
\blocation{Braunschweig}
(\byear{1992}).
\doiurl{10.1007/978-3-322-84222-0}
\end{bbook}
\endbibitem

%%% 17
\bibitem[\protect\citeauthoryear{F{\"{a}}rber et~al.}{2021}]{faerber:2021:mlct}
\begin{barticle}
\bauthor{\bsnm{F{\"{a}}rber}, \binits{M.}},
\bauthor{\bsnm{Kaliszyk}, \binits{C.}},
\bauthor{\bsnm{Urban}, \binits{J.}}:
\batitle{Machine learning guidance for connection tableaux}.
\bjtitle{J. Autom. Reasoning}
\bvolume{65}(\bissue{2}),
\bfpage{287}--\blpage{320}
(\byear{2021}).
\doiurl{10.1007/s10817-020-09576-7}
\end{barticle}
\endbibitem

%%% 18
\bibitem[\protect\citeauthoryear{Fitelson and
  Wos}{2001}]{fitelson:missing:2001}
\begin{barticle}
\bauthor{\bsnm{Fitelson}, \binits{B.}},
\bauthor{\bsnm{Wos}, \binits{L.}}:
\batitle{Missing proofs found}.
\bjtitle{J. Autom. Reasoning}
\bvolume{27}(\bissue{2}),
\bfpage{201}--\blpage{225}
(\byear{2001}).
\doiurl{10.1023/A:1010695827789}
\end{barticle}
\endbibitem

%%% 19
\bibitem[\protect\citeauthoryear{Flajolet et~al.}{1990}]{flajolet:1990}
\begin{bchapter}
\bauthor{\bsnm{Flajolet}, \binits{P.}},
\bauthor{\bsnm{Sipala}, \binits{P.}},
\bauthor{\bsnm{Steyaert}, \binits{J.}}:
\bctitle{Analytic variations on the common subexpression problem}.
In: \bbtitle{ICALP90}.
\bsertitle{LNCS},
vol. \bseriesno{443},
pp. \bfpage{220}--\blpage{234}.
\bpublisher{Springer},
\blocation{Berlin}
(\byear{1990}).
\doiurl{10.1007/BFb0032034}
\end{bchapter}
\endbibitem

%%% 20
\bibitem[\protect\citeauthoryear{Fuchs and Fuchs}{1997}]{fuchs:code:1997}
\begin{bchapter}
\bauthor{\bsnm{Fuchs}, \binits{D.}},
\bauthor{\bsnm{Fuchs}, \binits{M.}}:
\bctitle{{CODE}: A powerful prover for problems of condensed detachment}.
In: \beditor{\bsnm{McCune}, \binits{W.}} (ed.)
\bbtitle{CADE-14},
pp. \bfpage{260}--\blpage{263}.
\bpublisher{Springer},
\blocation{Berlin}
(\byear{1997}).
\doiurl{10.1007/3-540-63104-6_25}
\end{bchapter}
\endbibitem

%%% 21
\bibitem[\protect\citeauthoryear{Fuchs}{1999}]{fuchs:lemmas:ijcai:1999}
\begin{bchapter}
\bauthor{\bsnm{Fuchs}, \binits{M.}}:
\bctitle{Lemma generation for model elimination by combining top-down and
  bottom-up inference}.
In: \beditor{\bsnm{Dean}, \binits{T.}} (ed.)
\bbtitle{IJCAI 1999},
pp. \bfpage{4}--\blpage{9}.
\bpublisher{Morgan Kaufmann},
\blocation{San Francisco, CA}
(\byear{1999}).
\burl{http://ijcai.org/Proceedings/99-1/Papers/001.pdf}
\end{bchapter}
\endbibitem

%%% 22
\bibitem[\protect\citeauthoryear{Genitrini et~al.}{2020}]{genitrini:2020}
\begin{barticle}
\bauthor{\bsnm{Genitrini}, \binits{A.}},
\bauthor{\bsnm{Gittenberger}, \binits{B.}},
\bauthor{\bsnm{Kauers}, \binits{M.}},
\bauthor{\bsnm{Wallner}, \binits{M.}}:
\batitle{Asymptotic enumeration of compacted binary trees of bounded right
  height}.
\bjtitle{J. Comb. Theory, Ser. A}
\bvolume{172},
\bfpage{105177}
(\byear{2020}).
\doiurl{10.1016/j.jcta.2019.105177}
\end{barticle}
\endbibitem

%%% 23
\bibitem[\protect\citeauthoryear{H{\"{a}}hnle}{2001}]{handbook:ar:haehnle}
\begin{bchapter}
\bauthor{\bsnm{H{\"{a}}hnle}, \binits{R.}}:
\bctitle{Tableaux and related methods}.
In: \beditor{\bsnm{Robinson}, \binits{A.}},
\beditor{\bsnm{Voronkov}, \binits{A.}} (eds.)
\bbtitle{Handb. of Autom. Reasoning}
vol. \bseriesno{1},
pp. \bfpage{101}--\blpage{178}.
\bpublisher{Elsevier},
\blocation{Amsterdam}
(\byear{2001}).
\bcomment{Chap. 3}.
\doiurl{10.1016/b978-044450813-3/50005-9}
\end{bchapter}
\endbibitem

%%% 24
\bibitem[\protect\citeauthoryear{Hindley}{1997}]{hindley:book:1997}
\begin{bbook}
\bauthor{\bsnm{Hindley}, \binits{J.R.}}:
\bbtitle{Basic Simple Type Theory}.
\bpublisher{Cambridge University Press},
\blocation{Cambridge}
(\byear{1997}).
\doiurl{10.1017/CBO9780511608865}
\end{bbook}
\endbibitem

%%% 25
\bibitem[\protect\citeauthoryear{Hindley and
  Meredith}{1990}]{hindley:meredith:cd:1990}
\begin{barticle}
\bauthor{\bsnm{Hindley}, \binits{J.R.}},
\bauthor{\bsnm{Meredith}, \binits{D.}}:
\batitle{Principal type-schemes and condensed detachment}.
\bjtitle{J. Symb. Log.}
\bvolume{55}(\bissue{1}),
\bfpage{90}--\blpage{105}
(\byear{1990}).
\doiurl{10.2307/2274956}
\end{barticle}
\endbibitem

%%% 26
\bibitem[\protect\citeauthoryear{Jakubuv et~al.}{2020}]{enigma:2020}
\begin{bchapter}
\bauthor{\bsnm{Jakubuv}, \binits{J.}},
\bauthor{\bsnm{Chvalovsk{\'{y}}}, \binits{K.}},
\bauthor{\bsnm{Ols{\'{a}}k}, \binits{M.}},
\bauthor{\bsnm{Piotrowski}, \binits{B.}},
\bauthor{\bsnm{Suda}, \binits{M.}},
\bauthor{\bsnm{Urban}, \binits{J.}}:
\bctitle{{ENIGMA} {Anonymous}: Symbol-independent inference guiding machine
  (system description)}.
In: \beditor{\bsnm{Peltier}, \binits{N.}},
\beditor{\bsnm{Sofronie-Stokkermans}, \binits{V.}} (eds.)
\bbtitle{IJCAR 2020}.
\bsertitle{LNCS (LNAI)},
vol. \bseriesno{12167},
pp. \bfpage{448}--\blpage{463}.
\bpublisher{Springer},
\blocation{Cham}
(\byear{2020}).
\doiurl{10.1007/978-3-030-51054-1\_29}
\end{bchapter}
\endbibitem

%%% 27
\bibitem[\protect\citeauthoryear{Kalman}{1983}]{kalman:cd:1983}
\begin{barticle}
\bauthor{\bsnm{Kalman}, \binits{J.A.}}:
\batitle{Condensed detachment as a rule of inference}.
\bjtitle{Studia Logica}
\bvolume{42},
\bfpage{443}--\blpage{451}
(\byear{1983}).
\doiurl{10.1007/BF01371632}
\end{barticle}
\endbibitem

%%% 28
\bibitem[\protect\citeauthoryear{Knuth}{1968}]{knuth:1}
\begin{bbook}
\bauthor{\bsnm{Knuth}, \binits{D.E.}}:
\bbtitle{The Art of Computer Programming: Volume 1 / Fundamental Algorithms}.
\bpublisher{Addison-Wesley},
\blocation{Reading, MA}
(\byear{1968})
\end{bbook}
\endbibitem

%%% 29
\bibitem[\protect\citeauthoryear{Kov{\'a}cs and Voronkov}{2013}]{vampire}
\begin{bchapter}
\bauthor{\bsnm{Kov{\'a}cs}, \binits{L.}},
\bauthor{\bsnm{Voronkov}, \binits{A.}}:
\bctitle{First-order theorem proving and \textsc{Vampire}}.
In: \beditor{\bsnm{Sharygina}, \binits{N.}},
\beditor{\bsnm{Veith}, \binits{H.}} (eds.)
\bbtitle{CAV 2013}.
\bsertitle{LNCS},
vol. \bseriesno{8044},
pp. \bfpage{1}--\blpage{35}
(\byear{2013}).
\doiurl{10.1007/978-3-642-39799-8_1}.
\bcomment{Springer}
\end{bchapter}
\endbibitem

%%% 30
\bibitem[\protect\citeauthoryear{Lemmon
  et~al.}{1969}]{lemmon:meredith:purestrict:1957}
\begin{bchapter}
\bauthor{\bsnm{Lemmon}, \binits{E.J.}},
\bauthor{\bsnm{Meredith}, \binits{C.A.}},
\bauthor{\bsnm{Meredith}, \binits{D.}},
\bauthor{\bsnm{Prior}, \binits{A.N.}},
\bauthor{\bsnm{Thomas}, \binits{I.}}:
\bctitle{Calculi of pure strict implication}.
In: \beditor{\bsnm{Davis}, \binits{J.W.}},
\beditor{\bsnm{Hockney}, \binits{D.J.}},
\beditor{\bsnm{Wilson}, \binits{W.K.}} (eds.)
\bbtitle{Philosophical Logic},
pp. \bfpage{215}--\blpage{250}.
\bpublisher{Springer},
\blocation{Dordrecht}
(\byear{1969}).
\doiurl{10.1007/978-94-010-9614-0\_17}.
\bcomment{Reprint of a technical report, Canterbury University College,
  Christchurch, 1957}
\end{bchapter}
\endbibitem

%%% 31
\bibitem[\protect\citeauthoryear{Letz}{1999}]{letz:habil}
\begin{botherref}
\oauthor{\bsnm{Letz}, \binits{R.}}:
Tableau and connection calculi. structure, complexity, implementation.
Habilitationsschrift,
TU München
(1999).
\url{https://web.archive.org/web/20230604101128/https://www2.tcs.ifi.lmu.de/~letz/habil.ps},
  accessed Jul 09, 2024
\end{botherref}
\endbibitem

%%% 32
\bibitem[\protect\citeauthoryear{Letz et~al.}{1994}]{letz:cut:1994}
\begin{barticle}
\bauthor{\bsnm{Letz}, \binits{R.}},
\bauthor{\bsnm{Mayr}, \binits{K.}},
\bauthor{\bsnm{Goller}, \binits{C.}}:
\batitle{Controlled integration of the cut rule into connection tableaux
  calculi}.
\bjtitle{J. Autom. Reasoning}
\bvolume{13}(\bissue{3}),
\bfpage{297}--\blpage{337}
(\byear{1994})
\end{barticle}
\endbibitem

%%% 33
\bibitem[\protect\citeauthoryear{Letz et~al.}{1992}]{setheo:92}
\begin{barticle}
\bauthor{\bsnm{Letz}, \binits{R.}},
\bauthor{\bsnm{Schumann}, \binits{J.}},
\bauthor{\bsnm{Bayerl}, \binits{S.}},
\bauthor{\bsnm{Bibel}, \binits{W.}}:
\batitle{{SETHEO:} {A} high-performance theorem prover}.
\bjtitle{J. Autom. Reasoning}
\bvolume{8}(\bissue{2}),
\bfpage{183}--\blpage{212}
(\byear{1992}).
\doiurl{10.1007/BF00244282}
\end{barticle}
\endbibitem

%%% 34
\bibitem[\protect\citeauthoryear{Lohrey}{2015}]{lohrey:survey:2015}
\begin{bchapter}
\bauthor{\bsnm{Lohrey}, \binits{M.}}:
\bctitle{Grammar-based tree compression}.
In: \beditor{\bsnm{Potapov}, \binits{I.}} (ed.)
\bbtitle{DLT 2015}.
\bsertitle{LNCS},
vol. \bseriesno{9168},
pp. \bfpage{46}--\blpage{57}.
\bpublisher{Springer},
\blocation{Cham}
(\byear{2015}).
\doiurl{10.1007/978-3-319-21500-6\_3}
\end{bchapter}
\endbibitem

%%% 35
\bibitem[\protect\citeauthoryear{Lohrey et~al.}{2013}]{lohrey:treerepair:2013}
\begin{barticle}
\bauthor{\bsnm{Lohrey}, \binits{M.}},
\bauthor{\bsnm{Maneth}, \binits{S.}},
\bauthor{\bsnm{Mennicke}, \binits{R.}}:
\batitle{{XML} tree structure compression using {RePair}}.
\bjtitle{Inf. Syst.}
\bvolume{38}(\bissue{8}),
\bfpage{1150}--\blpage{1167}
(\byear{2013}).
\doiurl{10.1016/j.is.2013.06.006}.
\bcomment{System available from \url{https://github.com/dc0d32/TreeRePair},
  accessed Jul 09, 2024}
\end{barticle}
\endbibitem

%%% 36
\bibitem[\protect\citeauthoryear{Loveland}{1978}]{loveland:1978}
\begin{bbook}
\bauthor{\bsnm{Loveland}, \binits{D.W.}}:
\bbtitle{Automated Theorem Proving: A Logical Basis}.
\bpublisher{North-Holland},
\blocation{Amsterdam}
(\byear{1978})
\end{bbook}
\endbibitem

%%% 37
\bibitem[\protect\citeauthoryear{{\L}ukasiewicz}{1948}]{luk:1948}
\begin{bchapter}
\bauthor{\bsnm{{\L}ukasiewicz}, \binits{J.}}:
\bctitle{The shortest axiom of the implicational calculus of propositions}.
In: \bbtitle{Proc. of the Royal Irish Academy},
vol. \bseriesno{52, Sect.~A, No.~3},
pp. \bfpage{25}--\blpage{33}
(\byear{1948}).
\burl{http://www.jstor.org/stable/20488489}
\end{bchapter}
\endbibitem

%%% 38
\bibitem[\protect\citeauthoryear{{\L}ukasiewicz}{1963}]{luk:book}
\begin{bbook}
\bauthor{\bsnm{{\L}ukasiewicz}, \binits{J.}}:
\bbtitle{Elements of Mathematical Logic}.
\bpublisher{Pergamon Press},
\blocation{Oxford}
(\byear{1963}).
\bcomment{English translation of the second edition (1958) of \textit{Elementy
  logiki matematycznej}, PWM, Warszawa}
\end{bbook}
\endbibitem

%%% 39
\bibitem[\protect\citeauthoryear{{\L}ukasiewicz}{1970}]{luk:selected:1970}
\begin{bbook}
\bauthor{\bsnm{{\L}ukasiewicz}, \binits{J.}}:
\bbtitle{Selected Works}.
\bpublisher{North-Holland},
\blocation{Amsterdam}
(\byear{1970}).
\bcomment{Edited by L. Borkowski}
\end{bbook}
\endbibitem

%%% 40
\bibitem[\protect\citeauthoryear{{\L}ukasiewicz and
  Tarski}{1930}]{luk:tarski:aussagenkalkuel:1930}
\begin{botherref}
\oauthor{\bsnm{{\L}ukasiewicz}, \binits{J.}},
\oauthor{\bsnm{Tarski}, \binits{A.}}:
Untersuchungen über den {A}ussagenkalkül.
Comptes rendus des séances de la Soc. d. Sciences et d. Lettres de Varsovie
\textbf{23}
(1930).
{E}nglish translation in \cite{luk:selected:1970}, p.~131--152
\end{botherref}
\endbibitem

%%% 41
\bibitem[\protect\citeauthoryear{Lusk and McCune}{1992}]{roo:parallel:1992}
\begin{bchapter}
\bauthor{\bsnm{Lusk}, \binits{E.L.}},
\bauthor{\bsnm{McCune}, \binits{W.W.}}:
\bctitle{Experiments with {ROO}, a parallel automated deduction system}.
In: \beditor{\bsnm{Fronh{\"o}fer}, \binits{B.}},
\beditor{\bsnm{Wrightson}, \binits{G.}} (eds.)
\bbtitle{Parallelization in Inference Systems}.
\bsertitle{LNCS (LNAI)},
vol. \bseriesno{590},
pp. \bfpage{139}--\blpage{162}.
\bpublisher{Springer},
\blocation{Berlin}
(\byear{1992}).
\doiurl{10.1007/3-540-55425-4\_6}
\end{bchapter}
\endbibitem

%%% 42
\bibitem[\protect\citeauthoryear{McCune}{2003}]{otter}
\begin{botherref}
\oauthor{\bsnm{McCune}, \binits{W.}}:
{OTTER} 3.3 {Reference} {Manual}.
Technical Report ANL/MCS-TM-263,
Argonne National Laboratory
(2003).
\url{https://www.cs.unm.edu/~mccune/otter/Otter33.pdf}, accessed Jul 09, 2024
\end{botherref}
\endbibitem

%%% 43
\bibitem[\protect\citeauthoryear{McCune}{2005--2010}]{prover9}
\begin{botherref}
\oauthor{\bsnm{McCune}, \binits{W.}}:
Prover9 and {Mace4}.
\url{http://www.cs.unm.edu/~mccune/prover9}, accessed Jul 09, 2024
(2005--2010)
\end{botherref}
\endbibitem

%%% 44
\bibitem[\protect\citeauthoryear{McCune and Wos}{1992}]{mccune:wos:cd:1992}
\begin{bchapter}
\bauthor{\bsnm{McCune}, \binits{W.}},
\bauthor{\bsnm{Wos}, \binits{L.}}:
\bctitle{Experiments in automated deduction with condensed detachment}.
In: \beditor{\bsnm{Kapur}, \binits{D.}} (ed.)
\bbtitle{CADE-11}.
\bsertitle{LNCS (LNAI)},
vol. \bseriesno{607},
pp. \bfpage{209}--\blpage{223}.
\bpublisher{Springer},
\blocation{Berlin}
(\byear{1992}).
\doiurl{10.1007/3-540-55602-8\_167}
\end{bchapter}
\endbibitem

%%% 45
\bibitem[\protect\citeauthoryear{Megill}{1995}]{megill:1995}
\begin{barticle}
\bauthor{\bsnm{Megill}, \binits{N.D.}}:
\batitle{A finitely axiomatized formalization of predicate calculus with
  equality}.
\bjtitle{Notre Dame J. of Formal Logic}
\bvolume{36}(\bissue{3}),
\bfpage{435}--\blpage{453}
(\byear{1995}).
\doiurl{10.1305/ndjfl/1040149359}
\end{barticle}
\endbibitem

%%% 46
\bibitem[\protect\citeauthoryear{Megill and Wheeler}{2019}]{metamath:book}
\begin{bbook}
\bauthor{\bsnm{Megill}, \binits{N.}},
\bauthor{\bsnm{Wheeler}, \binits{D.A.}}:
\bbtitle{Metamath: A Computer Language for Mathematical Proofs},
\bedition{2}nd edn.
\bpublisher{lulu.com},
\blocation{Morrisville}
(\byear{2019}).
\bcomment{Online \url{https://us.metamath.org/downloads/metamath.pdf}}
\end{bbook}
\endbibitem

%%% 47
\bibitem[\protect\citeauthoryear{Meredith and
  Prior}{1963}]{meredith:notes:1963}
\begin{barticle}
\bauthor{\bsnm{Meredith}, \binits{C.A.}},
\bauthor{\bsnm{Prior}, \binits{A.N.}}:
\batitle{Notes on the axiomatics of the propositional calculus}.
\bjtitle{Notre Dame J. of Formal Logic}
\bvolume{4}(\bissue{3}),
\bfpage{171}--\blpage{187}
(\byear{1963}).
\doiurl{10.1305/ndjfl/1093957574}
\end{barticle}
\endbibitem

%%% 48
\bibitem[\protect\citeauthoryear{Meredith}{1977}]{meredith:memoriam:1977}
\begin{barticle}
\bauthor{\bsnm{Meredith}, \binits{D.}}:
\batitle{In memoriam: {C}arew {A}rthur {M}eredith (1904--1976).}
\bjtitle{Notre Dame J. of Formal Logic}
\bvolume{18}(\bissue{4}),
\bfpage{513}--\blpage{516}
(\byear{1977}).
\doiurl{10.1305/ndjfl/1093888116}
\end{barticle}
\endbibitem

%%% 49
\bibitem[\protect\citeauthoryear{{OEIS Foundation Inc.}}{2022}]{oeis}
\begin{botherref}
\oauthor{\bsnm{{OEIS Foundation Inc.}}}:
The {O}n-{L}ine {E}ncyclopedia of {I}nteger {S}equences.
\url{http://oeis.org}
(2022)
\end{botherref}
\endbibitem

%%% 50
\bibitem[\protect\citeauthoryear{Otten}{2010}]{leancop}
\begin{barticle}
\bauthor{\bsnm{Otten}, \binits{J.}}:
\batitle{Restricting backtracking in connection calculi}.
\bjtitle{AI Communications}
\bvolume{23}(\bissue{2-3}),
\bfpage{159}--\blpage{182}
(\byear{2010}).
\doiurl{10.3233/AIC-2010-0464}
\end{barticle}
\endbibitem

%%% 51
\bibitem[\protect\citeauthoryear{Pfenning}{1988}]{pfenning:single:1988}
\begin{bchapter}
\bauthor{\bsnm{Pfenning}, \binits{F.}}:
\bctitle{Single axioms in the implicational propositional calculus}.
In: \beditor{\bsnm{Lusk}, \binits{E.}},
\beditor{\bsnm{Overbeek}, \binits{R.}} (eds.)
\bbtitle{CADE-9}.
\bsertitle{LNCS (LNAI)},
vol. \bseriesno{310},
pp. \bfpage{710}--\blpage{713}.
\bpublisher{Springer},
\blocation{Berlin}
(\byear{1988}).
\doiurl{10.1007/BFb0012869}
\end{bchapter}
\endbibitem

%%% 52
\bibitem[\protect\citeauthoryear{Prior}{1956}]{prior:logicians:1956}
\begin{barticle}
\bauthor{\bsnm{Prior}, \binits{A.N.}}:
\batitle{Logicians at play; or {S}yll, {S}imp and {H}ilbert}.
\bjtitle{Australasian Journal of Philosophy}
\bvolume{34}(\bissue{3}),
\bfpage{182}--\blpage{192}
(\byear{1956}).
\doiurl{10.1080/00048405685200181}
\end{barticle}
\endbibitem

%%% 53
\bibitem[\protect\citeauthoryear{Prior}{1962}]{prior:formal:logic:1962}
\begin{bbook}
\bauthor{\bsnm{Prior}, \binits{A.N.}}:
\bbtitle{Formal Logic},
\bedition{2nd} edn.
\bpublisher{Clarendon Press},
\blocation{Oxford}
(\byear{1962}).
\doiurl{10.1093/acprof:oso/9780198241560.001.0001}
\end{bbook}
\endbibitem

%%% 54
\bibitem[\protect\citeauthoryear{Rawson et~al.}{2023}]{mrcwzzwb:lemmas:2023}
\begin{bchapter}
\bauthor{\bsnm{Rawson}, \binits{M.}},
\bauthor{\bsnm{Wernhard}, \binits{C.}},
\bauthor{\bsnm{Zombori}, \binits{Z.}},
\bauthor{\bsnm{Bibel}, \binits{W.}}:
\bctitle{Lemmas: Generation, selection, application}.
In: \beditor{\bsnm{Ramanayake}, \binits{R.}},
\beditor{\bsnm{Urban}, \binits{J.}} (eds.)
\bbtitle{TABLEAUX 2023}.
\bsertitle{LNCS (LNAI)},
vol. \bseriesno{14278},
pp. \bfpage{153}--\blpage{174}.
\bpublisher{Springer},
\blocation{Cham}
(\byear{2023}).
\doiurl{10.1007/978-3-031-43513-3_9}.
\bcomment{Extended version: \url{https://arxiv.org/abs/2303.05854}}
\end{bchapter}
\endbibitem

%%% 55
\bibitem[\protect\citeauthoryear{Rezus}{1982}]{rezus:1982}
\begin{barticle}
\bauthor{\bsnm{Rezus}, \binits{A.}}:
\batitle{On a theorem of {T}arski}.
\bjtitle{Libertas Mathematica}
\bvolume{2},
\bfpage{63}--\blpage{97}
(\byear{1982})
\end{barticle}
\endbibitem

%%% 56
\bibitem[\protect\citeauthoryear{Rezuş}{2020a}]{rezus:tarski:2019}
\begin{bchapter}
\bauthor{\bsnm{Rezuş}, \binits{A.}}:
\bctitle{{T}arski singleton bases: 1925-1932 (on an allegedly lost `method of
  proof' of {A}lfred {T}arski) (2019)}.
In: \bbtitle{Witness Theory -- Notes on $\lambda$-calculus and Logic}.
\bsertitle{Studies in Logic},
vol. \bseriesno{84},
pp. \bfpage{227}--\blpage{243}.
\bpublisher{College Publications},
\blocation{London}
(\byear{2020}).
\bcomment{Preprint (2019): \url{https://doi.org/10.13140/RG.2.2.10955.34081}}
\end{bchapter}
\endbibitem

%%% 57
\bibitem[\protect\citeauthoryear{Rezuş}{2020b}]{rezus:tc:2016}
\begin{bchapter}
\bauthor{\bsnm{Rezuş}, \binits{A.}}:
\bctitle{{T}arski's {C}laim thirty years later (2010)}.
In: \bbtitle{Witness Theory -- Notes on $\lambda$-calculus and Logic}.
\bsertitle{Studies in Logic},
vol. \bseriesno{84},
pp. \bfpage{217}--\blpage{225}.
\bpublisher{College Publications},
\blocation{London}
(\byear{2020}).
\bcomment{Preprint (2016):
  \url{http://www.equivalences.org/editions/proof-theory/ar-tc-20160512.pdf}}
\end{bchapter}
\endbibitem

%%% 58
\bibitem[\protect\citeauthoryear{Rezuş}{2020c}]{rezus:2020:witness}
\begin{bbook}
\bauthor{\bsnm{Rezuş}, \binits{A.}}:
\bbtitle{Witness Theory -- Notes on $\lambda$-calculus and Logic}.
\bsertitle{Studies in Logic},
vol. \bseriesno{84}.
\bpublisher{College Publications},
\blocation{London}
(\byear{2020})
\end{bbook}
\endbibitem

%%% 59
\bibitem[\protect\citeauthoryear{Robinson}{1965}]{robinson:1965}
\begin{barticle}
\bauthor{\bsnm{Robinson}, \binits{J.A.}}:
\batitle{A machine-oriented logic based on the resolution principle}.
\bjtitle{JACM}
\bvolume{12}(\bissue{1}),
\bfpage{23}--\blpage{41}
(\byear{1965})
\end{barticle}
\endbibitem

%%% 60
\bibitem[\protect\citeauthoryear{Schulz et~al.}{2019}]{eprover}
\begin{bchapter}
\bauthor{\bsnm{Schulz}, \binits{S.}},
\bauthor{\bsnm{Cruanes}, \binits{S.}},
\bauthor{\bsnm{Vukmirovi{\'c}}, \binits{P.}}:
\bctitle{Faster, higher, stronger: {E} 2.3}.
In: \beditor{\bsnm{Fontaine}, \binits{P.}} (ed.)
\bbtitle{CADE~27}.
\bsertitle{LNAI},
pp. \bfpage{495}--\blpage{507}.
\bpublisher{Springer},
\blocation{Cham}
(\byear{2019}).
\doiurl{10.1007/978-3-030-29436-6\_29}
\end{bchapter}
\endbibitem

%%% 61
\bibitem[\protect\citeauthoryear{Schumann}{1994}]{schumann:delta:1994}
\begin{bchapter}
\bauthor{\bsnm{Schumann}, \binits{J.M.P.}}:
\bctitle{{DELTA} -- {A} bottom-up preprocessor for top-down theorem provers}.
In: \bbtitle{CADE-12}.
\bsertitle{LNCS (LNAI)},
vol. \bseriesno{814},
pp. \bfpage{774}--\blpage{777}.
\bpublisher{Springer},
\blocation{Berlin}
(\byear{1994}).
\doiurl{10.1007/3-540-58156-1\_58}
\end{bchapter}
\endbibitem

%%% 62
\bibitem[\protect\citeauthoryear{Sobocínski}{1932}]{sobocinski:1932}
\begin{barticle}
\bauthor{\bsnm{Sobocínski}, \binits{B.}}:
\batitle{Z badań nad teorią dedukcji}.
\bjtitle{Przegląd Filozoficzny}
\bvolume{35},
\bfpage{171}--\blpage{193}
(\byear{1932}).
\bcomment{Excerpts translated into English and edited by A. Rezuş are
  published as \cite[p.~257--268 (Appendix: Boles{\l}aw Sobocínski 1932,
  §~1)]{rezus:2020:witness}}
\end{barticle}
\endbibitem

%%% 63
\bibitem[\protect\citeauthoryear{Stickel}{1988}]{pttp}
\begin{barticle}
\bauthor{\bsnm{Stickel}, \binits{M.E.}}:
\batitle{A {P}rolog technology theorem prover: implementation by an extended
  {P}rolog compiler}.
\bjtitle{J. Autom. Reasoning}
\bvolume{4}(\bissue{4}),
\bfpage{353}--\blpage{380}
(\byear{1988}).
\doiurl{10.1007/BF00297245}
\end{barticle}
\endbibitem

%%% 64
\bibitem[\protect\citeauthoryear{Sutcliffe}{2017}]{tptp}
\begin{barticle}
\bauthor{\bsnm{Sutcliffe}, \binits{G.}}:
\batitle{The {TPTP} problem library and associated infrastructure. {F}rom {CNF}
  to {TH0}, {TPTP v6.4.0}}.
\bjtitle{J. Autom. Reasoning}
\bvolume{59}(\bissue{4}),
\bfpage{483}--\blpage{502}
(\byear{2017}).
\doiurl{10.1007/s10817-017-9407-7}
\end{barticle}
\endbibitem

%%% 65
\bibitem[\protect\citeauthoryear{Sutcliffe and Suttner}{2001}]{tptp-rating}
\begin{barticle}
\bauthor{\bsnm{Sutcliffe}, \binits{G.}},
\bauthor{\bsnm{Suttner}, \binits{C.}}:
\batitle{{Evaluating General Purpose Automated Theorem Proving Systems}}.
\bjtitle{AI}
\bvolume{131}(\bissue{1--2}),
\bfpage{39}--\blpage{54}
(\byear{2001}).
\doiurl{10.1016/S0004-3702(01)00113-8}
\end{barticle}
\endbibitem

%%% 66
\bibitem[\protect\citeauthoryear{Thomas}{1970}]{thomas:final}
\begin{barticle}
\bauthor{\bsnm{Thomas}, \binits{I.}}:
\batitle{Final word on a shortest implicational axiom}.
\bjtitle{Notre Dame J. of Formal Logic}
\bvolume{11}(\bissue{1}),
\bfpage{16}
(\byear{1970})
\end{barticle}
\endbibitem

%%% 67
\bibitem[\protect\citeauthoryear{Ulrich}{2001}]{ulrich:legacy:2001}
\begin{barticle}
\bauthor{\bsnm{Ulrich}, \binits{D.}}:
\batitle{A legacy recalled and a tradition continued}.
\bjtitle{J. Autom. Reasoning}
\bvolume{27}(\bissue{2}),
\bfpage{97}--\blpage{122}
(\byear{2001}).
\doiurl{10.1023/A:1010683508225}
\end{barticle}
\endbibitem

%%% 68
\bibitem[\protect\citeauthoryear{Ulrich}{2007}]{ulrich:cd:web}
\begin{botherref}
\oauthor{\bsnm{Ulrich}, \binits{D.}}:
Sentential Calculi Pages.
Online: \url{https://web.ics.purdue.edu/~dulrich/Home-page.htm}, accessed Jul
  09, 2024
(2007)
\end{botherref}
\endbibitem

%%% 69
\bibitem[\protect\citeauthoryear{Ulrich}{2016}]{ulrich:single:2016}
\begin{bchapter}
\bauthor{\bsnm{Ulrich}, \binits{D.}}:
\bctitle{Single axioms and axiom-pairs for the implicational fragments of {R},
  {R-Mingle}, and some related systems}.
In: \beditor{\bsnm{Bimbó}, \binits{K.}} (ed.)
\bbtitle{J.~Michael Dunn on Information Based Logics}.
\bsertitle{Outstanding Contributions to Logic},
vol. \bseriesno{8},
pp. \bfpage{53}--\blpage{80}.
\bpublisher{Springer},
\blocation{Cham}
(\byear{2016}).
\doiurl{10.1007/978-3-319-29300-4\_4}
\end{bchapter}
\endbibitem

%%% 70
\bibitem[\protect\citeauthoryear{Veroff}{1996}]{veroff:hints:1996}
\begin{barticle}
\bauthor{\bsnm{Veroff}, \binits{R.}}:
\batitle{Using hints to increase the effectiveness of an automated reasoning
  program: Case studies}.
\bjtitle{J. Autom. Reasoning}
\bvolume{16}(\bissue{3}),
\bfpage{223}--\blpage{239}
(\byear{1996}).
\doiurl{10.1007/BF00252178}
\end{barticle}
\endbibitem

%%% 71
\bibitem[\protect\citeauthoryear{Veroff}{2001}]{veroff:shortest:2001}
\begin{barticle}
\bauthor{\bsnm{Veroff}, \binits{R.}}:
\batitle{Finding shortest proofs: An application of linked inference rules}.
\bjtitle{J. Autom. Reasoning}
\bvolume{27}(\bissue{2}),
\bfpage{123}--\blpage{139}
(\byear{2001}).
\doiurl{10.1023/A:1010635625063}
\end{barticle}
\endbibitem

%%% 72
\bibitem[\protect\citeauthoryear{Veroff}{2011}]{veroff:cd:2011}
\begin{botherref}
\oauthor{\bsnm{Veroff}, \binits{R.}}:
Challenge Problems with Condensed Detachment.
Online: \url{https://www.cs.unm.edu/~veroff/CD/}, accessed Jul 09, 2024
(2011)
\end{botherref}
\endbibitem

%%% 73
\bibitem[\protect\citeauthoryear{Walsh and
  Fitelson}{2021}]{fitelson:walsh:2021}
\begin{botherref}
\oauthor{\bsnm{Walsh}, \binits{M.}},
\oauthor{\bsnm{Fitelson}, \binits{B.}}:
Answers to some open questions of {U}lrich and {M}eredith
(2021).
Under review, preprint: \url{http://fitelson.org/walsh.pdf}, accessed Jul 09,
  2024
\end{botherref}
\endbibitem

%%% 74
\bibitem[\protect\citeauthoryear{Wernhard}{2022a}]{cw:cdtools:2022}
\begin{botherref}
\oauthor{\bsnm{Wernhard}, \binits{C.}}:
{CD Tools} -- {C}ondensed detachment and structure generating theorem proving
  (system description).
CoRR
\textbf{abs/2207.08453}
(2022).
\doiurl{10.48550/arXiv.2207.08453}
\end{botherref}
\endbibitem

%%% 75
\bibitem[\protect\citeauthoryear{Wernhard}{2022b}]{cw:ccs}
\begin{bchapter}
\bauthor{\bsnm{Wernhard}, \binits{C.}}:
\bctitle{Generating compressed combinatory proof structures -- an approach to
  automated first-order theorem proving}.
In: \beditor{\bsnm{Konev}, \binits{B.}},
\beditor{\bsnm{Schon}, \binits{C.}},
\beditor{\bsnm{Steen}, \binits{A.}} (eds.)
\bbtitle{PAAR~2022}.
\bsertitle{CEUR Workshop Proc.},
vol. \bseriesno{3201}.
\bpublisher{CEUR-WS.org},
\blocation{Aachen}
(\byear{2022}).
\bcomment{Preprint: \url{https://arxiv.org/abs/2209.12592}}
\end{bchapter}
\endbibitem

%%% 76
\bibitem[\protect\citeauthoryear{Wernhard}{2024}]{cw:sgcd}
\begin{bchapter}
\bauthor{\bsnm{Wernhard}, \binits{C.}}:
\bctitle{Structure-generating first-order theorem proving}.
In: \beditor{\bsnm{Otten}, \binits{J.}},
\beditor{\bsnm{Bibel}, \binits{W.}} (eds.)
\bbtitle{AReCCa 2023}.
\bsertitle{CEUR Workshop Proc.},
vol. \bseriesno{3613},
pp. \bfpage{64}--\blpage{83}.
\bpublisher{CEUR-WS.org},
\blocation{Aachen}
(\byear{2024}).
\burl{https://ceur-ws.org/Vol-3613/AReCCa2023\_paper5.pdf}
\end{bchapter}
\endbibitem

%%% 77
\bibitem[\protect\citeauthoryear{Wernhard and Bibel}{2021a}]{cwwb:lukas:2021}
\begin{bchapter}
\bauthor{\bsnm{Wernhard}, \binits{C.}},
\bauthor{\bsnm{Bibel}, \binits{W.}}:
\bctitle{Learning from {{\L}}ukasiewicz and {M}eredith: Investigations into
  proof structures}.
In: \beditor{\bsnm{Platzer}, \binits{A.}},
\beditor{\bsnm{Sutcliffe}, \binits{G.}} (eds.)
\bbtitle{CADE~28}.
\bsertitle{LNCS (LNAI)},
vol. \bseriesno{12699},
pp. \bfpage{58}--\blpage{75}.
\bpublisher{Springer},
\blocation{Cham}
(\byear{2021}).
\doiurl{10.1007/978-3-030-79876-5_4}
\end{bchapter}
\endbibitem

%%% 78
\bibitem[\protect\citeauthoryear{Wernhard and
  Bibel}{2021b}]{cwwb:lukas:2021:extended}
\begin{botherref}
\oauthor{\bsnm{Wernhard}, \binits{C.}},
\oauthor{\bsnm{Bibel}, \binits{W.}}:
Learning from {{\L}}ukasiewicz and {M}eredith: Investigations into proof
  structures (extended version).
CoRR
\textbf{abs/2104.13645}
(2021).
\doiurl{10.48550/arXiv.2104.13645}
\end{botherref}
\endbibitem

%%% 79
\bibitem[\protect\citeauthoryear{Wielemaker et~al.}{2012}]{swiprolog}
\begin{barticle}
\bauthor{\bsnm{Wielemaker}, \binits{J.}},
\bauthor{\bsnm{Schrijvers}, \binits{T.}},
\bauthor{\bsnm{Triska}, \binits{M.}},
\bauthor{\bsnm{Lager}, \binits{T.}}:
\batitle{{SWI-Prolog}}.
\bjtitle{Theory and Practice of Logic Programming}
\bvolume{12}(\bissue{1-2}),
\bfpage{67}--\blpage{96}
(\byear{2012}).
\doiurl{10.1017/S1471068411000494}
\end{barticle}
\endbibitem

%%% 80
\bibitem[\protect\citeauthoryear{Wos}{1991}]{wos:bledsoe:91}
\begin{bchapter}
\bauthor{\bsnm{Wos}, \binits{L.}}:
\bctitle{Automated reasoning and {B}ledsoe's dream for the field}.
In: \beditor{\bsnm{Boyer}, \binits{R.S.}} (ed.)
\bbtitle{Automated Reasoning: Essays in Honor of Woody Bledsoe}.
\bsertitle{Automated Reasoning Series},
pp. \bfpage{297}--\blpage{345}.
\bpublisher{Kluwer Academic Publishers},
\blocation{Dordrecht}
(\byear{1991}).
\doiurl{10.1007/978-94-011-3488-0_15}
\end{bchapter}
\endbibitem

%%% 81
\bibitem[\protect\citeauthoryear{Wos}{1995}]{wos:resonance:95}
\begin{barticle}
\bauthor{\bsnm{Wos}, \binits{L.}}:
\batitle{The resonance strategy}.
\bjtitle{Computers Math. Applic.}
\bvolume{29}(\bissue{2}),
\bfpage{133}--\blpage{178}
(\byear{1995}).
\doiurl{10.1016/0898-1221(94)00220-F}
\end{barticle}
\endbibitem

%%% 82
\bibitem[\protect\citeauthoryear{Wos}{1996}]{wos:combining:96}
\begin{barticle}
\bauthor{\bsnm{Wos}, \binits{L.}}:
\batitle{The power of combining resonance with heat}.
\bjtitle{J. Autom. Reasoning}
\bvolume{17}(\bissue{1}),
\bfpage{23}--\blpage{81}
(\byear{1996}).
\doiurl{10.1007/BF00247668}
\end{barticle}
\endbibitem

%%% 83
\bibitem[\protect\citeauthoryear{Wos}{2001}]{wos:meredith}
\begin{barticle}
\bauthor{\bsnm{Wos}, \binits{L.}}:
\batitle{Conquering the {M}eredith single axiom}.
\bjtitle{J. Autom. Reasoning}
\bvolume{27}(\bissue{2}),
\bfpage{175}--\blpage{199}
(\byear{2001}).
\doiurl{10.1023/A:1010691726881}
\end{barticle}
\endbibitem

%%% 84
\bibitem[\protect\citeauthoryear{Wos et~al.}{1990}]{wos:contributes:1990}
\begin{bchapter}
\bauthor{\bsnm{Wos}, \binits{L.}},
\bauthor{\bsnm{Winker}, \binits{S.}},
\bauthor{\bsnm{McCune}, \binits{W.}},
\bauthor{\bsnm{Overbeek}, \binits{R.}},
\bauthor{\bsnm{Lusk}, \binits{E.}},
\bauthor{\bsnm{Stevens}, \binits{R.}},
\bauthor{\bsnm{Butler}, \binits{R.}}:
\bctitle{Automated reasoning contributes to mathematics and logic}.
In: \beditor{\bsnm{Stickel}, \binits{M.E.}} (ed.)
\bbtitle{CADE-10},
pp. \bfpage{485}--\blpage{499}.
\bpublisher{Springer},
\blocation{Berlin}
(\byear{1990}).
\doiurl{10.1007/3-540-52885-7\_109}
\end{bchapter}
\endbibitem

\end{thebibliography}

\end{document}